\documentclass[
    twocolumn,
    superscriptaddress,
    longbibliography,
    nofootinbib,
    amsmath,
    amssymb,
    aps,
    prx
]{revtex4-2}

\usepackage[utf8]{inputenc}

\usepackage{amsmath,amsthm,amsfonts,amssymb}
\usepackage{mathtools}

\usepackage{verbatim}
\usepackage[colorlinks=true,linkcolor=magenta,urlcolor=blue,citecolor=blue,anchorcolor=green,pdfusetitle]{hyperref}
\usepackage{dsfont}
\usepackage{xcolor}
\usepackage{graphicx}
\usepackage{hhline}

\usepackage[shortlabels]{enumitem}
\usepackage{cleveref}
\Crefformat{enumi}{#2\textup{#1}#3} 

\usepackage{tikz}
\usepackage{pgfplots}
\pgfplotsset{compat=1.17}
\usetikzlibrary{external}
\usepackage{tikz-network}
\SetVertexStyle[Shape=circle,InnerSep=0pt,MinSize=1.5pt,FillColor=black]
\SetEdgeStyle[LineWidth=0.75pt]

\usepackage{braket}

\allowdisplaybreaks

\newtheorem{theorem}{Theorem}\Crefname{theorem}{Theorem}{Theorems}
\newtheorem{proposition}[theorem]{Proposition}\Crefname{proposition}{Proposition}{Propositions}
\newtheorem{lemma}[theorem]{Lemma}\Crefname{lemma}{Lemma}{Lemmas}
\newtheorem{corollary}[theorem]{Corollary}\Crefname{corollary}{Corollary}{Corollaries}

\newtheorem*{theorem*}{Theorem}
\newtheorem*{proposition*}{Proposition}
\newtheorem*{lemma*}{Lemma}
\newtheorem*{corollary*}{Corollary}

\theoremstyle{definition}
\newtheorem{definition}[theorem]{Definition}\Crefname{definition}{Definition}{Definitions}
\newtheorem*{definition*}{Definition}
\newtheorem*{remark}{Remark}
\newtheorem*{remark*}{Remark}
\newtheorem{example}[theorem]{Example}\Crefname{example}{Example}{Examples}
\newtheorem*{example*}{Example}

\newcommand{\mbf}{\mathbf}
\newcommand{\mbb}{\mathbb}
\newcommand{\mc}{\mathcal}

\newcommand{\tr}{\mathrm{Tr}}
\newcommand{\Tr}{\mathrm{Tr}}
\newcommand{\id}{\mathrm{id}}
\newcommand{\op}[2]{|#1\rangle\langle #2|}
\newcommand{\ip}[2]{\langle #1| #2 \rangle}
\newcommand{\1}{\mathds{1}}

\newcommand{\supp}{\mathrm{supp}}

\newcommand{\wt}{\widetilde}
\newcommand{\rk}{\text{rk}}

\newcommand{\EPR}[1]{{#1}\text{-}{\mathrm{EPR}}}
\newcommand{\POVM}{\mathrm{POVM}}

\newcommand{\bC}{\mathbb{C}}
\newcommand{\bR}{\mathbb{R}}

\newcommand{\bZ}{\mathbb{Z}}

\newcommand{\cC}{\mathcal{C}}
\newcommand{\cE}{\mathcal{E}}
\newcommand{\cG}{\mathcal{G}}

\newcommand{\cS}{\mathcal{S}}

\newcommand{\cN}{\mathcal{N}}
\newcommand{\cR}{\mathcal{R}}

\newcommand{\cU}{\mathcal{U}}
\newcommand{\q}{\quad}

\newcommand{\pbt}{{\mathsf{PBT}}}
\newcommand{\pbts}{\mathsf{PBT}_{\mathrm{cl}}}
\newcommand{\pbtcl}{\mathsf{PBT}_{\mathrm{cl}}}
\newcommand{\pbtq}{\mathsf{PBT}_{\mathrm{q}}}
\newcommand{\pbtds}{\mathsf{PBT}_{\max}}
\newcommand{\pbtmax}{\mathsf{PBT}_{\max}}

\DeclarePairedDelimiter{\ceil}{\lceil}{\rceil}
\DeclarePairedDelimiter{\floor}{\lfloor}{\rfloor}

\newcommand{\iso}{\cong}
\newcommand{\End}{\mathrm{End}}
\newcommand{\Mat}{\mathrm{Mat}}
\newcommand{\Span}{\mathrm{Span}}

\newcommand{\rank}{\mathrm{rank}}
\newcommand{\wh}[1]{\widehat{#1}}
\newcommand{\pure}[1]{|{#1}\rangle\hspace{-1pt}\langle{#1}|}
\newcommand{\bs}{\backslash}

\newcommand{\ds}{\displaystyle}
\newcommand{\F}{\mathsf{F}}

\definecolor{cb_blue}{HTML}{0077BB}
\definecolor{cb_teal}{HTML}{009988}
\definecolor{cb_red}{HTML}{CC3311}
\definecolor{cb_magenta}{HTML}{EE3377}
\definecolor{cb_orange}{HTML}{EE7733}
\definecolor{cb_purple}{HTML}{AA4499}

\renewcommand{\emph}[1]{\textit{#1}}

\begin{document}

\title{Asynchronous quantum computation through port-based teleportation}

\author{Chloe Kim}
\email{kim705@illinois.edu}
\thanks{corresponding author}
\affiliation{Department of Electrical and Computer Engineering, University of Illinois Urbana-Champaign}
\author{Eric Chitambar}
\email{echitamb@illinois.edu}
\affiliation{Department of Electrical and Computer Engineering, University of Illinois Urbana-Champaign}
\affiliation{Illinois Quantum Information Science and Technology Center, University of Illinois Urbana-Champaign}
\author{Felix Leditzky}
\email{leditzky@illinois.edu}
\affiliation{Illinois Quantum Information Science and Technology Center, University of Illinois Urbana-Champaign}
\affiliation{Department of Mathematics, University of Illinois Urbana-Champaign}

\date{\today}

\begin{abstract}
In standard quantum teleportation, the receiver must wait for a classical message from the sender before subsequently processing the transmitted quantum information. However, in port-based teleportation (PBT), this local processing can begin before the classical message is received, thereby allowing for asynchronous quantum information processing. Motivated by resource-theoretic considerations and practical applications, we propose different communication models that progressively allow for more powerful decoding strategies while still permitting asynchronous distributed quantum computation, a salient feature of standard PBT. Specifically, we consider PBT protocols augmented by free classical processing and/or different forms of quantum pre-processing, and we investigate the maximum achievable teleportation fidelities under such operations. Our analysis focuses specifically on the PBT power of isotropic states, bipartite graph states, and symmetrized EPR states, and we compute tight bounds on the optimal teleportation fidelities for such states. We finally show that, among this hierarchy of communication models consistent with asynchronous quantum information processing, the strongest resource theory is equally as powerful as any one-way teleportation protocol for surpassing the classical teleportation threshold. Thus, a bipartite state can break the one-way classical teleportation threshold if and only if it can be done using the trivial decoding map of discarding subsystems.
\end{abstract}

\maketitle

\section{Introduction}

Quantum teleportation is one of the most important tasks in quantum information theory \cite{Bennett_1993}.
Besides its potential use in future communication technologies, it illustrates the interplay among three different types of resources for information processing: quantum communication, classical communication, and shared entanglement \cite{Devetak_2004}.
By sharing entanglement and having access to a classical channel, two distant parties can simulate a quantum channel to exchange quantum information. In resource-theoretic terms, teleportation transforms a ``static'' resource, e.g.~an EPR state, into a ``dynamic'' resource, e.g.~a noiseless qubit channel, using local quantum operations and classical communication (LOCC) \cite{Bennett_1996, Plenio_2007, Chitambar_2019}.

        \begin{figure*}
\includegraphics[width=17cm]{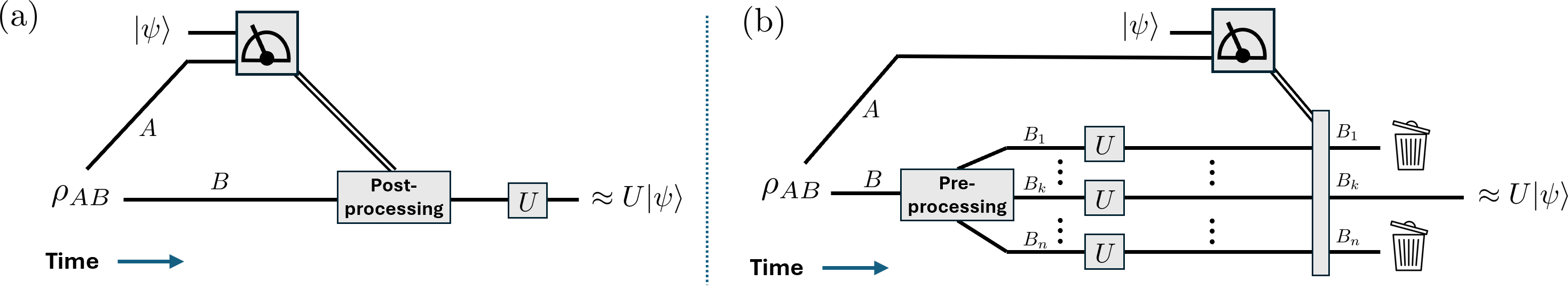}
\caption{Using teleportation in distributed computation. (a) In standard teleportation or general one-way LOCC protocols, any quantum circuit $U$ that Bob runs must be applied on his system \textit{after} he receives communication from Alice.
(b) In PBT, Bob can perform the circuit \textit{before} Alice communicates, or even before Alice receives her input $\ket{\psi}$; this allows for asynchronous distributed computation.
The fidelity between Bob's actual output and the ideal output $U\ket{\psi}$ will depend on the initial state $\rho_{AB}$.
The operational models of $\pbt$, $\pbtcl$, and $\pbtq$ do not allow Bob an arbitrary pre-processing step and assume the initial state already has a port structure $\rho_{AB^n}$.
On the other hand, $\pbtmax$ allows for pre-processing and enables the use of any initial state $\rho_{AB}$ in a PBT protocol.
}
\label{Fig:teleport_models}
        \end{figure*}   

However, static-to-dynamic quantum transformations can be achieved in a more restricted operational setting than general (two-way) LOCC.
First, if the target is a quantum channel from Alice to Bob, then it is sufficient to consider one-way LOCC, with classical communication from Alice to Bob.
Second, in the original teleportation protocol of Bennett \textit{et al.}~\cite{Bennett_1993}, Bob only needs to perform local Pauli gates; hence, one can narrow the operational model to one-way LOCC with Bob's local quantum operations being limited to single-qubit Pauli gates.
Remarkably, the model can be restricted even further while still being able to simulate a noiseless quantum channel with arbitrary precision.
The most extreme restriction is the port-based teleportation (PBT) protocol introduced by Ishizaka and Hiroshima \cite{Ishizaka_2008,Ishizaka_2009}.
In PBT, Bob's local operations only involve discarding quantum systems conditioned on Alice's classical message.
While more entanglement is needed for PBT compared to the original scheme, high-fidelity quantum communication is achievable with minimal quantum capabilities on Bob's part.

These minimal requirements for Bob's allowed operations in PBT enable the realization of asynchronous distributed quantum information processing (see Fig.~\ref{Fig:teleport_models}(b)), a capability that can be used to attack quantum position verification (QPV) protocols \cite{Beigi_2011}, simplify channel discrimination \cite{Pirandola_2019}, prepare states remotely \cite{muguruza2024}, store and retrieve quantum operations \cite{Bisio-2010a, Sedl_k_2019}, and invert the action of quantum circuits \cite{Quintino-2019a}. 
Indeed, PBT is critical in these various applications (and others referenced in Section \ref{sec:background}) precisely due to its simplified decoding step, in contrast to the standard teleportation protocol \cite{Bennett_1993}.
For example, attacks on QPV require adversaries Alice and Bob to perform distributed quantum computation with pre-shared entanglement and one round of simultaneous communication.
This communication constraint prohibits interactive local processing.  
However, using PBT, any time-ordered interactive attack can be replaced by Alice and Bob pre-processing their shared entanglement prior to exchanging classical communication \cite{Beigi_2011}.

Inspired by the operational capabilities of PBT, we develop a hierarchy of one-way quantum communication models and assess their performance.  
In the simplest model, denoted by $\pbt$, the sender (Alice) can perform arbitrary local quantum operations on her half of a shared entangled state $\rho_{AB_1B_2\dots B_n}$, while the receiver (Bob) can only discard subsystems $B_i$ based on a classical message from Alice.
We then generalize this model by considering two different motivations.  

First, from a resource-theoretic perspective one might be interested in understanding what type of distributed quantum information processing is possible if Bob's quantum capabilities are at the bare minimum of trivially discarding systems.
However in this case, it is reasonable to assume that Bob still has unconstrained \textit{classical} control, allowing him to prepare arbitrary classical systems (like random coins) in a teleportation protocol.
This suggests a natural operational theory that augments $\pbt$ by also allowing Bob to perform arbitrary classical operations, a model we denote by $\pbtcl$.

Second, more powerful communication models are physically relevant when considering the use of PBT in time-sensitive communication tasks, such as those referenced above.
This motivates a one-way communication model that builds on $\pbt$ by allowing Bob to perform arbitrary local pre-processing on his half of an initially shared state prior to receiving any communication from Alice.
We refer to this operational setting as $\pbtmax$ since it arguably contains the broadest range of physical operations that still enable asynchronous processing, the key feature of PBT.

Even though $\pbtmax$ is strictly more powerful than $\pbtcl$, they both characterize physically interesting situations.
$\pbtcl$ reflects a communication scenario in which the receiver essentially has no machinery available to process the held quantum information, while $\pbtmax$ reflects a scenario in which the receiver does not have enough \textit{time} to interactively process the held quantum information.
This is in sharp contrast to standard teleportation, in which the receiver must apply a Pauli correction conditioned on the classical message from the sender before passing the held information along for subsequent use in some protocol.

The main objective of this paper is to investigate how well a given bipartite entangled state can simulate a $d$-dimensional noiseless quantum channel using PBT and the other operational classes just introduced.
One fundamental question we address is when it becomes feasible for a bipartite state to generate a quantum channel whose entanglement fidelity exceeds the so-called classical threshold \cite{Horodecki_1999a}.
This threshold describes the largest fidelity to a noiseless $d$-dimensional quantum channel achievable using classical communication and no entanglement; numerically, its value is given by $\frac{1}{d}$.   
Having an entanglement fidelity larger than $\frac{1}{d}$ implies that the channel can transmit states with distillable entanglement \cite{Horodecki-1999b} as well as states capable of exhibiting multi-copy Bell nonlocality \cite{Cavalcanti-2013a}.  
In the foundational work of Ref.~\cite{Horodecki_1999a}, Horodecki \textit{et al.} showed that a bipartite state $\rho_{AB}$ can violate the classical threshold using some one-way LOCC teleportation protocol if and only if it can be done using the following steps:
\begin{enumerate}
    \item[(i)] $\rho_{AB}$ is transformed using one-way LOCC into a state $\sigma_{AB}$ whose singlet fraction exceeds $\frac{1}{d}$, i.e. $\bra{\Phi^+_d}\sigma_{AB}\ket{\Phi^+_d}>\frac{1}{d}$, where $\ket{\Phi^+_d}=\frac{1}{\sqrt{d}}\sum_{i=1}^d\ket{ii}$;
    \item[(ii)] Alice and Bob perform a random joint unitary rotation $U\otimes\overline{U}$ that ``twirls'' $\sigma_{AB}$ into an isotropic state,
    \begin{equation}
        \rho_{d,f}\coloneqq f\Phi^+_d+\frac{1-f}{d^2-1}(\1\otimes\1-\Phi^+_d),
    \end{equation}
    with $\Phi^+_d\coloneqq\op{\Phi^+_d}{\Phi^+_d}$ and $f=\bra{\Phi^+_d}\sigma_{AB}\ket{\Phi^+_d}>\frac{1}{d}$;
    \item[(iii)] The standard teleportation protocol from \cite{Bennett_1993} is performed on $\rho_{d,f}$.
\end{enumerate}
This result implies that certain entangled states cannot break the classical teleportation threshold using one-way LOCC.
In particular, since achieving a singlet fraction larger than $\frac{1}{d}$ in step (i) is also sufficient for the distillability of $\rho_{AB}$ \cite{Horodecki-1997a}, it follows that every non-distillable state cannot violate the classical teleportation threshold.
It remains a long-standing open problem to efficiently characterize the types of entangled states that admit a non-classical teleportation advantage and construct explicit protocols that realize this advantage. 
Partial progress towards resolving this question was recently reported in \cite{Chitambar_2023} which derived a necessary and sufficient criterion for states to admit non-classical teleportation advantage. 

One of the main results presented here sheds new light on this problem by showing that $\pbtmax$ is sufficient for surpassing the classical threshold.
More precisely, we show that a bipartite state $\rho_{AB}$ can violate the classical threshold using some one-way LOCC teleportation protocol if and only if it can be done using the following steps:
\begin{enumerate}
    \item[(i')] Bob performs some pre-processing map on his half of $\rho_{AB}$ that prepares $B$ as a collection of subsystems $B_1\dots B_n$;
    \item[(ii')] Alice makes a local measurement on her system $A$ and the teleported state;
    \item[(iii'')] Bob discards all but one of his subsystems.
\end{enumerate}
Steps (ii') and (iii') are just the PBT protocol, while (i') involves an optimization of Bob's state prior to measurement (see Fig.~\ref{Fig:teleport_models}(b)).
While the PBT protocol of Ishizaka and Hiroshima shows that perfect teleportation can be achieved with arbitrarily small infidelity by minimal post-processing on Bob's side (i.e., just discarding subsystems), our work establishes the complementary result: By including a pre-processing step, all feasible instances of non-classical teleportation are also possible by minimal post-processing on Bob's side.  

Overall, the core contributions of this work are summarized as follows.
\begin{itemize}
    \item We fully characterize the communication channels generated by $\pbt$ and the variants $\pbtcl$ and $\pbtmax$  described above.
    We also introduce another type of PBT model, denoted as $\pbtq$, which sits in between  $\pbtcl$ and $\pbtmax$.
    While $\pbtcl$ allows Bob to prepare arbitrary classical states, $\pbtq$ expands this capability to allow for the preparation of general quantum states in the pre-processing step.
    We show that the classes form a strict hierarchy  
    \begin{align} 
    \pbt \subsetneq \pbtcl \subsetneq \pbtq \subsetneq \pbtmax
    \end{align}
    in terms of the channels they can generate.
    \item For a fixed state $\rho_{AB^n}$, we derive bounds on the quantities 
    \begin{equation}
    \label{Eq:teleportation-fidelity}
    \F_{\Omega}(\rho_{AB^n})\coloneqq\max_{\Lambda \in \Omega(\rho_{AB^n},C_0\to B)}\F(\Lambda),
    \end{equation}
    where $\Omega(\rho_{AB^n}, C_0\to B)$ denotes the set of quantum channels $\Lambda\colon C_0\to B$ that are generated from $\rho_{AB^n}$ using the operational class $\Omega \in \{\pbt,\pbtcl,\pbtq\}$, and $\F(\Lambda)$ is the entanglement fidelity of the channel $\Lambda$ (defined in Eq.~\eqref{Eq:entanglement-fidelity}).
    We characterize both $\F_{\pbt}(\rho_{AB^n})$ and $\F_{\pbtcl}(\rho_{AB^n})$ as semi-definite programs (SDPs).

    \item We explicitly compute $\F_{\pbt}$, $\F_{\pbtcl}$, and $\F_{\pbtq}$ for the family of $d$-dimensional isotropic states $\rho_{d,f}$. In the computation of $\F_{\pbtq}(\rho_{d,f})$, we derive a closed-form expression for a convex-roof extended entanglement measure for $\rho_{d,f}$, which coincides with the geometric measure of entanglement of isotropic states computed in \cite{Wei_2003}.
    \item The PBT protocol for $\rho_{d,f}$ is generalized to any bipartite state $\rho_{AB}$ with suitable pre-processing.
    We show that $\F_{\max}(\rho_{AB},|C|)>\frac{1}{|C|}$ whenever there exists a map $\mc{E}\colon B\to C$ on Bob's side that yields a state which violates the reduction criterion \cite{Horodecki-1999b}. The existence of such a map was recently shown to be necessary and sufficient for the state $\rho_{AB}$ to break the $|C|$-dimensional classical teleportation threshold \cite{Chitambar_2023}.
    Thus, we establish that $\pbtmax$ is just as strong as general one-way LOCC for achieving non-classical teleportation.
    \item  We perform a detailed analysis of the PBT capabilities of bipartite graph states.
    The $\pbtq$ fidelity is determined by the number of isolated EPR pairs contained in the graph, and the $\pbt$ and $\pbtcl$ fidelities are bounded above and below in terms of them. If there is no isolated pair, then the entanglement fidelity of the resulting teleportation protocol cannot exceed the classical threshold.
    \item Specifically, we derive upper and lower bounds on $\F_{\pbtcl}$ of $n$-EPR pair states.  
    We find that already one EPR pair can surpass the classical threshold in the $\pbtcl$ model.  Equivalently, the four-qubit state $\Phi^+_{AB_1}\otimes \pure{0}_{B_2}\otimes \pure{1}_{B_3}$ can break the classical bound by standard PBT, improving upon the smallest previously known example of non-classical PBT using a six-qubit state \cite{Mozrzymas_2018}.
    \item Motivated by applications that utilize unitary covariant channels, we consider states built by symmetrizing $k$ EPR pairs over $n\geq k$ ports.   
    We explicitly compute the optimal PBT fidelity for those symmetrized states when $n\geq 2k-1$, and we find that one EPR pair symmetrized over three ports can generate a state that exceeds the classical bound.
\end{itemize}

While our study is motivated by the problem of asynchronous quantum computation, our technical interest lies in identifying the minimal conditions for non-classical advantages in this setting.  An analogy can be made with entanglement theory.  While most entanglement-based protocols utilize the EPR state $\ket{\Phi^+_2}=\frac{1}{\sqrt{2}}(\ket{00}+\ket{11})$, an important part of entanglement theory investigates the separability boundary, seeking to understand the fundamental structure of non-separable states and their operational advantages.  This work proceeds in the same spirit by looking at the basic structure of states that provide some operational resource in the setting of asynchronous quantum computation.

The paper proceeds by first providing a more detailed background and review of prior work on PBT in \Cref{sec:background}.
The operational frameworks for $\pbt$, $\pbtcl$, $\pbtq$ and $\pbtmax$ are presented in \Cref{sec:framework}, along with their associated entanglement fidelity measures.
Specific case studies of PBT are then carried out in \Cref{sec:examples}.
We begin in \Cref{sec:isotropic} by analyzing $d$-dimensional isotropic states shared on a single port.
We then consider in \Cref{sec:bipartite_graph} bipartite graph states as multi-port resources in PBT protocols.
A simple class of permutation-symmetric states constructed from EPR pairs is proposed and analyzed in \Cref{sec:symmetrized_EPR}.
Finally, in \Cref{sec:general_fidelity} we demonstrate how PBT can be used to surpass the classical teleportation threshold for general one-way-feasible states.
Concluding remarks are given in \Cref{Section:Conclusion}, and the Appendices \ref{Appendix-Cardinality bound}--\ref{sec:A:1-epr-Fpbt} contain the technical details and proofs of the main results.

\section{Background and Prior Work}\label{sec:background}

In the following we give a compressed but, to the best of our knowledge, comprehensive account of previous works studying port-based teleportation and its applications.
Port-based teleportation (PBT) was introduced and discussed by \textcite{Ishizaka_2008,Ishizaka_2009} as a modification of a related teleportation protocol implemented using linear optics \cite{Knill_2001,franson2002highfidelity}.
It was shown to give an asymptotic implementation of a universal programmable processor \cite{Ishizaka_2008}, thereby circumventing a known no-go theorem valid for finite resources \cite{Nielsen_1997}.
There are two well-studied variants of PBT: \emph{Deterministic} PBT implements an approximate teleportation protocol, whose quality is measured by the entanglement fidelity of the associated teleportation channel.
\emph{Probabilistic} PBT implements a perfect teleportation protocol, but may fail with a certain error probability.
Hybrid protocols combining both settings have recently been studied in \cite{Strelchuk_2023}.
In our work, we focus on the deterministic variant of PBT.

Converse bounds for the success probability of probabilistic PBT and the fidelity of deterministic PBT were derived in \cite{Pitalua-Garcia_2013,Ishizaka_2015,majenz2018entropy,kubicki2019resource}, and exact formulas in special cases were given in \cite{Wang_2016}. A representation-theoretic characterization of success probability and fidelity of the standard PBT protocol (operating on a collection of maximally entangled states) and fully optimized PBT protocols, yielding exact expressions for these quantities in the general case, were given in \cite{Studzinski_2017,Mozrzymas_2018}.
These results were rederived using different methods in \cite{Leditzky_2020}, where it was also explicitly shown that the pretty good measurement is optimal for standard and fully optimized PBT (see also \cite{Mozrzymas_2018}).
The asymptotic behavior of both probabilistic and deterministic PBT protocols was derived in \cite{Mozrzymas_2018,Christandl_2018}.
Further representation-theoretic characterizations of PBT protocols were studied in \cite{Mozrzymas2018simplified,Mozrzymas2024frobenius,studzinski2025irreducible}, and the recycling and degradation of ports is discussed in \cite{Strelchuk_2013,Studzinski2022squareroot,kopszak2025entanglementrecyclingprobabilisticportbased}.

Different variants of PBT have been introduced in the literature since the inception of the original protocol in \cite{Ishizaka_2008,Ishizaka_2009}.
These include multiport-based teleportation protocols \cite{Kopszak_2021,Mozrzymas2021optimalmultiport,Studzinski_2022}, continuous-variable PBT protocols \cite{boisselle2014pbtcv,Pereira_2023}, restricted or noisy PBT protocols \cite{anshu2017coherent,Jeong_2020,Collins2024teleportationofpost,muguruza2024,Eum_Kim_2024asymptotic}, and hybrid protocols combining probabilistic and deterministic PBT \cite{Strelchuk_2023}.
Recently, a series of works derived efficient quantum algorithms for performing PBT measurements using their representation-theoretic characterization \cite{Fei_2023,Wills_2023,grinko2023gelfandtsetlinbasispartiallytransposed,grinko2024efficientquantumcircuitsportbased,nguyen2023mixedschurtransformefficient}.

Apart from the aforementioned asymptotic implementation of universal programmable quantum processors \cite{Ishizaka_2008,Banchi2020convex}, PBT was used to demonstrate violations of Bell inequalities using communication complexity advantage \cite{Buhrman_2016}, and to implement attacks on position-based cryptography based on instantaneous non-local computation \cite{Beigi_2011,dolev2022nonlocalcomputationquantumcircuits,dolev2022holographyresourcenonlocalquantum,May2022complexity}.
PBT also finds applications in quantum channel discrimination and simulation \cite{Pirandola_2019,Pereira_2021}, various tasks to manipulate unitaries \cite{Sedl_k_2019,Quintino2019probabilistic,Quintino2022deterministic,yoshida2024}, storage and retrieval of quantum programs \cite{grosshans2024multicopyquantumstateteleportation} and measurements \cite{lewandowska2022storage}, entanglement distribution \cite{Kim_2024noisy}, and telecloning \cite{okada2025portbased}.  

All of these applications utilize in some way the possibility for asynchronous quantum information processing enabled by PBT.
Other scenarios of asynchronous distributed processing have been studied in the context of state discrimination \cite{Ballester-2008a, Carmeli-2018a}, measurement incompatibility \cite{Buscemi-2020a, Ji-2024a}, and Bell nonlocality \cite{Chitambar-2021a}.  
The work presented here contributes to this broader research program in understanding the role of interaction and time delays in quantum information processing. 

\section{Framework} \label{sec:framework}

\subsection{A resource-theoretic approach to PBT }

In this section, we present four different resource theories of quantum communication inspired by asynchronous processing and PBT.
The theories of $\pbt$, $\pbtcl$, and $\pbtq$ are distinguished from $\pbtmax$ in that they assume a \textit{port structure} for Bob's system.
A port structure is a decomposition of Bob's total system into $n$ subsystems, or ports, for some integer $n$.
We denote a port structure by $B^n=B_1B_2\dots B_n$, where (by embedding smaller subsystems into larger ones) it is assumed without loss of generality that each subsystem $B_i$ is a copy of the same $|B|$-dimensional system $B$.
We let $D(A\otimes B^n)$ denote the collection of density matrices $\rho_{AB^n}$ shared between Alice and Bob with a given port structure $B^n$.
Note that it may be possible for the same physical state to be realized on different port structures.
For example, $B_{\mathrm{I}}^{2n}$ might denote a port structure of $2n$ qubits while $B_{\mathrm{II}}^n$ denotes a port structure of $n$ ququart $(d=4)$ systems; yet, $D(A\otimes B_{\mathrm{I}}^{2n})$ and $D(A\otimes B_{\mathrm{II}}^n)$ represent the same collections of physical states.
When writing $\rho_{AB^n}$, it is implied that we are considering a density matrix with port structure $B^n$.
For $\pbtmax$, no port structure for Bob's system needs to be specified, and we consider arbitrary bipartite states $\rho_{AB}$.

\subsubsection{Standard PBT}

Every quantum resource theory is characterized by a set of free operations that reflect whatever actions can be performed given the experimental or algorithmic constraints \cite{Chitambar_2019}.
We first propose a quantum resource theory (QRT) for bipartite quantum communication with the following free operations:
\begin{enumerate}[(i)]
    \item Arbitrary local quantum operations by Alice; \label{it:free-op-1}
    \item Classical communication from Alice to Bob; \label{it:free-op-2}
    \item Relabeling and discarding of subsystems by Bob. \label{it:free-op-3}
\end{enumerate}
When limited to these three types of action, we refer to the QRT as \textit{standard PBT}. Under this QRT, every $\rho_{AB^n}\in D(A\otimes B^n)$ can be transformed into a quantum channel $\Lambda_{C_0\to B}$ by the following procedure.
Alice first performs a measurement on systems $AC_0$ described by a POVM $\{\Pi^{(i)}_{AC_0}\}_{i=1}^n$.
She announces the measurement outcome $i$ to Bob, who then discards all subsystems except $B_i$ and relabels the latter to $B$.
This generates the overall channel $\Lambda_{C_{0}\to B}$ whose action is given by 
\begin{equation}
\label{Eq:PBT_channel}
    \Lambda_{C_{0}\to B}(\tau_{C_0})=\sum_{i=1}^n\tr_{AC_0}[\Pi^{(i)}_{AC_0}(\rho^{(i)}_{AB}\otimes\tau_{C_0})],
\end{equation}
where $\rho^{(i)}_{AB}\coloneqq\tr_{B_i^c}(\rho_{AB^n})$ and $B_i^c \coloneqq B^n\setminus B_i$.  In Eq. \eqref{Eq:PBT_channel}, we implicitly assume that a relabeling map is performed after the partial trace that relabels $B_i$ to $B$.  In principle, we could also consider protocols in which Alice's message $i$ does not correspond to a single port, but rather to a subset of ports. 
For such a protocol, Bob discards all ports except those identified in this subset.
However, we do not consider maps of this form in this paper for the following reason.
If there were two non-disjoint subsets with, say, labels $i$ and $j$ in such a protocol, then it would generally not be possible to send the ports in parallel through the same quantum circuit prior to Alice's message, as in Fig.~\ref{Fig:teleport_models}(b), since Alice's message would specify which subset of qubits need to be grouped together in the circuit; so protocols of this form are less physically motivated.
On the other hand, if all the subsets of ports were disjoint in such a protocol, then we could just identify each collection of composite ports as a single port, and we would have a standard PBT protocol using the same physical state except with a different port structure; hence the resulting channel would still have the form of Eq.~\eqref{Eq:PBT_channel}.

For a given state $\rho_{AB^n}$, we let $\pbt(\rho_{AB^n},C_0\to B)$ denote the collection of all channels $\Lambda_{C_0\to B}$ that can be built from $\rho_{AB^n}$ in this way. 
Since the output dimension of every channel built in standard PBT using $\rho_{AB^n}$ will always be $|B|$, we can unambiguously write $\text{PBT}(\rho_{AB^n})$ to denote the collection of channels whose input and output dimensions are both $|B|$, i.e. $\pbt(\rho_{AB^n})\coloneqq\pbt(\rho_{AB^n},B'\to B)$.
Here and throughout the paper we use the convention that labels differing by primes, e.g. $B$, $B'$, and $B''$, always refer to systems of the same dimension, $|B|=|B'|=|B''|$.

\subsubsection{PBT with classical processing}

From a resource-theoretic perspective, it is also natural to regard classical ancilla registers and classical randomness as being free since they are relatively easy to produce.
We thus also consider a QRT with the following free operation in addition to \Cref{it:free-op-1}--\Cref{it:free-op-3}:
\begin{enumerate}[(i)]
\setcounter{enumi}{3}
    \item Preparation of classical states by Bob.
\end{enumerate}
We denote this resource-theoretic QRT by $\pbts$, and Bob's enhanced capability in $\pbts$ allows for a richer family of channels to be built.
Alice now has reason to consider measurements on systems $AC_0$ with $s\geq n$ number of outcomes. 
Like before, for outcomes $i=1,\dots ,n$ Bob discards all subsystems except $B_i$.
But for outcomes $i=n+1,\dots,s$ Bob is now free to discard all systems in $B$, prepare some $|B|$-dimensional classical state $\omega^{(i)}$ on an ancilla register, and then relabel this register as $B$.
By a classical state, we mean that each $\omega^{(i)}$ is diagonal in the fixed computational basis.
Every channel $\Lambda_{C_0\to B}$ obtained from $\rho_{AB^n}$ in $\pbts$ thus has the form
\begin{multline}
    \Lambda_{C_0\to B}(\tau_{C_0})= \sum_{i=1}^n\tr_{AC_0}[\Pi^{(i)}_{AC_0}(\rho^{(i)}_{AB}\otimes\tau_{C_0})]\\
    {}+\sum_{i=n+1}^s\tr[\Pi^{(i)}_{AC_0}(\rho_{A}\otimes\tau_{C_0})]\omega^{(i)}_{B}.
    \label{Eq:PBT*_channel}
\end{multline}
Observe that we can express each classical state as $\omega^{(i)}_{B}=\sum_{j=1}^{|B|}p(j|i)\op{j}{j}_{B}$ for conditional distributions $p(j|i)$.
Substituting this into the last term in Eq.~\eqref{Eq:PBT*_channel} gives
\begin{align}
\label{Eq:PBT*_channel2}
    \sum_{i=n+1}^s&\tr[\Pi^{(i)}_{AC_0}(\rho_A\otimes\tau_{C_0})]\omega^{(i)}_{B}\notag\\
    &=\sum_{i=n+1}^s\sum_{j=1}^{d}p(j|i)\tr[\Pi^{(i)}_{AC_0}(\rho_A\otimes\tau_{C_0})]\op{j}{j}_{B}\\
    &=\sum_{j=1}^{d}\tr[\tilde{\Pi}^{(j)}_{AC_0}(\rho_A\otimes\tau_{C_0})]\op{j}{j}_{B},
\end{align}
where $\tilde{\Pi}^{(j)}_{AC_0}=\sum_{i=n+1}^sp(j|i)\Pi^{(i)}_{AC_0}$.
It is straightforward to verify that the set $\{\Pi_{AC_0}^{(i)}\}_{i=1}^n\cup\{\tilde{\Pi}_{AC_0}^{(j)}\}_{j=1}^{d}$ constitutes a valid POVM.
We can thus restrict to $s=n+d$ and $\omega_B^{(n+i)}=\op{i}{i}$ in Eq.~\eqref{Eq:PBT*_channel} without loss of generality.
We let $\pbtcl(\rho_{AB^n}, C_0\to B)$ and $\pbtcl(\rho_{AB^n})\coloneqq\pbtcl(\rho_{AB^n}, B'\to B)$ be the family of channels obtained from $\rho_{AB^n}$ using $\pbts$.    

A careful comparison of Eqns.~\eqref{Eq:PBT_channel}--\eqref{Eq:PBT*_channel2} shows that 
\begin{align}
    &\pbts(\rho_{AB^n},C_0\to B)\notag\\
    &{}=\pbt(\rho_{AB^n}\otimes\op{1}{1}_{B_{1}'}\otimes \dots\otimes\op{d}{d}_{B_{d}'},C_0\to B),
    \label{Eq:PBT*=PBT}
\end{align}
where $B'_{1}\dots B'_{d}$ are local ancilla registers held by Bob.
Equation \eqref{Eq:PBT*=PBT} offers an appealing physical interpretation of channels in $\pbtcl(\rho_{AB^n},C_0\to B)$.
The simulation of any $\Lambda\in \pbts(\rho_{AB^n},C_0\to B)$ can be accomplished by Bob first preparing $d$ additional ports independent of $\rho_{AB^n}$, with the additional port $n+i$ being in classical state $\op{i}{i}$.  
Alice then performs a POVM on $AC_0$ having $n+d$ outcomes, and just as in standard PBT, upon receiving outcome $k$ Bob discards everything except the content of port $k$.
In this way, the channels of $\pbtcl$ can have the temporal structure of Fig.~\ref{Fig:teleport_models}(b).

\subsubsection{Maximal PBT}

For a third QRT, we consider another variation to the operational model.
Motivated by applications of PBT in which Bob needs to perform processing before receiving Alice's classical communication, we can increase Bob's operational capability while still allowing him to respect this temporal constraint.
Namely, we allow for:
\begin{enumerate}
    \item[(v)] Arbitrary local pre-processing by Bob prior to him receiving Alice's classical communication.
\end{enumerate}
Note that by the discussion at the end of the previous section, (v) supersedes (iv).
The key restriction in (v) is that Bob's choice of local pre-processing map cannot depend on Alice's message.
The only action that Bob can perform conditioned on Alice's message is a trivial discarding of subsystems.  We refer to this operational setting as \textit{maximal} PBT, denoted by $\pbtmax$, since it captures the most general physical processes feasible under the constraint of trivial post-processing for Bob.

Adding (v) as an allowed operation now enables any bipartite state $\rho_{AB}$ to be used in a PBT protocol.  In both standard PBT and $\pbtcl$, one requires that Alice and Bob initially share a state $\rho_{AB^n}$ with Bob's system already possessing a multi-port structure.  But in $\pbtmax$, a mapping of Bob's system $B$ into an arbitrary number of ports can be facilitated by a pre-processing map (see Fig.~\ref{Fig:teleport_models}(b)).  The collection of channels from input system $C_0$ to output system $C$ generated from $\rho_{AB}$ using $\pbtmax$ is given by
\begin{align}
    \pbtmax(\rho_{AB},C_0\to C)=\bigcup_{\sigma_{AC^n}}\pbt(\sigma_{AC^n},C_0\to C),
\end{align}
where the union is taken over all $\sigma_{AC^n}$ (and all integers $n$) such that 
\begin{align}
    \sigma_{AC^n}=\text{id}_A\otimes \mc{E}_{B\to C^n}(\rho_{AB})
\end{align}
for some completely-positive trace-preserving (CPTP) map $\mc{E}_{B\to C^n}$.

We also consider a subset of $\pbtmax(\rho_{AB},C_0\to C)$ that is similar to $\pbtcl$ in its operational form.
We refer to this as $\pbtq$, and it involves upgrading (iv) to a more general form:
\begin{itemize}
\item[(iv')] Preparation of arbitrary states by Bob.
\end{itemize}
In other words, $\pbtq$ consists of all PBT protocols for which Bob's pre-processing is restricted to just preparing additional ports in arbitrary states that are independent of Alice's input.
The resource theory of $\pbtq$ can address a scenario in which Bob's $n$ ports in $\rho_{AB^n}$ are stored in some quantum memory that does not allow for reliable quantum processing, and preparing independent auxiliary states is relatively easier to do.  
In other words, $\pbtq$ captures the setting of asynchronous quantum computation in which Bob is able to utilize fully quantum ancilla registers in addition to his ports that are entangled with Alice.
As we will see below, $\pbtq$ is mathematically much simpler to analyze than $\pbtmax$, and it is still sufficiently strong to enable non-classical teleportation fidelity using many states.
The hierarchy among the operational classes introduced here is
\begin{align}
    \pbt\subset\pbtcl\subset\pbtq\subset\pbtmax.
\end{align}

Letting $\pbtq(\rho_{AB^n},C_0\to B)$ and $\pbtq(\rho_{AB^n})\coloneqq\pbtq(\rho_{AB^n},B'\to B)$ denote the channels generated by $\rho_{AB^n}$ using $\pbtq$, the channels in $\pbtq(\rho_{AB^n}, C_0\to B)$ will have the same form as Eq.~\eqref{Eq:PBT*_channel} except with replacing the classical states $\omega_{B}^{(i)}$ by arbitrary quantum states $\sigma_{B}^{(i)}$.
It is also possible to place a cardinality bound on $s$.
\begin{proposition}\label{prop:cardinality-bound}
    Every channel in $\pbtq(\rho_{AB^n}, C_0\to B)$ can be generated by a POVM of Alice having $s\leq n+|C_0|^2(|A|^2+|B|^2)$ outcomes.
\end{proposition}

The proof is provided in \Cref{Appendix-Cardinality bound}.  Similar to Eq.~\eqref{Eq:PBT*=PBT}, we thus conclude that
\begin{align}
    &\pbtq(\rho_{AB^n},C_0\to B)\notag\\
    &=\bigcup_{\{\sigma^{(i)}\}_{i=1}^{\kappa}}\pbt(\rho_{AB^n}\otimes\sigma^{(1)}_{B'_1}\otimes \dots\otimes\sigma^{(\kappa)}_{B'_{\kappa}},C_0\to B),
    \label{Eq:PBT**=PBT}
\end{align}
where the union is taken over all sets $\{\sigma^{(i)}_{B'_i}\}_{i=1}^{\kappa}$ of quantum states that Bob could prepare in his lab prior to Alice's communication.
Consequently, we have that $\pbtq(\rho_{AB^n},C_0\to B)$ is a topologically closed family of quantum channels. Indeed, the right-hand side of Eq.~\eqref{Eq:PBT**=PBT} is the image of a compact set under a continuous map. The compact set consists of the POVMs with $n+\kappa$ outcomes together with the $\kappa$-tuples of states, so the image is compact and in particular closed.

\subsection{Entanglement Fidelity}

The performance of standard PBT is commonly measured by its entanglement fidelity.
In general, the \emph{entanglement fidelity} of a channel $\Lambda\colon B'\to B$ is defined as
\begin{equation}
\label{Eq:entanglement-fidelity}
    \F(\Lambda) \coloneqq \bra{\Phi^+}\id \otimes \Lambda(\Phi^{+}_{B''B'})\ket{\Phi^{+}}_{B''B},
\end{equation}
where $\ket{\Phi^+}_{B''B'}=\frac{1}{\sqrt{|B|}}\sum_{i=1}^{|B|}\ket{ii}$ is the canonical maximally entangled state and $\Phi^+=\op{\Phi^+}{\Phi^+}$.
Because of its linear relationship with the average input-to-output fidelity $\F_{\text{avg}}$ \cite{Horodecki_1999a},
\begin{align}
  \F_{\text{avg}} = \frac{|B|\F + 1}{|B|+1},
\end{align}
the entanglement fidelity of a quantum teleportation channel is often the main figure of merit when quantifying the performance of a teleportation protocol.
Ishizaka and Hiroshima showed that the entanglement fidelity of the standard PBT protocol with the shared resource state $\rho_{AB^n}$ and positive operator-valued measure (POVM) $\Pi=\{\Pi^{(i)}_{AB'}\}_{i=1}^n$ on Alice's side is equivalent to the success probability of using $\Pi$ to discriminate the ensemble of states $\{\rho^{(i)}_{AB} = \Tr_{B_i^c} (\rho_{AB^n})\}_{i=1}^n$ with uniform prior distribution \cite{Ishizaka_2008}.
More precisely, if $\Lambda$ is the teleportation channel generated by $\rho_{AB^n}$ and $\Pi$ is Alice's measurement, then
\begin{equation}
    \F(\Lambda)=\frac{1}{|B|^2}\sum_{i=1}^n\tr[\rho^{(i)}_{AB}\Pi^{(i)}_{AB}].
\end{equation}
Optimizing over all POVMs $\Pi$ gives the ultimate performance of $\rho_{AB^n}$ under PBT,
\begin{align}
    \F_{\pbt}(\rho_{AB^n})&\coloneqq\max\{\F(\Lambda)\;|\;\Lambda\in\pbt(\rho_{AB^n})\} \label{eq:F-PBT}\\
    &=\max_{\{\Pi^{(i)}_{AB}\}_{i=1}^n} \frac{1}{|B|^2}\sum_{i=1}^n\tr[\rho^{(i)}_{AB}\Pi^{(i)}_{AB}].
\end{align}
Following the same reasoning (see also \cite{Chitambar_2023}), we can define 
\begin{widetext}
\begin{align}
    \F_{\pbtcl}(\rho_{AB^n})
    &\coloneqq\max\{\F(\Lambda)\;|\;\Lambda\in\pbtcl(\rho_{AB^n})\} \label{eq:F-PBTcl}\\
    &=\frac{1}{|B|^2}\max_{\{\Pi^{(i)}_{AB}\}_{i=1}^{n+|B|}}\bigg\{\sum_{i=1}^n\tr[\Pi^{(i)}_{AB} \rho^{(i)}_{AB}] + \sum_{i=1}^{|B|}\tr[\Pi^{(n+i)}_{AB}(\rho_{A}\otimes\op{i}{i}_{B})]\bigg\} \label{eq:Fpbtcl-SDP-primal}\\    
    \F_{\pbtq}(\rho_{AB^n})&\coloneqq\max\{\F(\Lambda)\;|\;\Lambda\in\pbtq(\rho_{AB^n})\} \\
    &=\frac{1}{|B|^2}\max_{\{\Pi^{(i)}_{AB}\}_{i=1}^{n+\kappa'},\;\{\sigma^{(i)}_B\}_{i=1}^{\kappa'}}\bigg\{\sum_{i=1}^n\tr[\Pi^{(i)}_{AB} \rho^{(i)}_{AB}] + \sum_{i=1}^{\kappa'}\tr[\Pi^{(n+i)}_{AB}(\rho_{A}\otimes\sigma^{(i)}_{B})]\bigg\},
\end{align}
\end{widetext}
where the value of $\kappa'=|A|^2|B|^2$ follows again by another application of Carath\'{e}odory's theorem \cite{Davies-1978a}, applied to boundary points of the convex hull, since a maximizer lies on its boundary.
We observe that both $\F_{\pbt}(\rho_{AB^n})$ and $\F_{\pbtcl}(\rho_{AB^n})$ are computable as a semi-definite program (SDP), and we exploit this structure when making explicit calculations in \Cref{sec:examples}.

For $\pbtmax$, the channel dimension is not limited to dimension of Bob's system in the shared state $\rho_{AB}$ due to his pre-processing capability.
Hence, the ultimate entanglement fidelity achievable by $\rho_{AB}$ is dimension dependent, and we define it as
\begin{multline}
    \F_{\pbtmax} (\rho_{AB},|C|)\\
    =\sup\{\F(\Lambda)\;|\;\Lambda\in\pbtmax(\rho_{AB},C'\to C)\}.
\end{multline}

We close this section by observing general bounds on  the teleportation entanglement fidelities under $\pbtcl$ and $\pbtq$.
\begin{proposition}
\label{Prop:PBT-bounds}
For any state $\rho_{AB^n}$ and entanglement fidelity $\F=\F_{\pbtcl}(\rho_{AB^n})$ or $\F_{\pbtq}(\rho_{AB^n})$, we have
\begin{align}
\label{Eq:PBT-bounds}
    \frac{1}{|B|}\leq \F\leq \frac{1}{|B|}+\alpha
\end{align}
where
\begin{align}
    \alpha=\max_{\substack{\{\Pi^{(i)}_{AB}\}_{i=1}^n\\ \sum_{i=1}^n\Pi^{(i)}_{AB}\leq\1_{AB}}}\frac{1}{|B|^2}\sum_{i=1}^n\tr[\Pi^{(i)}_{AB}(\rho^{(i)}_{AB}-\rho_{A}\otimes\1_{B})]. \label{eq:alpha}
\end{align}
\end{proposition}
\begin{proof}
Since $\sigma^{(i)}_{B}\leq \1$, we have
\begin{align}
     &\sum_{i=1}^{|B|}\tr[\Pi^{(n+i)}_{AB}(\rho_{A}\otimes\sigma^{(i)}_{B})]\notag\\
     &\qquad {} \leq\sum_{i=n+1}^{n+|B|}\tr[\Pi^{(i)}_{AB}(\rho_{A}\otimes\1_B)] \\
     &\qquad {} =|B|-\sum_{i=1}^n\tr[\Pi^{(i)}_{AB}(\rho_{A}\otimes\1_B)], 
\end{align}
where we have used the completion relation of the POVM.
Substituting into the definition of $\F_{\pbtq}(\rho_{AB^n})$ proves the upper bound in Eq.~\eqref{Eq:PBT-bounds}.

For the lower bound, observe that $\pbtcl$ allows for the optimal classical protocol to be performed.
That is, Alice is able to simply measure her input in the computational basis, announce the outcome to Bob, and then Bob prepares in system $B$ the corresponding computational basis state; this protocol attains a teleportation fidelity of $|B|^{-1}$, thereby establishing the lower bound in Eq.~\eqref{Eq:PBT-bounds}.
The lower bound of $|B|^{-1}$ can also be seen explicitly in the definition of $\F_{\pbtcl}$ by setting $\Pi^{(i)}_{AB}=0$ for $i=1,\dots, n$ and $\Pi_{AB}^{(i)}=\1_A\otimes\op{i-n}{i-n}_B$ for $i=n+1,\dots,n+|B|$.
\end{proof}

Each PBT class $\Omega\in\{\pbt,\pbtcl,\pbtq,\pbtmax\}$  induces an ordering on the set of density operators through the teleportation fidelity as
\begin{align}
  \rho \preceq_{\Omega} \rho' \!\!\iff \F_{\Omega}(\rho_{AB}) \leq \F_{\Omega}(\rho').
\end{align}
We say $\rho_{AB}$ and $\rho_{AB}'$ are \emph{$\Omega$-equivalent}, or $\rho_{AB} \sim_{\Omega} \rho_{AB}'$, if and only if $\F_{\Omega}(\rho_{AB}) = \F_{\Omega}(\rho_{AB}')$.
From the definition of $\F_{\Omega}(\rho_{AB})$, it is immediately clear that 
\begin{equation}
    \label{Eq:fidelity-monotone}
    \F_{\Omega}(\rho_{AB})\geq\F_{\Omega}(\mc{E}_A\otimes\text{id}_B(\rho_{AB}))
\end{equation}
for every CPTP map on Alice's side.
This implies that $\rho_{AB} \sim_{\Omega} \rho_{AB}'$ whenever $\rho_{AB}$ and $\rho_{AB}'$ are related by a local isometry on Alice's side.
Since all purifications of Bob's reduced state $\rho_B$ are related by some isometry on Alice's system, we arrive at the following observation.
\begin{proposition}
    \label{prop:fidelity-pure-state}
    For any pure state $\ket{\psi}_{AB}$ and $\Omega\in\{\pbt,\pbtcl,\pbtq\}$, the PBT teleportation fidelity $\F_\Omega(\ket{\psi}_{AB})$ depends only on Bob's marginal state $\rho_B=\tr_A\op{\psi}{\psi}_{AB}$.
    For $\pbtmax$, the fidelity $\F_{\pbtmax}(\ket{\psi}_{AB})$ depends only on the spectrum of $\rho_B$.
\end{proposition}
\noindent The reason why we have the stronger statement for $\pbtmax$ is because Bob can also perform arbitrary local unitaries on $\ket{\psi}_{AB}$ as part of his pre-processing map.
Hence, 
\begin{equation}
\label{Eq:PBT-max-pure-LU}
    \ket{\psi}_{AB} \sim_{\text{LU}} \ket{\psi'}_{AB} \Longrightarrow\ket{\psi}_{AB} \sim_{\pbtmax} \ket{\psi'}_{AB} ,
\end{equation}
where $\sim_{\text{LU}}$ denotes a local unitary equivalence of states.

\subsubsection{The classical threshold}

The classical threshold of the entanglement fidelity for teleportation channels $\Lambda\colon B\to B$ is $|B|^{-1}$.
As shown in Proposition \ref{Prop:PBT-bounds}, this is the fidelity that can be attained by Alice simply sending a classical message to Bob after measuring her quantum input in a specified basis, and Bob preparing the corresponding basis state.
Unlike $\pbt$, the enhanced capability of $\pbtcl$ allows it to at least attain this threshold for every state $\rho_{AB^n}$.
The next natural question is to understand which states $\rho_{AB^n}$ can surpass this threshold and thereby offer a truly non-classical communication advantage in teleportation.
Interestingly, Proposition \ref{Prop:PBT-bounds} implies that not all entangled states offer such an advantage.

\begin{corollary}
\label{Cor:reduction}
    If $\rho_{AB^n}$ satisfies $\rho^{(i)}_{AB}\leq\rho_A\otimes\1$ for all $i=1,\dots,n$ (where $\rho^{(i)}_{AB}=\tr_{B_i^c}(\rho_{AB^n}))$, then 
    \begin{align} 
    	\F_{\pbtcl}(\rho_{AB^n})=\F_{\pbtq}(\rho_{AB^n})=\frac{1}{|B|}.
    \end{align}    
\end{corollary}
\begin{proof}
This fact follows from the definition of $\alpha$ in \eqref{eq:alpha}: if $\rho^{(i)}_{AB}\leq\rho_A\otimes\1$ $\forall i$, then $\alpha$ will be non-positive. Indeed, let $\{\Pi^{(i)}_{AB}\}_{i=1}^n$ be feasible in \eqref{eq:alpha}. Each $\Pi^{(i)}_{AB}$ is positive semidefinite and $\rho^{(i)}_{AB}-\rho_A\otimes\1_B\leq0$, so $\tr[\Pi^{(i)}_{AB}(\rho^{(i)}_{AB}-\rho_A\otimes\1_B)]\leq0$ for every $i$. Hence $\alpha\leq0$, and Eq.~\eqref{Eq:PBT-bounds} gives $\F=\frac{1}{|B|}$.
\end{proof}

We can state the conditions of this corollary in a more familiar way. Recall that a bipartite density matrix $\rho_{AB}$ is said to satisfy the reduction criterion \cite{Horodecki-1999b} if 
\begin{align}
    \rho_A\otimes\1-\rho_{AB}\geq 0.
\end{align}
It is known that all separable states respect this inequality, but the converse is not true.
Furthermore, a violation of the reduction criterion ensures that $\rho_{AB}$ is distillable \cite{Horodecki-1999b}.
Corollary \ref{Cor:reduction} says that $\rho_{AB^n}$ cannot break the classical threshold by $\pbtq$ if each of its marginal states $\rho^{(i)}_{AB}$ satisfies the reduction criterion.
For general one-way LOCC protocols, the classical bound holds exactly when every state $(\id_A\otimes\mc{E}_{B\to C})(\rho_{AB})$ obtained by a local channel on Bob's side satisfies the reduction criterion \cite{Chitambar_2023}.
Corollary \ref{Cor:reduction} presents an analogous result that is specialized to the setting of $\pbtq$.
Satisfying the reduction criterion is therefore not enough once Bob may process his system before Alice's message.
For example, let $\ket{\Phi_2^+}=\frac{1}{\sqrt{2}}(\ket{00}+\ket{11})\in\mbb{C}^3\otimes\mbb{C}^3$ and
\begin{equation}
    \rho_{AB}=\frac{1}{2}\Phi_2^+ +\frac{1}{4}\big(\op{0}{0}+\op{1}{1}\big)_A\otimes\op{2}{2}_B .
    \label{eq:reduction-example}
\end{equation}
Then $\rho_A\otimes\1-\rho_{AB}\geq 0$, so $\F_{\pbtq}(\rho_{AB})=\frac{1}{3}$ by Corollary~\ref{Cor:reduction}.
If Bob first applies the channel that keeps $\Span\{\ket{0},\ket{1}\}$ and maps $\ket{2}$ to $\ket{0}$, the resulting state has singlet fraction $\frac{5}{12}>\frac{1}{3}$ and violates the reduction criterion.
Corollary~\ref{cor:one-way} then gives $\F_{\pbtmax}(\rho_{AB},3)>\frac{1}{3}$, so this state separates $\pbtq$ from $\pbtmax$.%
\footnote{This example was found with the assistance of a generative AI model (Claude Opus 5.5) and verified by the authors.}

\section{Examples of the PBT teleportation fidelities}

\label{sec:examples}

In the remainder of this paper we compute the PBT teleportation entanglement fidelities for different quantum states.
In most cases we exploit symmetries in the given states to make the calculation feasible.
An underlying goal is to identify which states are capable of breaking the classical threshold in the different PBT resource theories.

Note that for any bipartite state $\rho_{AB}$ shared on a single port of dimension $|B|$, the definition of $\F_{\pbt}$ immediately implies that
\begin{equation}
    \F_{\pbt}(\rho_{AB})=\frac{1}{|B|^2}.
\end{equation}
Hence, it might be tempting to conclude that $\rho_{AB}$ has no PBT resourcefulness when using just a single port.
However, as we will see in the first example of isotropic states, it is possible to ``activate'' the PBT resourcefulness of some states on a single port using $\pbtcl$ and $\pbtq$.
In \Cref{sec:general_fidelity} we will see that it is, in fact, possible to activate all one-way teleportation resource states shared on a single port using $\pbtmax$.

\subsection{Isotropic states}

\label{sec:isotropic}

The $d$-dimensional isotropic states \cite{Horodecki-1999b} constitute a one-parameter family of states defined as
\begin{align}
\label{Eq:isotropic-state}
    \rho_{d,f}=f\Phi^+_d+\frac{1-f}{d^2-1}\left(\1_{AB}-\Phi^+_d\right),
\end{align}
where $\Phi^+_d=\op{\Phi^+_d}{\Phi^+_d}$ with $\ket{\Phi_d^+}=\frac{1}{\sqrt{d}}\sum_{i=1}^d\ket{ii}$, and $f=\tr[\rho_{d,f}\Phi^+_d]\in[0,1]$ is the singlet fraction.
The isotropic state $\rho_{d,f}$ is entangled if and only if $f>1/d$ \cite{Horodecki-1999b}.
All members of this family share the same symmetries as $\ket{\Phi^+_d}$.
To make this explicit,  let $\text{U}(d)$ denote the group of unitaries on $\mbb{C}^d$ and define the twirling map $\mc{T}_{U\otimes\overline{U}}(X)=\int_{U\in\text{U}(d)}dU(U\otimes \overline{U})X(U\otimes \overline{U})^\dagger$, where $dU$ denotes the Haar measure on $\text{U}(d)$.
The action of $\mc{T}_{U\otimes\overline{U}}$ on an arbitrary operator $X$ transforms it into isotropic form
\begin{align}
    \mc{T}_{U\otimes\overline{U}}(X)&=\tr[\Phi^+_d X]\Phi^+_d+\frac{\tr X-\tr[\Phi^+_d X]}{d^2-1}\left(\1_{AB}-\Phi^+_d\right). 
\end{align}
We note that the map $\mc{T}_{U\otimes\overline{U}}$ can be equivalently implemented by sampling over a finite set of unitaries: a unitary $2$-design reproduces the twirl of every operator with uniform weights \cite{GrossAudenaertEisert2007,Dankert2009}, whereas a $1$-design is not enough.  

In the following, we assume that a single copy of $\rho_{d,f}$ is shared between Alice and Bob, and explicitly compute both $\F_{\pbtcl}(\rho_{d,f})$ and $\F_{\pbtq}(\rho_{d,f})$ in this setting.
While the fidelities go to zero as $d$ grows large, we are interested in deciding whether it goes to zero strictly slower than the classical threshold of $1/d$.
We find that for both $\pbtcl(\rho_{d,f})$ and $\pbtq(\rho_{d,f})$ this is indeed the case, thereby demonstrating the positive utility of isotropic states in finite-dimensional PBT protocols with a single port.
As for $\pbtmax(\rho_{d,f})$, since $\ket{\Phi_d^+}$ is locally equivalent to $n\approx \log d$ EPR pairs, Alice and Bob can rotate into $\ket{\Phi_2^+}^{\otimes n}$ and apply the original PBT protocol of Ishizaka and Hiroshima \cite{Ishizaka_2008}.
This gives $\F_{\pbtmax}(\rho_{d,f},2)\geq f\,\F_{\pbt}(\Phi_2^{+\otimes n})$, and the right-hand side tends to $f$ as $d\to\infty$.

\subsubsection{$\pbtcl$ fidelity}

\label{sec:isotropic-cl}

We start by considering the optimal teleportation fidelity of an isotropic state $\rho_{d,f}$ in the $\pbtcl$ framework.

\begin{theorem}\label{thm:iso-pbtcl}
    The classical-processing PBT fidelity of the isotropic state $\rho_{d,f}$ is given by
	\begin{align}
		\F_{\pbtcl}(\rho_{d,f}) = \begin{cases}\displaystyle
			\frac{1}{d} + \frac{(1-df)^2}{d^2(d^2f-1)} & \text{\upshape if $f > \frac{1}{d}$},\\[1em]
			\displaystyle\frac{1}{d} & \text{\upshape otherwise.}
		\end{cases} 
	\end{align}
	For $f>\frac{1}{d}$, this fidelity exceeds the classical fidelity, $\F_{\pbtcl} > \frac{1}{d}$.
\end{theorem}
\noindent Since the isotropic state $\rho_{d,f}$ is separable for $f\leq 1/d$, \Cref{Cor:reduction} immediately gives $\F_{\pbtcl}(\rho_{d,f})=\frac{1}{d}$ when $f\leq 1/d$, since separable states satisfy the reduction criterion.
The argument for the case $f > 1/d$ is presented in \Cref{app:iso-pbtcl-proof}.
The basic idea is to characterize the teleportation fidelity $\F_{\pbtcl}$ in dual formulation as
\begin{align}
	\min \left\{ \tr(K): K \geq \tfrac{1}{d^2} \sigma^{(i)} \text{ for } i=1,\dots,d+1 \right\},\label{eq:dual2}
\end{align}
where $\sigma^{(i)} = \frac{1}{d}\1 \otimes |i\rangle\langle i|$ for $i=1,\dots,d$, and $\sigma^{(d+1)} = \rho_{d,f}$.  
We then identify a dual optimal solution to this problem in the appendix.

\subsubsection{$\pbtq$ fidelity}

\label{sec:isotropic-q}

Turning to $\pbtq$, we establish the following expression for the entanglement fidelity.
\begin{theorem}
\label{thm:iso-fidelity-q}
The optimal PBT fidelity with quantum state preparation of the isotropic state $\rho_{d,f}$ is equal to
    \begin{equation}
         \label{Eq:PBT**-EPR-text}
         \F_{\pbtq}(\rho_{d,f})=
         \begin{cases}\displaystyle \frac{1}{d}+\frac{f}{d^2}-\frac{1}{d^3} & \text{\upshape if } f\geq 1/d\\[1em]
         \displaystyle\frac{1}{d} & \text{\upshape otherwise.}
         \end{cases}
    \end{equation}
\end{theorem}

\begin{remark}
It follows from Theorems \ref{thm:iso-pbtcl} and \ref{thm:iso-fidelity-q} that there is a strict separation between $\F_{\pbts}(\rho_{d,f})$ and $\F_{\pbtq}(\rho_{d,f})$ for all $f>1/d$.
\end{remark}

The full proof of \Cref{thm:iso-fidelity-q} is given in \Cref{app:iso-q-optimization}. Here, we note that a key step in the argument is computing a convex-roof extended entanglement monotone evaluated on isotropic states. The underlying measure defined on a bipartite pure state $\ket{\phi}$ is the $(d-1)$-Ky-Fan entanglement measure, defined as
\begin{align} 
\mathsf{E}_{d-1}(\phi)=\sum_{k=2}^{d}\lambda_k^\downarrow(\phi),
\end{align}
where $\ket{\phi} = \sum_{k=1}^d \sqrt{\lambda_k^\downarrow(\phi)} \ket{\alpha_k}\ket{\beta_k}$ is the Schmidt decomposition of $\ket{\phi}$ with its Schmidt coefficients $\lambda_1^\downarrow(\phi)\geq \dots \geq \lambda_d^\downarrow(\phi)$ enumerated in non-increasing order \cite{Vidal-1999a}.  The convex-roof extension of this entanglement measure for mixed states is defined as
\begin{align}
\hat{\mathsf{E}}_{d-1}(\rho)=\min_{\sum_i p_i\psi_i=\rho}\sum\nolimits_i p_i\mathsf{E}_{d-1}(\psi_i), \label{eq:Ky-Fan-measure-mixed}
\end{align}
where the minimization is taken over all pure state ensembles $\{\ket{\psi_i}, p_i\}$ that generate $\rho$ \cite{Vidal-2000a}.  The entanglement of formation of $\rho_{d,f}$ (which is also a convex-roof extended entanglement measure) has been computed in Ref.~\cite{Terhal-2000a}, so there is hope that we can also compute $\hat{\mathsf{E}}_{d-1}(\rho_{d,f})$.
In fact, we perform such a computation in \Cref{app:E-isotropic} and obtain the following result.
\begin{lemma}
\label{lem:E-isotropic}
The entanglement measure $\hat{\mathsf{E}}_{d-1}$ defined in \eqref{eq:Ky-Fan-measure-mixed} satisfies the following on isotropic states $\rho_{d,f}$:
\begin{multline}
    1-\hat{\mathsf{E}}_{d-1}(\rho_{d,f})\\
    =\begin{cases}
        1 & \text{\upshape for }  f\leq \frac{1}{d}\\
        1-f+\frac{-1+2(f+\sqrt{(d-1)(1-f)f})}{d} &\text{\upshape for } f>\frac{1}{d}.
    \end{cases}
\end{multline}    
\end{lemma}

\subsection{Bipartite graph states}

\label{sec:bipartite_graph}

Multi-qubit graph states are an appealing family of quantum states to consider in PBT protocols given their general utility for quantum information processing \cite{Raussendorf_2001, Raussendorf_2003, Hein_2006} and their natural port structure.
    Here we specifically focus on bipartite graphs with $m+n$ nodes, $m$ of which belong to Alice and $n$ belong to Bob.
    Thus, we consider bipartite systems $A^mB^n$ with each of Alice and Bob's subsystems being a qubit.
    A brief overview of graph states is provided in \Cref{app:graph_states}.

Let $\cG_{m,n}$ denote the set of bipartite simple graphs with partition $(m,n)$.
This means that for any $\Gamma\in\cG_{m,n}$, its edges will only be connecting an Alice node to a Bob node.
    For example, the graph of $n$ independent EPR pairs shared between Alice and Bob belongs to $\cG_{n,n}$, but the fully connected GHZ graph state does not.
    For every $\Gamma\in\cG_{m,n}$, we let $\ket{\Gamma}$ denote the corresponding quantum graph state.
    Using the fact that $\Gamma_{B^n}\coloneqq\tr_{A^m}\op{\Gamma}{\Gamma}_{A^mB^n}$ is a normalized projection, we can write
\begin{align}
  \ket{\Gamma}_{A^mB^n} = \sqrt{2^{n}\rk(\Gamma_{B^n})}\,(V\otimes\Gamma_{B^n})\ket{\Phi^+_n}_{A^nB^n},
\end{align}
for a partial isometry $V$ from $A^n$ to $A^m$, because both sides are purifications of $\Gamma_{B^n}$.

From Proposition \ref{prop:fidelity-pure-state}, the PBT fidelities of $\ket{\Gamma}_{A^mB^n}$ depend only on $\Gamma_{B^n}$.
    Furthermore, two graph states $\ket{\Gamma}$ and $\ket{\Gamma'}$ will have the same reduced density for Bob iff there exists an isometry $U_{A^m}$ for Alice such that $\ket{\Gamma'}=U_{A^m}\otimes\1\ket{\Gamma}$.
    Additionally, the PBT fidelities are invariant under permutations on Bob's subsystems. In the following we propose a canonical form for any graph state based on these two unitary degrees of freedom.
    The canonical form will then be used to compute the PBT fidelities.

\subsubsection{PBT canonical form of graphs}

Our goal is to introduce a unique canonical graph $\wh{\Gamma} \in \cG_{r,n}$ for any graph $\Gamma \in \cG_{m,n}$ with $r\coloneqq\log_2\rk(\Gamma_{B^n})$ such that $\ket{\Gamma}\sim_{\Omega}\ket{\wh{\Gamma}}$ for $\Omega\in\{\pbt,\pbtcl,\pbtq\}$.
    While the form we propose is specific to bipartite graphs shared between Alice and Bob, its generalization to arbitrary graphs is straightforward.

Let $\Gamma$ be a simple graph with vertex set $\{A_1,\dots,A_m,B_1,\dots,B_n\}$ and adjacency matrix $N_\Gamma \in \Mat_{m+n}(\bZ_2)$ of the form
\begin{align}
  &N_{\Gamma} = \begin{pmatrix}
    N_{A^mA^m} & N_{A^mB^n}\\
    N_{A^mB^n}^{T} & N_{B^nB^n}
  \end{pmatrix}.
\end{align}  
It is not difficult to show that $r=\log_2\rk(\Gamma_{B^n})=\rk(N_{A^mB^n})$ \cite{Hein_2006}, where $\rk(N_{A^mB^n})$ is the rank over $\bZ_2$. Indeed, by Lemma~\ref{L:reduced-graph}(ii) with $T=A^m$, $\Gamma_{B^n}$ is the uniform mixture of the vectors $Z_{f(\mbf{b})}\ket{\Gamma-A^m}$ with $\mbf{b}$ in the row space of $N_{A^mB^n}$, and these vectors are orthonormal because $\{Z_W\ket{\Gamma-A^m}\}_W$ is an orthonormal basis.
\begin{definition}
We say that a graph $\Gamma $ with vertices $\{A_1,\dots,A_r,B_1,\dots,B_n\}$ is in \emph{PBT canonical form} (\emph{PCF}) if 
\begin{enumerate}[(i)]
    \item There are no edges among the $A^r$; \label{it:no-edge-Alice}
    \item There are no isolated vertices among the $A^r$;
    \item $B_i$ is only connected to $A_i$ for $i \leq r$ \label{it:epr-pairs}.
\end{enumerate}
\end{definition}
\noindent In terms of the adjacency matrix, conditions (i) --- (iii) of a PCF graph $\Gamma$ can be equivalently expressed as
\begin{enumerate}
    \item[(i)] $\iff$ $N_{A^rA^r}=0$;
    \item[(ii)]  + (iii) $\iff$ Up to a permutation of columns, 
    \begin{equation}
    \label{Eq:adjacency-reduced-form}N_{A^rB^n}=\begin{pmatrix}I_r&S\end{pmatrix}_{r\times n},\end{equation}
\end{enumerate}
where $S$ is an $r\!\times\! (n-r)$ matrix with elements in $\mbb{Z}_2$ and $I_r$ is the $r\times r$ identity matrix.
    Note that it follows from the definition of the PCF graph that $\Gamma_{B^n}$ has rank $2^r$ for any PCF graph $\Gamma \in \cG_{r,n}$, and hence necessarily $r\leq n$.

We now show that every graph state defined in terms of a simple graph $G$ is unitarily equivalent to a graph state defined in terms of a PCF which is uniquely determined by $G$.
We give the proof of this result in the main text, since it illustrates the construction of the PCF for a given graph.

\begin{proposition}\label{P:min-G}
  Let $\Gamma$ be a simple graph with vertices $A^mB^n$.
  Then there exists a unique PCF graph $\wh{\Gamma}$ with vertices $A^rB^n$ and $r = \log_2\rk(\Gamma_{B^n})$, a  unitary $U_{A^m}$ on Alice's side, and a permutation of subsystems $U_\pi$ on Bob's side such that
  \begin{equation}
  \ket{\wh{\Gamma}}_{A^rB^n} \otimes \ket{\wh{+}}_{A^{m-r}} =U_{A^m}\otimes U_\pi\ket{\Gamma},
  \end{equation}
  where $\ket{\wh{+}}_{A^{m-r}}\coloneqq\ket{+}^{\otimes(m-r)}_{A^{m-r}}$.
    Hence, $\ket{\Gamma} \sim_{\Omega} \ket{\wh{\Gamma}}$ for $\Omega\in\{\pbt,\pbtcl,\pbtq\}$.
\end{proposition}

\begin{proof}

Starting from the state $\ket{\Gamma}_{A^mB^n}$, Alice first disconnects any edges among the subsystems $A^m$ by a local unitary.
    The resulting graph $\Gamma'$ will have an adjacency matrix of the form
\begin{align}
  &N_{\Gamma'} = \begin{pmatrix}
    0 & N_{A^mB^n}\\
    N^{T}_{A^mB^n} & N_{B^nB^n}
  \end{pmatrix}.
  \end{align}

For every $\mbf{b}=b_1b_2\dots b_n\in\mbb{Z}_2^n$, we identify a subset of Bob's nodes $f(\mbf{b})\subset B^n$ such that $B_i\in f(\mbf{b})$ iff $b_i=1$.
    Using this representation, define a $\bZ_2$-linear subspace $\cN$ of $\bZ_2^n$ as
\begin{align}
  \cN(\Gamma) \coloneqq \Span_{\bZ_2} \{ \mbf{b}\;|\; f(\mbf{b})=N(A_j)\cap B^n,\; j=1,\dots, m\},
\end{align}
where $N(A_j)$ denotes the set of neighbors of the node $A_{j}$.
In other words, $\cN(\Gamma)$ is the row space of $N_{A^{m}B^{n}}$.
Then \Cref{L:reduced-graph} allows us to express Bob's reduced density matrix $\Gamma_{B^n}$ as
\begin{align}
  \Gamma_{B^n}=\Gamma'_{B^n} = \frac{1}{2^{r}}\sum_{\mbf{b} \in \cN(\Gamma)} Z_{f(\mbf{b})}\left( E \pure{\wh{+}}_{B^n} E\right) Z_{f(\mbf{b})},
\end{align}
where $E \coloneqq \prod_{e\in \cE_{B^n}} CZ_e$ with $\cE_{B^n}$ defined as the set of edges among the $B^n$ in $\Gamma$. 
 The key observation is that the summation depends only on the subspace $\mc{N}(\Gamma)$.
    We can therefore replace the $N(A_i)$ with any other basis for $\mc{N}(\Gamma)$.
In particular, the row space of the reduced row echelon form $R$ of $N_{A^mB^n}$ is $\mc{N}(\Gamma)$, up to a permutation of columns.
The form of $R$ is given by
\begin{equation}
    R = \begin{pmatrix}
        I_r & S \\ 0 & 0
    \end{pmatrix}_{m\times n},
\end{equation}
and a column permutation is equivalent to a relabeling of Bob's nodes $B_i\mapsto B_{\pi(i)}$, which is an allowed operation in our PBT framework.
Hence, with $\wh{R} = \begin{pmatrix} I_r & S \end{pmatrix}_{r\times n}$ there exists a graph $\wh{\Gamma}$ whose adjacency matrix is
\begin{align}
    N_{\wh{\Gamma}}=\begin{pmatrix}0&\wh{R}\\ \wh{R}^{T}&\pi N_{B^nB^n}\pi^{T}\end{pmatrix},
\end{align}
which is in PCF.
Since $\wh{\Gamma}_{B^n}=U_{\pi} \Gamma_{B^n} U_{\pi}^{\dagger}$, where $U_{\pi}$ applies the permutation $\pi$ to Bob's subsystems, there must exist some unitary $U_{A^m}$ on Alice's side such that $\ket{\wh{\Gamma}}_{A^r B^n} \otimes \ket{\wh{+}}_{A^{m-r}}=U_{A^m}\otimes U_\pi\ket{\Gamma}$, because both are purifications of $U_\pi\Gamma_{B^n}U_\pi^\dagger$ on $A^mB^n$.
By \Cref{prop:fidelity-pure-state}, this completes the proof.
\end{proof}

\begin{example}
  Consider a bipartite graph $\Gamma \in \cG_{3,3}$ and its adjacency matrix $N_{A^3B^3}$
  \begin{align}
  \Gamma = \begin{tikzpicture}[baseline=-2em]
    \Vertex{A1}
    \Vertex[y=-0.5]{A2}
    \Vertex[y=-1]{A3}
    \Vertex[x=1]{B1}
    \Vertex[x=1,y=-0.5]{B2}
    \Vertex[x=1,y=-1]{B3}
    \Edge(A1)(B1)
    \Edge(A1)(B2)
    \Edge(A2)(B1)
    \Edge(A2)(B3)
    \Edge(A3)(B1)
    \Edge(A3)(B3)
    \end{tikzpicture} \q\q\q
      N_{A^3B^3} = \begin{pmatrix}
        1 & 1 & 0 \\
        1 & 0 & 1 \\
        1 & 0 & 1
      \end{pmatrix}
    \end{align}
    Its associated PCF graph $\wh{\Gamma}$ is then
    \begin{align}
    \wh{\Gamma} = \begin{tikzpicture}[baseline=-2em]
    \Vertex{A1}
    \Vertex[y=-0.5]{A2}
    \Vertex[x=1]{B1}
    \Vertex[x=1,y=-0.5]{B2}
    \Vertex[x=1,y=-1]{B3}
    \Edge(A1)(B1)
    \Edge(A1)(B3)
    \Edge(A2)(B2)
    \Edge(A2)(B3)
    \end{tikzpicture}\q\q\q
      \wh{N}_{A^2B^3} = \begin{pmatrix}
        1 & 0 & 1 \\
        0 & 1 & 1 \\
      \end{pmatrix}.
    \end{align}
\end{example}

\begin{remark}
In general, every PCF graph $\Gamma \in \cG_{r,n}$ consists of $k$-EPR pairs ($k$ isolated edges) for some $k \leq \rk(\Gamma_{B^n})$ along with some connected component having more than a single edge.
Denoting $\Gamma_{\EPR{k}}$ as the graph consisting of $k$ EPR pairs and $\Gamma_0$ as the remaining subgraph with multiple connected edges, the graph state $\Gamma$ is expressed as
\begin{equation}
\label{Eq:PCF-k-iso-form}
    \Gamma \equiv \Gamma_{\EPR{k}} \otimes \Gamma_0.
\end{equation}
Without loss of generality, we assume that $\Gamma_{\EPR{k}}$ is held on nodes $\{A_i,B_i\}_{i=1}^k$ while $\Gamma_0$ belongs to nodes $\{A_{k+1},\dots,A_m,B_{k+1},\dots,B_n\}$.

\end{remark}

\begin{example} \label{ex:graph-2}
  \begin{align}
  \Gamma = \q\begin{tikzpicture}[baseline=-2.3em]
    \Vertex{A1}
    \Vertex[y=-0.5]{A2}
    \Vertex[y=-1]{A3}
    \Vertex[y=-1.5]{A4}
    \Vertex[x=1]{B1}
    \Vertex[x=1,y=-0.5]{B2}
    \Vertex[x=1,y=-1]{B3}
    \Vertex[x=1,y=-1.5]{B4}
    \Edge(A1)(B1)
    \Edge(A2)(B2)
    \Edge(A2)(B3)
    \Edge(A2)(B4)
    \Edge(A3)(B3)
    \Edge(A3)(B4)
    \Edge(A4)(B3)
    \Edge(A4)(B4)
    \end{tikzpicture} 
    \q\q\q
      N_{A^4B^4} = \begin{pmatrix}
        1 & 0 & 0 & 0\\
        0 & 1 & 1 & 1\\
        0 & 0 & 1 & 1\\
        0 & 0 & 1 & 1
      \end{pmatrix}
    \end{align}
    Its associated PCF graph $\wh{\Gamma}$ decomposes as
    \begin{align}
    \wh{\Gamma} &= \q\begin{tikzpicture}[baseline=-2.3em]
    \Vertex{A1}
    \Vertex[y=-0.5]{A2}
    \Vertex[y=-1]{A3}
    \Vertex[x=1]{B1}
    \Vertex[x=1,y=-0.5]{B2}
    \Vertex[x=1,y=-1]{B3}
    \Vertex[x=1,y=-1.5]{B4}
    \Edge(A1)(B1)
    \Edge(A2)(B2)
    \Edge(A3)(B3)
    \Edge(A3)(B4)
    \end{tikzpicture}
    \q\q\q
      N_{A^3B^4} = \begin{pmatrix}
        1 & 0 & 0 & 0\\
        0 & 1 & 0 & 0\\
        0 & 0 & 1 & 1
      \end{pmatrix}\\
    &= (\Gamma_{\EPR{2}})_{A_1A_2B_1B_2} \otimes (\wh{\Gamma}_{0})_{A_3B_3B_4}
    \end{align}
\end{example}

\Cref{P:min-G} implies that it is sufficient to consider the PCF of a graph when computing its PBT fidelities, which is what we do in the next section. 
But before doing so, we characterize the marginal states of PCF graphs.

\begin{lemma}\label{L:reduced-pbt-ens}
  Let $\Gamma \in \cG_{r,n}$ be in PBT canonical form with $k$-EPR pairs.
  Then the reduced states $\Gamma^{(i)}_{A^r B}\coloneqq\tr_{B_i^c}\op{\Gamma}{\Gamma}$ are
  \begin{align}
     \Gamma^{(i)}_{A^rB}=\begin{cases}
      \ds \Psi^+_{A_iB} \otimes \wh{\1}_{A^r\setminus A_i} & i \leq k\\
       \\
      \ds   E^{(i)}(\wh{\1}_{A^r} \otimes \pure{+}_{B})E^{(i)} & i > k
    \end{cases}
  \end{align}
  where $\wh{\1}$ denotes the maximally mixed state (i.e. the normalized identity), $\ket{\Psi^+}\coloneqq\frac{1}{\sqrt{2}}(\ket{0+}+\ket{1-})$, and $E^{(i)}$ is the unitary operator
  \begin{equation}\label{eq:edge-op-text}
    E^{(i)} = \prod_{A_j\in N(B_i)}CZ_{A_jB}.
  \end{equation}
\end{lemma}
\noindent The proof is given in \Cref{app:graph_states_lemma-reduced}.
Note that $\Gamma^{(i)}_{A^rB}$ is separable with respect to partition $(A^{r}, B)$ whenever $i>k$.

\subsubsection{PBT ordering among bipartite graphs}
We now analyze the PBT teleportation fidelities for bipartite graph states.
    As noted above, it suffices to consider PCF graphs in $\mc{G}_{r,n}$.
    Throughout this section, we assume that Bob has a port structure $B^n$ with each system $B_i$ being a qubit.
Hence, the classical threshold for all teleportation channels considered is $\frac{1}{2}$.
Our main finding shows that the teleportation fidelity of any bipartite PCF graph expressed in the form of Eq.~\eqref{Eq:PCF-k-iso-form} as $\Gamma = \Gamma_{\EPR{k}}\, \otimes\,\Gamma_0$ depends essentially on the number $k$ of isolated EPR pairs.
\Cref{Prop:PBT-connected} first addresses the extreme case in which $\Gamma$ does not have any isolated EPR pairs, i.e. $\Gamma=\Gamma_0$; its proof is given in \Cref{Appendix-proof_of_graph_no-epr}.

\begin{proposition}
\label{Prop:PBT-connected}
    Suppose that $\Gamma \in \cG_{r,n}$ is a PCF bipartite graph with $r=\log_2\rk(\Gamma_{B^n})$ and no independent EPR pairs.
    Then $\F_{\pbtq}(\Gamma) = \frac{1}{2}$ and
  \begin{align}
    \F_{\pbt}(\Gamma) &=\begin{cases} \frac{1}{2} - \frac{1}{2^{r+2}}\quad \text{if $\exists B_j$ w. $|N(B_j)|$ being even}\\ \frac{1}{2} - \frac{1}{2^{r+1}}\quad\text{otherwise}
    \end{cases}
  \end{align}
  \end{proposition}

We now consider the general case of a general bipartite PCF graph with $k$-EPR pairs, i.e.
$\Gamma = \Gamma_{\EPR{k}} \otimes \Gamma_{0}$.
As shown in Proposition \ref{Prop:PBT-connected}, $\Gamma_0$ is useless in itself for teleportation.
    However, when combined with $\Gamma_{\EPR{k}}$ one might speculate whether it can add to the overall teleportation power of $\Gamma$.
The following theorem shows that it cannot add more than what is feasible by just giving Bob two local ancillary states.
Its proof can be found in \Cref{app:graph-std-bound}.

\begin{theorem}\label{T:graph-std-bound}
  Suppose $\Gamma \in \cG_{r,n}$ has PCF with $k$ total isolated EPR pairs; i.e. $\Gamma=\Gamma_{\EPR{k}}\otimes\Gamma_0$.
    Then
  \begin{align}
    \Gamma_{\EPR{k}} \preceq_{\Omega} \Gamma \preceq_{\Omega} \Gamma_{\EPR{k}} \otimes \pure{+}_{B_{k+1}} \otimes \pure{-}_{B_{k+2}}
  \end{align}
  for $\Omega\in\{\pbt,\pbtcl\}$. Furthermore, $\Gamma\sim_{\pbtq}\Gamma_{\EPR{k}}$.
\end{theorem}

\subsubsection{Comparing $\F_{\pbt}(\Gamma_{\EPR{k}})$ and $\F_{\pbtcl}(\Gamma_{\EPR{k}})$}

\Cref{T:graph-std-bound} expresses the performance of $\pbt$, $\pbtcl$, and $\pbtq$ for any PCF bipartite graph state in relation to $\Gamma_{\EPR{k}} = \Psi^{+\otimes k}$, where $k$ denotes the number of isolated EPR pairs contained as a sub-graph.
Therefore, it is of interest to compute the teleportation fidelities of $\Gamma_{\EPR{k}}$.
The precise computation of $\F_{\pbtcl}(\Gamma_{\EPR{k}})$ is significantly more challenging than that of $\F_{\pbt}(\Gamma_{\EPR{k}})$: we cannot effectively utilize Schur--Weyl duality since the full  $\cU(2)$ symmetry is broken by the classical states in $B_{k+1}$ and $B_{k+2}$.
For arbitrary $k$, we can numerically compute $\F_{\pbtcl}$ via semidefinite-programming (SDP) although the numerical computation is not feasible for bigger $k$.
We refer to \Cref{Appendix:pbtcl-symmetries} for further details regarding this SDP formulation.

However, for $k = 1$ we obtain $\F_{\pbtcl}(\Gamma_{1\text{-EPR}}) = \frac{7}{12}$ from a straightforward calculation (see \Cref{sec:A:1-epr-Fpbt}), which surpasses the classical threshold of $0.5$.
This value corresponds to the PBT fidelity of $(\Gamma_{1\text{-EPR}})_{AB_1} \otimes \pure{01}_{B_2B_3}$. Notably, the minimum number of qubit systems required for the PBT protocol to surpass the classical threshold is four, with one system on Alice's side and three on Bob's side.
Indeed, the optimal PBT fidelity of arbitrary states in the qubit systems $A^2B^2$ is numerically determined to be $0.5$ \cite{Mozrzymas_2018}, and hence the optimal fidelity across three systems cannot exceed $0.5$.

\begin{figure*}
\begin{tikzpicture}
\begin{axis}[
ytick scale label code/.code={},
ymax = 1,
symbolic x coords={1,2,3,4,5,6,7},
xtick=data,
height=8cm,
width=10cm,
grid=major,
xlabel={\# of EPR pairs},
ylabel={Entanglement fidelity},
legend style={
cells={anchor=west},
at = {(1.05,0)},
anchor = south west}
]
\addplot[color=cb_orange,densely dotted,very thick] coordinates {
(1,0.625000000) (2,0.716506351) (3,0.781250000) (4,0.826588894) (5,0.858547963) (6,0.881472412) (7,0.898314184)
};
\addplot[color=cb_magenta,very thick,mark=square*] coordinates {
  (1,0.583333333) (2,0.684812868) (3,0.757203797) (4,0.811572127) (5,0.849094392) (6,0.875953745) (7,0.895077946)
};
\addplot[color=cb_teal,densely dashed,very thick] coordinates {
(1,0.562500000) (2,0.674839684) (3,0.750000000) (4,0.807838894) (5,0.846829213) (6,0.874775983) (7,0.894407934)
};
\addplot[color=cb_blue,very thick,mark=*] coordinates {
  (1,0.25) (2,0.466506351) (3,0.625) (4,0.732838894) (5,0.803860463) (6,0.850222412) (7,0.880736059)
};
\addplot[color=cb_purple,dashdotted,very thick] coordinates {
(1,0.5) (2,0.5) (3,0.5) (4,0.5) (5,0.5) (6,0.5) (7,0.5)
};
\legend{$\F_{\text{upper}}$,$\F_{\pbtcl}$, $\F_{\text{lower}}$,$\F_{\pbt}$,$\F_{\text{cl}}$}
\end{axis}
\end{tikzpicture}
     \caption{The teleportation fidelity of PBT protocols as a function of the number of EPR pairs used as a resource. 
     $\F_{\pbt}=\F_{\pbt}(\Gamma_{n\text{-EPR}})$ and $\F_{\pbtcl}=\F_{\pbtcl}(\Gamma_{n\text{-EPR}})$ are the entanglement fidelities defined in \eqref{eq:F-PBT} and \eqref{eq:F-PBTcl}, respectively.
     $\F_{\text{lower}}$ denotes the lower bound $\F_{\pbt}(\rho,M)$ on $\F_{\pbtcl}$ using the POVM $M$ in \Cref{eq:povm-ansatz}, and $\F_{\text{upper}}$ is the upper bound $\F_{\pbt}(\Gamma_{n\text{-EPR}}) + \frac{n+2}{2^{n+2}}$ on $\F_{\pbtcl}$ from \Cref{Prop:n-EPR}.
     $\F_{\text{cl}}=1/2$ shows the classical threshold.}
     \label{fig:fidelity}
\end{figure*}
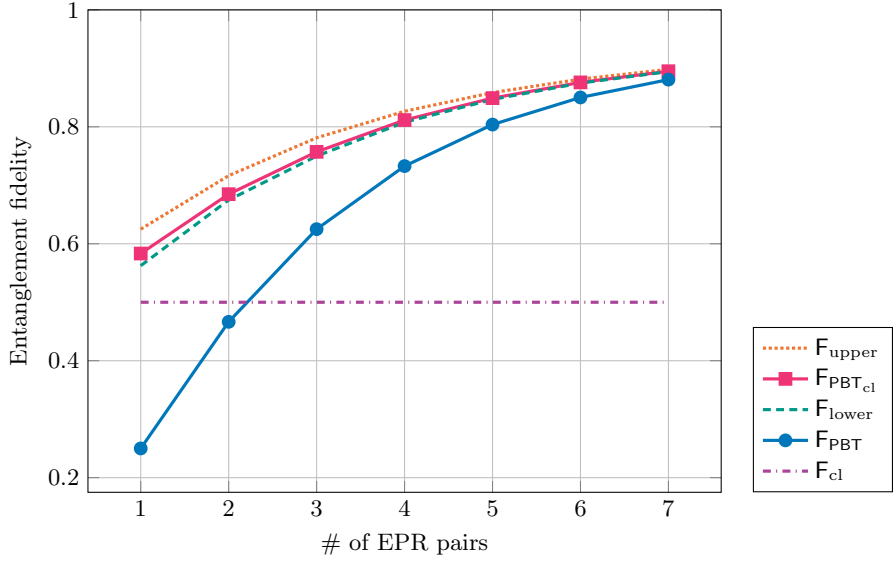

We now turn our attention to deriving bounds on $\F_{\pbtcl}(\Gamma_{\EPR{k}})$.
Let $\rho = \Gamma_{\EPR{k}}\otimes \pure{0}_{B_{k+1}} \otimes \pure{1}_{B_{k+2}}$, and recall from Eq.~\eqref{Eq:PBT*=PBT} that $\F_{\pbt}(\rho)=\F_{\pbtcl}(\Gamma_{\EPR{k}})$.
Since the fidelities are invariant under local unitaries on Alice's side, we write $\Gamma_{\EPR{k}}=H^{\otimes k}_{A^k}\Phi^{+\otimes k}H^{\otimes k}_{A^k}$ and work with $\Phi^{+\otimes k}$ in place of $\Gamma_{\EPR{k}}$ from here on, as in \Cref{app:Fpbt-upper-bound}.
In other words, finding $\F_{\pbtcl}(\Gamma_{\EPR{k}})$ is equivalent to finding an optimal POVM for the state discrimination problem associated with the state ensemble $\{\rho^{(i)} \coloneqq \Tr_{B_i^c}(\rho)\}_{i=1}^{k+2}$.
Let $\{E^{(j)}\}_{j=1}^k$ be the \emph{pretty good measurement} (PGM) associated with $\{\rho^{(j)}\}_{j=1}^k$, i.e.
\begin{equation}
    E^{(j)} \coloneqq \bar{\rho}^{-\frac{1}{2}} \rho^{(j)} \bar{\rho}^{-\frac{1}{2}}, \q\q \bar{\rho} = \sum_{i=1}^k \rho^{(i)},
\end{equation}
where $\bar{\rho}^{-\frac{1}{2}}$ denotes the pseudo-inverse of $\bar{\rho}^{\frac{1}{2}}$.
Since $E^{(j)}$ has no support on $\ket{0^k}_{A^k}\otimes \ket{1}_B$ or $\ket{1^k}_{A^k} \otimes \ket{0}_B$ for any $j$, we can establish a trivial lower bound for $\F_{\pbt}(\rho)$ by considering the following POVM $\Pi = \{\Pi^{(i)}\}_{i=1}^{k+2}$:
\begin{equation}\label{eq:lower-bound}
\Pi^{(j)} \coloneqq \begin{cases}
    E^{(j)} & j \leq k \\
    \pure{1^k}_{A^k} \otimes \pure{0}_B & j = k+1\\
    \pure{0^k}_{A^k} \otimes \pure{1}_B & j = k+2.
\end{cases}
\end{equation}
The operators in Eq.~\eqref{eq:lower-bound} are positive but need not sum to $\1$; for $k=1$, for instance, their sum has rank $3$ on the four-dimensional space $A_1B$.
This is harmless, since for a POVM it suffices to check that the effects sum to at most $\1$.
Indeed, $\sum_{j\leq k}E^{(j)}$ is the projection onto the support of $\bar{\rho}$, and $\Pi^{(k+1)}$, $\Pi^{(k+2)}$ are projections orthogonal to this support and to each other, so $\sum_{j}\Pi^{(j)}\leq\1$.
Adding the remainder $\1-\sum_j\Pi^{(j)}\geq0$ to any one effect gives a POVM, and this can only increase $\sum_j\Tr[\rho^{(j)}\Pi^{(j)}]$ because every $\rho^{(j)}$ is positive.
Since $\{E^{(i)}\}_{i=1}^k$ is the optimal measurement for standard PBT \cite{Leditzky_2020}, we achieve a separation between $\F_{\pbt}(\Gamma_{k\text{-EPR}})$ and $\F_{\pbtcl}(\Gamma_{k\text{-EPR}})$:
\begin{equation}
  \F_{\pbt}(\Gamma_{\EPR{k}}) + \frac{1}{2^{k+1}} \leq \F_{\pbtcl}(\Gamma_{\EPR{k}}).
\end{equation}
The term $\frac{1}{2^{k+1}}$ on the left-hand side of the inequality arises from the classical part $\frac{1}{2}\Tr[\rho^{(k+1)} \Pi^{(k+1)}]=\frac{1}{2}\Tr[(\pure{1}^{\otimes k}_{A^k} \otimes \pure{0}_B)\rho^{(k+1)}]$.

The above discussion gives a lower bound on the fidelity $\F_{\pbtcl}(\Gamma_{\EPR{k}})$.
In \Cref{app:Fpbt-upper-bound} we prove an  upper bound 
\begin{align}
	\F_{\pbtcl}(\Gamma_{\EPR{k}})\leq \F_{\pbt}(\Gamma_{\EPR{k}}) + \frac{2+k}{2^{k+2}}.
\end{align}
The gap between the two bounds vanishes exponentially in $k$.

We summarize these findings in the following proposition:
\begin{proposition}
\label{Prop:n-EPR}
\begin{align}
    \F_{\pbt}(\Gamma_{\EPR{k}}) + \frac{1}{2^{k+1}} &\leq \F_{\pbtcl}(\Gamma_{\EPR{k}}) \\
    &\leq \F_{\pbt}(\Gamma_{\EPR{k}}) + \frac{2+k}{2^{k+2}}
\end{align}
\end{proposition}

We further present a tighter lower bound given by the POVM $M = \{E^{(i)}\}_{i=1}^k \cup \{M^{(k+1)}, X^{\otimes k+1}M^{(k+1)} X^{\otimes k+1}\}$, where $M^{(k+1)}$ is defined as
\begin{equation}\label{eq:povm-ansatz}
    M^{(k+1)} \coloneqq \begin{cases} \ds
        \sum_{i=0}^{\floor{k/2}}P_i(\1-\wt{M})P_i + \frac{1}{2} P_{\ceil{k/2}}(\1-\wt{M})P_{\ceil{k/2}} \\ \hfill (k \text{ is odd})\\
        \ds\sum_{i=0}^{k/2}P_i(\1-\wt{M})P_i  \\ \hfill(k \text{ is even}).
    \end{cases}
\end{equation}
Here, $\wt{M} = \sum_{i=1}^k E^{(i)}$ and $P_i = \pure{\xi_i}$ is a projection where
\begin{align}
    \ket{\xi_i} &\coloneqq \frac{1}{\sqrt{|\cS_i| + |\cS_{i-1}|}}\bigg(\sum_{s\in \cS_{i}} \ket{s\,0}_{A^kB} - \sum_{t \in \cS_{i-1}} \ket{t\,1}_{A^kB}\bigg)\\
    \cS_{i} &\coloneqq \{ s \in \{0,1\}^{k} : k-|s| = i \},
\end{align}
with $|s|$ denoting the Hamming weight of $s$, so that $k-|s|$ is the number of zeros in $s$.
We can verify that $M$ is a valid POVM, and the numerical value of this lower bound is illustrated in Fig. \ref{fig:fidelity}.

\subsection{Symmetrized EPR states}

\label{sec:symmetrized_EPR}

In certain applications, it may be desirable to have permutation symmetry among the different port states.
A permutation invariant PBT state $\rho_{AB^n}$ is one such that $\rho^{(i)}_{AB}=\tr_{B_i^c}\rho_{AB^n}$ is the same reduced state for all $i$.
Utilizing these states ensures that the output of an asynchronous multi-port computation remains consistent, irrespective of the selected port.
Another useful symmetry is \emph{simultaneous unitary covariance} \cite{Christandl_2018}, which is crucial in applications such as universal programmable quantum processor \cite{Ishizaka_2008}, unitary gate estimation \cite{yoshida2024}, and port-based state preparation \cite{muguruza2024}.

Given an arbitrary state $\rho_{B^n}$ for Bob's system, we can construct a PBT channel that possesses both of these symmetries.
The resource state $\ket{\wt{\rho}}_{A^nB^n}$ we construct involves Alice's system also having a port structure $A^n$ with $|A_i|=|B_i|$, like it does when she and Bob share $n$ EPR pairs.
The method for building $\ket{\wt{\rho}}_{A^nB^n}$ begins by symmetrizing $\rho_{B^n}$ over system permutations and local copies of unitary operations
\begin{equation} \label{eq:symmetrization}
\rho_{B^n}\mapsto \wt{\rho}_{B^n}=\frac{1}{|S_n|}\sum_{\pi\in S_n}\int_{U(|B|)} U^{\otimes n}\pi\rho_{B^n} \pi^\dagger U^{\dagger \otimes n} \; dU.
\end{equation}
We then define $\ket{\wt{\rho}}_{A^nB^n}$ to be the canonical purification of $\wt{\rho}_{B^n}$,
\begin{equation}\label{eq:sym-pure}
    \ket{\wt{\rho}}_{A^nB^n} = \sqrt{|B|^n}(\mathbb{I}_{A^n} \otimes \wt{\rho}_{B^n}^{\frac{1}{2}})\ket{\Phi_{|B|}^+}^{\otimes n}.
\end{equation}
We emphasize that this is just a mathematical recipe for generating symmetric PBT states, and in general it does not correspond to a physical protocol that Alice and Bob can perform to build $\ket{\wt{\rho}}_{A^nB^n}$ from some initial $\rho_{A'B^n}$. 
By construction, the state $\ket{\wt{\rho}}_{A^nB^n}$ and its reduced state ensembles $\{\rho^{(i)} \equiv \Tr_{B_i^c} \op{\wt{\rho}}{\wt{\rho}}_{A^nB^n}\}_{i=1}^n$ will satisfy the following properties:
\begin{enumerate}[i)]
    \item $\ds \pi_{A^n}\rho_{A^nB}^{(i)} \pi_{A^n}^\dagger = \rho_{A^nB}^{(\pi(i))} \q\forall \pi \in S_n$
    \item $ \ds (U^{\otimes n}_{A^n} \otimes \overline{U}_B) \rho_{A^n B}^{(i)} (U^{\dagger \otimes n}_{A^n} \otimes U^T_{B}) = \rho_{A^nB}^{(i)}\q \forall U\in U(|B|)$
\end{enumerate}
Any state $\rho_{A^nB^n}$ satisfying these conditions is called \emph{symmetric PBT states}.

Symmetric PBT states enable the use of representation-theoretic tools to compute exact fidelity values.
    For example, consider any density matrix $\omega_{A^nB^n}$ that is invariant under the action of $U^{\otimes n}_{A^n}\otimes \overline{U}^{\otimes n}_{B^n}$.
    For $n=1$ these are the isotropic states.
    Assume further that $\omega_{A^n B^n}$ is invariant under arbitrary bilateral permutations $A_iB_i\mapsto A_{\pi(i)}B_{\pi(i)}$.
    Then the reduced state $\omega_{B^n}$ will be invariant under the twirling map of Eq.~\eqref{eq:symmetrization}, and let $\ket{\omega}_{A^nB^n}$ denote the canonical purification of $\omega_{B^n}$.
It is known that the PGM associated with the ensemble $\{\Phi^+_{A_iB} \otimes \wh{\1}_{A_i^c}\}_{i=1}^n$ yields the optimal $\F_{\pbt}(\pure{\omega})$.
Moreover, a formula for the fidelity is given by
\begin{equation} \label{eq:f-sym-pure}
    \F_{\pbt}(\pure{\omega}) = \frac{1}{4} \sum_{\alpha \vdash_2 n-1} (\sum_{\mu = \alpha + \square} \sqrt{c_{\mu} m_\mu d_\mu})^2,
\end{equation}
where $c_\mu$'s are the non-negative coefficients in the Schur--Weyl decomposition (see \Cref{Appendix:schur-weyl}) of $\omega_{B^n}$ \cite{Leditzky_2020}:
\begin{equation}
    \omega_{B^n} = \bigoplus_{\mu \vdash_d n} c_\mu \1_{V_{\mu}} \otimes \1_{W_{d,\mu}}.
\end{equation}
 
Our goal now is to perform a similar analysis and compute $\F_{\pbt}$ of the pure state $\ket{\wt{\Sigma}^{k,n}}_{A^nB^n}$ obtained by symmetrizing the state 
\begin{align}
   \Sigma^{k,n}_{B^n}=\wh{\1}_{B^k}\otimes\op{\wh{+}}{\wh{+}}_{B_{k+1}\dots B_n}. \label{eq:sigma-k-n}
\end{align}
We can interpret $\Sigma^{k,n}_{B^n}$ as arising from $k$ EPR pairs shared between Alice and Bob plus $(n-k)$ additional ports on Bob's side each being in the pure state $\ket{+}$.
    The PBT state $\ket{\wt{\Sigma}^{k,n}}_{A^nB^n}$ that we consider is the canonical purification of $\wt{\Sigma}^{k,n}_{B^n}$, with the latter being obtained from $\Sigma^{k,n}_{B^n}$ by applying the twirling map of Eq.~\eqref{eq:symmetrization}.

We first consider the simple case $k = 1$ to illustrate the proof idea, and then state a more general result in \Cref{P:k-epr-sym} below.
Schur--Weyl duality (see \Cref{Appendix:schur-weyl}) states that the symmetrization of the state in \eqref{eq:sigma-k-n} yields a twirled state $\wt{\Sigma}^{1,n}_{B^n}$ with the following block-diagonal form:
\begin{equation} 
\wt{\Sigma}^{1,n}_{B^n}
= \bigoplus_{\lambda \vdash_2 n} c_{\lambda} \1_{V_{\lambda}} \otimes \1_{W_{2,\lambda}}.
\end{equation}
The coefficients $c_\lambda$ can be computed via
\begin{equation}
c_\lambda = \frac{\Tr[\wt{\Sigma}^{1,n}_{B^n} P_\lambda]}{m_\lambda d_\lambda},
\end{equation}
where $P_\lambda$ is the isotypic projection (defined in \cref{eq:young-symmetrizer}) for the $S_n$-representation $V_\lambda$, and $m_\lambda$ and $d_\lambda$ are the dimensions of $W_{2,\lambda}$ and $V_\lambda$ respectively.
Then $x_\lambda \coloneqq c_\lambda m_\lambda d_\lambda$ becomes
\begin{align}
  x_\lambda &= \Tr[\wt{\Sigma}^{1,n}_{B^n} P_\lambda]\\
  &= \Tr \left[ \frac{1}{2} \1_{B_1} \otimes \pure{+}^{\otimes n-1}_{B^{n-1}} P_\lambda \right]\\
  &= \frac{1}{2} \bra{+}^{\otimes n-1} \Tr_{B_1}(P_\lambda) \ket{+}^{\otimes n-1}\\
  &= \frac{1}{2} \sum_{\alpha \in \lambda - \square} \frac{m_{\lambda}}{m_{\alpha}} \bra{+}^{\otimes n-1} P_\alpha \ket{+}^{\otimes n-1}.
\end{align}
For the last equality, we used the following formula for calculating the partial trace $\Tr_{B_1} P_\lambda$ \cite{Leditzky_2020}
\begin{equation}\label{eq:sym-partial-tr}
  \Tr_{B_1} P_\lambda = m_\lambda \sum_{\alpha \in \lambda - \square} \frac{1}{m_\alpha} \cdot P_\alpha,
\end{equation}
where $\alpha \in \lambda-\square$ denotes that the summation is over all Young diagrams $\alpha$ obtained by removing a box from $\lambda$.
Since $P_{\alpha}$ for any $\alpha = (\alpha_1,\alpha_2)$ with $\alpha_2 > 0$ involves an antisymmetrization over at least two systems, we have 
\begin{align}
    \bra{+}^{\otimes n-1} P_\alpha \ket{+}^{\otimes n-1} = \begin{cases} 1 & \text{if } \alpha = (n-1,0);\\ 0 & \text{otherwise} \end{cases}.
\end{align}
The only possible $\lambda = (\lambda_1,\lambda_2) \vdash_2 n$ with $\alpha = \lambda-\square$ and $x_\lambda\neq 0$ are $(n,0)$ and $(n-1,1)$.
In these cases, by the dimension formula \eqref{eq:dim-irreps}
\begin{align} 
x_\lambda = \frac{m_\lambda}{2\cdot m_{(n-1,0)}} = \frac{\lambda_1-\lambda_2+1}{2n}.
\label{eq:x-lambda-expression}
\end{align}
When $n = 2$, $\alpha = (1,0)$ is the only possible Young diagram and we have
\begin{align}
  &\F_{\pbt}\left(\wt{\Sigma}^{1,2}_{A^2B^2}\right)
  = \frac{1}{4} \sum_{\alpha \vdash_2 1} \Bigg(\sum_{\lambda \in \alpha + \square} \sqrt{x_\lambda} \Bigg)^2\\
  &=  \frac{1}{4} \left(\sqrt{x_{(2,0)}} + \sqrt{x_{(1,1)}} \right)^2\\
  &= \frac{1}{8}(2+\sqrt{3}),
\end{align}
whose value is smaller than the classical threshold of $\frac{1}{2}$.
For $n>2$,
\begin{align}
  &\F_{\pbt}\left(\wt{\Sigma}^{1,n}_{A^nB^n}\right)
  = \frac{1}{4} \sum_{\alpha \vdash_2 n-1} \Bigg(\sum_{\lambda \in \alpha + \square} \sqrt{x_\lambda}\Bigg)^2\\
  &= \frac{1}{4} \Bigg( \sum_{\lambda \in (n-1,0)+\square} \sqrt{x_\lambda} \Bigg)^2 + \frac{1}{4} \Bigg( \sum_{\lambda \in (n-2,1)+\square} \sqrt{x_\lambda}\Bigg)^2\\
  &= \frac{1}{4} \left(\left(\sqrt{x_{(n,0)}} + \sqrt{x_{(n-1,1)}} \right)^2 + x_{(n-1,1)} \right)\\
  &= \frac{1}{4} \left( \left(\sqrt{\frac{n+1}{2n}} + \sqrt{\frac{n-1}{2n}}\right)^2 + \frac{n-1}{2n}\right)\\
  &= \frac{3n+2\sqrt{n^2-1}-1}{8n}
\end{align}

We can generalize this result to all $k$ and $n$ with $2k-1\leq n$, as stated in the following proposition and proved in \Cref{sec:A-lemmas}.

\begin{proposition}\label{P:k-epr-sym}
Let $k, n$ be positive integers with $2k-1\leq n$. Then
\begin{align}\label{eq:sym-fidelity-app}
\begin{aligned}
  \F_{\pbt}(\wt{\Sigma}^{k,n})
  &= (2^{k+2}(n-k+1))^{-1}\Bigg[(n-2k+1) \delta_{n-1 \geq 2k}\\
  &+ \sum_{i=1}^{k} \delta_{n+1\geq 2i} \Big(\sqrt{ (n-2i+3) \binom{k}{i-1}}\\
  &+ \sqrt{(n-2i+1) \cdot \binom{k}{i}} \Big)^2 \Bigg]
\end{aligned}
\end{align}
where $\delta$ is the indicator function.
Moreover, for a fixed $k$ we have
\begin{multline}\label{eq:sym-fixed-k}
  \lim_{n\to \infty} \F_{\pbt}(\wt{\Sigma}^{k,n})\\ = 2^{-k-2}\Bigg(1 + \sum_{i=1}^{k} \Bigg(\sqrt{\binom{k}{i-1}} + \sqrt{\binom{k}{i}} \Bigg)^2\Bigg).
\end{multline}
\end{proposition}

Specifically when $k = 1$ and $n = 3$, we observe $\F_{\pbt}(\ket{\wt{\Sigma}^{1,3}}) \approx 0.5690$, which exceeds the classical bound.
From the fact that the optimal PBT fidelity with any resource state on $A^2B^2$ is numerically determined to be 0.5 \cite{Mozrzymas_2018}, we need at least three ports on Bob's side to witness quantum advantage.
However, as $n$ grows the asymptotic limit reaches
\begin{align}
  \lim_{n\to \infty} \F_{\pbt}(\wt{\Sigma}^{1,n}) = \frac{1 + (1 + 1)^2}{8} = \frac{5}{8} = \F_{\pbt}(\Phi^{+\otimes 3}),
\end{align}
which means that symmetrizing one EPR pair over however many ports cannot perform better than having three EPR pairs of entanglement.

\subsection{General states}

\label{sec:general_fidelity}

A close examination of the optimal measurement for Alice found in \Cref{sec:isotropic-q} for $\F_{\pbtq}(\rho_{d,f})$ shows that it involves a projection onto the maximally entangled state for $\Pi^{(1)}$, and then a choice of $\Pi^{(i)}$ and $\sigma^{(i)}$ (for $i>1$) such that $\sum_i\Pi^{(i)}_{AB}=\1\otimes\1-\Phi^+_d$ and $\sum_{i=1}^d\tr[\Pi^{(1+i)}_{AB}(\1\otimes\sigma^{(i)})]=d^2-1$. The reason why the latter equality can be achieved is because $\1\otimes\1-\Phi^+_d$ is separable.
In fact, it is known that $\1\otimes\1-\op{\varphi}{\varphi}$ is separable for \textit{any} bipartite state $\ket{\varphi}$ \cite{Bandyopadhyay-2004a}.
We should therefore expect that part of Theorem \ref{thm:iso-fidelity-q} can be generalized.
The following proposition shows this is the case.
\begin{proposition}
\label{Prop:fidelity-general}
    Let $\rho_{AB}$ be any bipartite state such that $\rho_{A}=\frac{\1}{|A|}$.
    Then
    \begin{align}
        \F_{\pbtq}(\rho_{AB})\geq 
        \max\left\{\frac{1}{|B|}+\frac{\Vert\rho_{AB}\Vert_\infty}{|B|^2}-\frac{1}{|A||B|^2},\;\frac{1}{|B|}\right\}.
    \end{align}
\end{proposition}

\begin{proof}
    Let $\ket{\varphi}$ be an eigenvector of $\rho_{AB}$ such that $\bra{\varphi}\rho_{AB}\ket{\varphi}=\Vert\rho_{AB}\Vert_\infty$.
    Then as noted above, $\1\otimes\1-\op{\varphi}{\varphi}$ is separable, and we can write it as $\1\otimes\1-\op{\varphi}{\varphi}=\sum_{i=1}^{s}\op{\alpha_i}{\alpha_i}\otimes\op{\beta_i}{\beta_i}$ with the $\ket{\beta_i}$ being normalized pure states and $|A||B|-1=\sum_{i}^s\ip{\alpha_i}{\alpha_i}$.
    The PBT protocol then involves Bob preparing the state $\ket{\beta_i}$ in port $i+1$, and Alice performs the POVM with measurement operator $\Pi^{(1)}=\op{\varphi}{\varphi}$ and $\Pi^{(1+i)}=\op{\alpha_i}{\alpha_i}\otimes\op{\beta_i}{\beta_i}$ for $i=1,\dots, s$.
    Hence,
    \begin{align}
        &\F_{\pbtq}(\rho_{AB})\notag\\
        &\geq \frac{1}{|B|^2}\!\left(\!\tr[\Pi^{(1)}\rho_{AB}]\!+\!\frac{1}{|A|}\sum_{i=1}^s\tr[\Pi^{(1+i)}(\1\otimes\op{\beta_i}{\beta_i})]\right)\\
        &=\frac{1}{|B|^2}\left(\Vert\rho_{AB}\Vert_\infty+\frac{1}{|A|}\sum_{i=1}^s\ip{\alpha_i}{\alpha_i}\right)\\
        &=\frac{1}{|B|}+\frac{\Vert\rho_{AB}\Vert_\infty}{|B|^2}-\frac{1}{|A||B|^2}.
    \end{align}
\end{proof}

As an application of Proposition \ref{Prop:fidelity-general}, we show that any bipartite state that can exceed the classical teleportation threshold using one-way LOCC can also exceed the threshold under $\pbtmax$.
Hence, interactive local quantum processing is not necessary to realize the one-way teleportation capability of any bipartite state.
\begin{corollary}\label{cor:one-way}
    A bipartite state $\rho_{AB}$ can generate a $|C|$-dimensional teleportation channel whose fidelity exceeds $\frac{1}{|C|}$ by one-way LOCC if and only if it can also generate such a channel using $\pbtmax$.
\end{corollary}
\begin{proof}
    Denote by $d$ the local dimension $|C|$. As shown in Ref.~\cite{Chitambar_2023}, a $d$-dimensional one-way teleportation channel built using $\rho_{AB}$ can attain a fidelity above $\frac{1}{d}$ \emph{iff} there exists a local CPTP map $\mc{E}:B\to C$ on Bob's side such that
    \begin{align}
    \label{Eq:Bob_processing-RC}
        \sigma_A\otimes\1_C-\sigma_{AC}\not\geq 0,
    \end{align}
    where $\sigma_{AC}=\id_A\otimes\mc{E}_{B\to C}(\rho_{AB})$ so that $\sigma_A=\rho_A$.
    Suppose this condition is satisfied so that there exists a map $\mc{E}:B\to C$ on Bob's side and a bipartite state $\ket{\psi}$ such that $\bra{\psi}\sigma_{AC}\ket{\psi}>\bra{\psi}\sigma_A\otimes\1\ket{\psi}$.
    Let $P_A$ denote the projector onto the support of $\Psi_A=\tr_C\op{\psi}{\psi}_{AC}$ and $P_A'$ the projector onto the support of $P_A\sigma_A P_A$.  Note, $\bra{\psi}\sigma_{AC}\ket{\psi}>0$ implies that $P_A\sigma_A P_A$ is nonzero, and so $\Vert(P_A\sigma_A P_A)^{-1}\Vert_\infty$ is finite. 
    Define the states $\hat{\sigma}_{AC}=\frac{1}{\tr[P_A']}(P_A\sigma_AP_A)^{-1/2}(P_A\sigma_{AC}P_A)(P_A\sigma_AP_A)^{-1/2}$ and $\ket{\varphi}=\frac{\sqrt{P_A\sigma_AP_A}\otimes\1\ket{\psi}}{\sqrt{\bra{\psi}P_A\sigma_AP_A\otimes\1\ket{\psi}}}$ so that
    \begin{align}
        1< \frac{\bra{\psi}\sigma_{AC}\ket{\psi}}{\bra{\psi}\sigma_A\otimes\1\ket{\psi}}
        &=\frac{\bra{\psi}P_A\sigma_{AC}P_A\ket{\psi}}{\bra{\psi}P_A\sigma_AP_A\otimes\1\ket{\psi}}\\
        &=\tr[P'_A]\bra{\varphi}\hat{\sigma}_{AC}\ket{\varphi}.
    \end{align}
    Therefore, the state $\hat{\sigma}_{AC}$ satisfies $\hat{\sigma}_A=\frac{1}{\tr[P'_A]}P'_A$ and $\Vert\hat{\sigma}_{AC}\Vert_\infty>\frac{1}{\tr[P'_A]}$.
    We can hence apply Proposition \ref{Prop:fidelity-general} to the state $\hat{\sigma}_{AC}$ to obtain
    \begin{align}
         \F_{\pbtq}(\hat{\sigma}_{AC})\geq\frac{1}{d}+\frac{\Vert\hat{\sigma}_{AC}\Vert_\infty}{d^2}-\frac{1}{\tr[P'_A]d^2}>\frac{1}{d}.
    \end{align}
    
    A full teleportation protocol in $\pbtds$ consists of Bob applying the pre-processing channel $\mc{E}$ of Eq. \eqref{Eq:Bob_processing-RC} to his system of $\rho_{AB}$, and also preparing in auxiliary ports the classical states $\op{1}{1}\otimes\dots\otimes\op{d}{d}$ along with the states $\ket{\beta_1}\otimes\dots\otimes\ket{\beta_s}$ from Proposition \ref{Prop:fidelity-general}.
    Alice's action then consists first of performing a filtering measurement that either applies $(P_A\sigma_AP_A)^{-1/2}$ to her system with some nonzero probability, or it collapses her system to some failure state.
    Formally, this is done using an instrument with two outcomes defined by Kraus operators $M_0=\lambda (P_A\sigma_AP_A)^{-1/2}$ and $M_1=\sqrt{\1-M_0^\dagger M_0}$, respectively, where $\lambda=\Vert (P_A\sigma_A P_A)^{-1}\Vert_\infty^{-1/2}$. The choice of $\lambda$ gives $M_0^\dagger M_0=\lambda^2(P_A\sigma_AP_A)^{-1}\leq\1$, so $M_1$ is well defined and $M_0^\dagger M_0+M_1^\dagger M_1=\1$.
    In the first case, Alice and Bob follow the PBT protocol of Proposition \ref{Prop:fidelity-general} to attain a conditional fidelity larger than $\frac{1}{d}$.
    In the second case, Alice and Bob just perform the optimal classical protocol on the channel input to attain a fidelity of $\frac{1}{d}$.
    Since Bob has prepared classical registers, the classical protocol can be run as a PBT protocol.
    Averaging over both cases yields a fidelity strictly larger than $\frac{1}{d}$.
\end{proof}

\section{Conclusion}
In this paper, we have developed a class of quantum resource theories (QRTs) for one-way communication protocols with pre-shared entanglement between two parties, Alice and Bob, motivated by port-based teleportation originally proposed by Ishizaka and Hiroshima \cite{Ishizaka_2008}.
These QRTs effectively capture the time-independence of Bob's local processing, a feature that is essential for numerous applications of PBT.
Our operational models are defined by free operations that consist of local operations on Alice's side, one-way classical communication, trivial decoding achieved by discarding Bob's subsystems, and a subset of quantum pre-processing operations on the receiver's side prior to receiving classical information.
The QRT of $\pbt$ pertains to the minimal possible decoding on Bob's side, which is analogous to the original protocol.
We then defined $\pbtcl$ as a scenario that combines $\pbt$ with arbitrary classical processing.  
Interestingly, we found that the optimal way to integrate classical processing into PBT is not simply to probabilistically choose between performing standard PBT and the optimal classical protocol.
Instead, the optimal strategy involves using the classical states as additional ports on Bob's side and modifying Alice's POVM to select them.  The fidelity can be improved even further by using quantum states in the additional ports, and we explored this in the QRT of $\pbtq$. 

For $\Omega \in \{\pbt,\pbtcl,\pbtq\}$, we showed that $\F_{\Omega}(\rho)$ cannot exceed the classical teleportation bound if each state $\rho^{(i)} \equiv \Tr_{B_i^c}\rho$ in the PBT ensemble satisfies the reduction criterion.  Conversely, we showed that if $\rho_{AB}$ can violate the reduction criterion after a pre-processing map on Bob's side, then the classical teleportation bound can always be broken using $\pbtmax$.  This finding not only demonstrates the power of $\pbtmax$, but it also establishes another connection between teleportation and the reduction criterion.

We also considered specific examples of $\rho_{AB}$.
We first examined the family of isotropic states involving one port on Bob's side and obtained specific values for $\F_{\Omega}(\rho_{AB})$.
These values, together with the example \eqref{eq:reduction-example}, enabled us to demonstrate strict inclusions in the hierarchy $\pbt \subsetneq \pbtcl \subsetneq \pbtq \subsetneq \pbtmax$.
We then moved on to the case where $\rho_{AB}$ is an arbitrary bipartite graph state with any number of ports.
We proposed a PBT canonical form (PCF) for such states, which is an LU-equivalent graph that yields the same optimal fidelity for the various teleportation protocols we consider.
Surprisingly, we could show that the number of isolated EPR pairs in a PCF graph state $\Gamma_{AB}$ exactly determines $\F_{\pbtq}(\Gamma_{AB})$, and it bounds $\F_{\pbt}(\Gamma_{AB})$ and $\F_{\pbtcl}(\Gamma_{AB})$ from above and below.
As a special case, we showed that $\F_{\pbtcl}(\Phi^+) = \F_{\pbt}(\Phi^+_{AB_1}\otimes \pure{0}_{B_2} \otimes \pure{1}_{B_3})$ exceeds the classical bound.
This provides the smallest possible example in which any protocol from $\pbt$, $\pbtcl$, or $\pbtq$ can exceed the classical threshold, since it is not possible to do so using these operations when Bob has two qubits \cite{Mozrzymas_2018}.
On the other hand, since $\pbtmax$ is sufficiently strong to break the classical threshold for all entangled two-qubit states, it is an interesting open problem to identify the minimal type of pre-processing beyond $\pbtq$ needed for non-classical teleportation of two-qubit states.

\paragraph*{Acknowledgments.}
We thank Juntai Zhou for helpful discussions. This research was supported by a grant through the IBM-Illinois Discovery Accelerator Institute. The calculations in this paper were verified with the assistance of a generative AI model (Claude Opus 5.5).

\label{Section:Conclusion}

\bibliography{pbt}

\newpage

\onecolumngrid

\appendix

\section{Cardinality bound for $\pbtq$ (Proof of Proposition \ref{prop:cardinality-bound})}

\label{Appendix-Cardinality bound}

Here and in the following appendices, we give detailed proofs of some of our main results. 
These propositions and theorems are restated without numbers at the beginning of each section for the convenience of the reader.

\begin{proposition*}[Restatement of Prop.~\ref{prop:cardinality-bound}]
    Every channel in $\pbtq(\rho_{AB^n}, C_0\to B)$ can be generated by a POVM of Alice having $s\leq n+|C_0|^2(|A|^2+|B|^2)$ outcomes.
\end{proposition*}
\begin{proof}
First observe that we can express the action of $\Lambda_{C_0\to B}$ on an arbitrary state $\tau_{C_0}$ as
\begin{align}
    \Lambda_{C_0\to B}(\tau_{C_0})
    &=\tr_{C_0'}[\tau_{C_0'}\id_{C_0'}\otimes\Lambda_{C_0\to B}(\mbb{F}_{C_0C_0'})]\\
    &=\sum_{i=1}^n\tr_{AC_0 C_0'}[\tau_{C_0'}\Pi^{(i)}_{AC_0}(\rho^{(i)}_{AB}\otimes \mbb{F}_{C_0C_0'})]\\
    &+\sum_{i=n+1}^s\tr_{C_0'}[\tau_{C_0'}\tr_{AC_0}[\Pi^{(i)}_{AC_0}(\rho_{A}\otimes\mbb{F}_{C_0C_0'})]]\sigma^{(i)}_{B}
\end{align}
where $\mbb{F}_{C_0C_0'}=\sum_{i,j=1}^{|C_0|}\op{ij}{ji}_{C_0C_0'}$.
Let us focus on the part $\sum_{i=n+1}^s\tr_{AC_0}[\Pi^{(i)}_{AC_0}(\rho_A\otimes\mbb{F}_{C_0C_0'})]\sigma^{(i)}_{B}$ in the last term of this equation.
For any POVM $\{\Pi^{(i)}_{AC_0}\}_i$ by Alice and any state preparations $\{\sigma_{B}^{(i)}\}_{i=1}^s$ by Bob, let us denote 
\begin{align}
    \Omega_{AC_0}&=\sum_{i=n+1}^s\Pi^{(i)}_{AC_0}=T\sum_{i=n+1}^sp_i\hat{\Pi}^{(i)}_{AC_0} \\
    \Upsilon_{C_0'B}&=\sum_{i=n+1}^sp_i\tr_{AC_0}[\hat{\Pi}^{(i)}_{AC_0}(\rho_{A}\otimes\mbb{F}_{C_0C_0'})]\sigma^{(i)}_{B},
\end{align} 
where $T=\sum_{i=n+1}^s\tr[\Pi^{(i)}_{AC_0}]$, $p_i= \tfrac{\tr[\Pi^{(i)}_{AC_0}]}{\sum_{i=n+1}^s\tr[\Pi^{(i)}_{AC_0}]}$ and $\hat{\Pi}^{(i)}_{AC_0}=\tfrac{\Pi^{(i)}_{AC_0}}{\tr[\Pi^{(i)}_{AC_0}]}$.
Thus, $(T^{-1}\Omega_{AC_0},\Upsilon_{C'_0B})$ lies in the convex hull generated by the points $(\hat{\Pi}^{(i)}_{AC_0}, \tr_{AC_0}[\hat{\Pi}^{(i)}_{AC_0}(\rho_{A}\otimes\mbb{F}_{C_0C_0'})]\sigma^{(i)}_{B})$.
Since this is a convex subset of $\mbb{R}^{|A|^2|C_0|^2-1}\times\mbb{R}^{|C_0|^2|B|^2}$, Carath\'{e}odory's theorem ensures that $(T^{-1}\Omega_{AC_0},\Upsilon_{C_0'B})$ can be generated by no more than $\kappa\coloneqq|C_0|^2(|A|^2+|B|^2)$ of these points.
Consequently, we can always replace Alice's POVM to have no more than $n+\kappa$ outcomes while still generating the same channel $\Lambda_{C_0\to B}$.
\end{proof}

\section{Computing $\F_{\pbtcl}(\rho_{d,f})$ (Proof of \Cref{thm:iso-pbtcl})} \label{app:iso-pbtcl-proof}

\begin{theorem*}[Restatement of Thm.~\ref{thm:iso-pbtcl}]
    The classical-processing PBT fidelity of the isotropic state $\rho_{d,f}$ is given by
	\begin{align}
		\F_{\pbtcl}(\rho_{d,f}) = \begin{cases}\displaystyle
			\frac{1}{d} + \frac{(1-df)^2}{d^2(d^2f-1)} & \text{\upshape if $f > \frac{1}{d}$},\\[1em]
			\displaystyle\frac{1}{d} & \text{\upshape otherwise.}
		\end{cases} 
	\end{align}
	For $f>\frac{1}{d}$, this fidelity exceeds the classical fidelity, $\F_{\pbtcl} > \frac{1}{d}$.
\end{theorem*}

\begin{proof}

Let $f>\frac{1}{d}$. We first consider the dual formulation of the semidefinite program \eqref{eq:Fpbtcl-SDP-primal} characterizing the teleportation fidelity $\F_{\pbtcl}$:
\begin{align}
	\min \left\lbrace \tr(K): K \geq \frac{1}{d^2} \sigma^{(i)} \text{ for } i=1,\dots,d+1 \right\rbrace,\label{eq:dual}
\end{align}
where $\sigma^{(i)} = \frac{1}{d}\1 \otimes |i\rangle\langle i|$ for $i=1,\dots,d$, and $\sigma^{(d+1)} = \rho_{d,f}$.

We make the following ansatz for $K$ in \eqref{eq:dual}:
\begin{equation}
	K = \frac{1}{d^3} \1_{d^2} + \frac{f(1-df)}{d^2(1-d^2f)} |\gamma\rangle\langle\gamma| - \frac{1-df}{d^3(1-d^2f)} \sum\nolimits_i |ii\rangle\langle ii|.\label{eq:K}
\end{equation}
Here $\ket{\gamma}\coloneqq\sum_{i=1}^d\ket{ii}=\sqrt{d}\,\ket{\Phi_d^+}$.
To show that this operator is feasible in \eqref{eq:dual}, we first consider the constraint $K-\frac{1}{d^2}\sigma^{(i)} \geq 0$ for $i=1,\dots,d$. 
Assuming $i=1$ for simplicity (the other cases follow analogously), we have with $\sigma_1 = \frac{1}{d}\1_d\otimes |1\rangle\langle 1|$ that
\begin{equation}
	K-\frac{1}{d^2}\sigma_1 = \frac{1}{d^3} \1_d \otimes \sum_{j=2}^d |j\rangle\langle j| + \frac{f(1-df)}{d^2(1-d^2f)} |\gamma\rangle\langle\gamma| - \frac{1-df}{d^3(1-d^2f)} \sum_{i=1}^d |ii\rangle\langle ii|.\label{eq:K-minus-state}
\end{equation}
This operator has a block-diagonal structure
\begin{align}
    K-\frac{1}{d^2}\sigma_1 \cong \frac{1}{d^3} \1 \oplus M,
\end{align}
consisting of a $(d\times d)$-matrix $M$ (given in \eqref{eq:M} below) acting on the subspace spanned by $\lbrace |ii\rangle\rbrace_{i=1}^d$, and a diagonal part with elements $1/d^3 \geq 0$.
Positive semidefiniteness of the operator in \eqref{eq:K-minus-state} is therefore equivalent to positive semidefiniteness of the $M$ matrix given by
\begin{align}
	M &= \begin{pmatrix}
		b-c & b & b&\dots & b\\
		b & a+b-c & b &\dots & b\\
		b & b & a+b-c & \dots & b\\
		\vdots& \vdots & \vdots & \ddots & \vdots\\
		b & b & b & \dots & a+b-c
	\end{pmatrix}, 
    \label{eq:M}
    \intertext{where}
    a &= \frac{1}{d^3}, \quad b  = \frac{f(1-df)}{d^2(1-d^2f)}, \quad c  = \frac{1-df}{d^3(1-d^2f)}.
\end{align}
This matrix has eigenvalues  
\begin{equation}
    a-c = \frac{1}{d^3}\left(1 - \frac{1-df}{1-d^2f}\right) \geq 0
    \text{\q(since $f\geq 1/d^2$),}
\end{equation}
and $\frac{1}{2} \left(a + db - 2c \pm \sqrt{D}\right)$, with $D = (a+db)^2 - 4ab$.
A straightforward calculation shows that, since $f\geq 1/d^2$,
\begin{align}
	a+db-2c = \sqrt{D} = \frac{1}{d^3(d^2f-1)}\left(d^3f^2 - 2df + 1\right) \geq 0.
\end{align}
So the eigenvalue with the minus sign is $0$, and the one with the plus sign is $a+db-2c\geq0$.
Hence, $M$ has non-negative eigenvalues, and the operator $K-\frac{1}{d^2}\sigma_1$ in \eqref{eq:K-minus-state} is indeed positive semidefinite.
The same argument can be used to show that $K-\frac{1}{d^2}\sigma^{(i)} \geq 0$ for $i=2,\dots,d$.

It remains to check the feasibility constraint for the isotropic state $\rho_{d,f} = \frac{1-f}{d^2-1}\1_{d^2} + \frac{d^2f-1}{d(d^2-1)}|\gamma\rangle\langle\gamma|$:
\begin{align}
	K - \frac{1}{d^2} \rho_{d,f} &= \underbrace{\left(\frac{1}{d^3}-\frac{1-f}{d^2(d^2-1)}\right)}_{{} \eqqcolon x} \1_{d^2}
    \quad{} + \underbrace{\left(\frac{f(1-df)}{d^2(1-d^2f)} - \frac{d^2f-1}{d^3(d^2-1)} \right)}_{ {} \eqqcolon y} |\gamma\rangle\langle\gamma|
    \quad {} - \underbrace{\frac{1-df}{d^3(1-d^2f)}}_{ {} \eqqcolon z}  \sum\nolimits_i |ii\rangle\langle ii|. \label{eq:K-minus-isotropic}
\end{align}
Similar to before, this matrix has a block-diagonal structure 
\begin{align}
    K - \frac{1}{d^2} \rho_{d,f} \cong x \1 \oplus N,
\end{align}
consisting of a diagonal matrix with diagonal elements $x$ defined in the first line of \eqref{eq:K-minus-isotropic}, and a $(d\times d)$-matrix $N$ (given in \eqref{eq:N} below) acting on the subspace spanned by $\lbrace |ii\rangle\rbrace_{i=1}^d$.
Since $d>1$ and $1/d^2\leq f\leq 1$, we have $x\geq 0$, and hence the diagonal part of the matrix \eqref{eq:K-minus-isotropic} is positive semidefinite.
The matrix $N$ is given by
\begin{align}
	N = \begin{pmatrix}
		x+y-z & y &  \dots & y\\
		y & x + y - z &  \dots & y\\
		\vdots & \vdots & \ddots & \vdots\\
		y & y &  \dots & x+y-z
	\end{pmatrix}.
    \label{eq:N}
\end{align}
Its distinct eigenvalues are $\lambda_1=0$ (with multiplicity $1$) and 
\begin{align}
	\lambda_2=\frac{d^2f^2 + (d^3-2d^2-d)f+1}{d^2(d^2-1)(d^2 f-1)}
\end{align}
with multiplicity $d-1$.
The eigenvalue $\lambda_2$ is non-negative for $d\geq 2$, so that the operator in \eqref{eq:K-minus-isotropic} is positive semidefinite, and thus $K\geq \frac{1}{d^2}\rho_{d,f}$.

The discussion so far shows that the $K$ given in \eqref{eq:K} is feasible in \eqref{eq:dual}, hence giving the following upper bound on the fidelity for $f\geq 1/d$:
\begin{align}
	\F_{\pbtcl}(\rho_{d,f}) \leq \tr(K) = \frac{1}{d} + \frac{(1-df)^2}{d^2(d^2f-1)}. \label{eq:fidelity}
\end{align}

To show that we actually have equality in \eqref{eq:fidelity} for $f\geq 1/d$, we use the primal problem \eqref{eq:Fpbtcl-SDP-primal}.
To this end, we define constants
\begin{equation}
	a= \frac{(d-1)df}{d^2f-1} \q\q\q  b=1-a \q\q\q c= \frac{2+d^3f^2-d(1+2f)}{(d^2f-1)^2}
\end{equation}
and vectors $|\psi_i\rangle = -a |ii\rangle + b \sum_{j\neq i} |jj\rangle$ for $i=1,\dots,d$.
Now consider the POVM $\Pi=\lbrace \Pi^{(i)}\rbrace_{i=1}^{d+1}$ with elements
\begin{align}
    &\Pi^{(i)} = |\psi_i\rangle\langle\psi_i| + \left(\1_d - |i\rangle\langle i|\right)\otimes |i\rangle\langle i| \text{\q\q for $i=1,\dots,d$}\\
    &\Pi_{d+1} = c |\gamma\rangle\langle\gamma|.
\end{align}
These are manifestly positive semidefinite, and one can check that $\sum_i \Pi^{(i)} = \1_{d^2}$.
Hence, $\Pi$ constitutes a valid POVM.
We further compute, for $i=1,\dots,d$, 
\begin{equation} 
\tr(\Pi^{(i)}\sigma^{(i)}) = \frac{1}{d} (a^2+d-1)\q\q\q\q
\tr(\Pi^{(d+1)}\sigma^{(d+1)}) = cdf.
\end{equation}
This POVM then gives the following lower bound on the fidelity $\F_{\pbtcl}(\rho_{d,f})$ for $f\geq 1/d$:
\begin{equation}
    \F_{\pbtcl}(\rho_{d,f}) \geq \frac{1}{d^2} \sum_{i=1}^{d+1} \tr(\Pi^{(i)}\sigma^{(i)})
    = \frac{1}{d^2} (a^2 + d - 1 + c d f)
    = \frac{1 + d^3 f + d^2 f^2 - d (1 + 2 f)}{d^2 (d^2 f -1)},
\end{equation}
which can be shown to be equal to the RHS of \eqref{eq:fidelity}, and thus proves equality therein.

Hence, we have shown that
\begin{align}
    \F_{\pbtcl}(\rho_{d,f}) = \begin{cases} 
    \dfrac{1}{d} + \dfrac{(1-df)^2}{d^2(d^2f-1)} & \text{if $f\geq 1/d$}\\
    \dfrac{1}{d} & \text{otherwise},
    \end{cases}
\end{align}
which concludes the proof.
\end{proof}

\section{Computing $\F_{\pbtq}(\rho_{d,f})$ (Proof of ~\Cref{thm:iso-fidelity-q})}

\label{app:iso-q-optimization}

\begin{theorem*}[Restatement of Thm.~\ref{thm:iso-fidelity-q}]
The optimal PBT fidelity with quantum state preparation of the isotropic state $\rho_{d,f}$ is equal to
    \begin{equation}
         \label{Eq:PBT**-EPR}
         \F_{\pbtq}(\rho_{d,f})=
         \begin{cases}\displaystyle \frac{1}{d}+\frac{f}{d^2}-\frac{1}{d^3} & \text{\upshape if } f\geq 1/d\\[1em]
         \displaystyle\frac{1}{d} & \text{\upshape otherwise.}
         \end{cases}
    \end{equation}
\end{theorem*}

\begin{proof}
Similar to before, we restrict our attention to the regime $f\geq 1/d$, as the isotropic state $\rho_{d,f}$ is separable for $f\leq 1/d$ and hence the classical teleportation fidelity $1/d$ is the maximal achievable fidelity \cite{Horodecki_1999a}.

We recall that the $\pbtq$ fidelity is given by
\begin{align}
\F_{\pbtq}(\rho_{d,f})
&=\frac{1}{d^2}\!\!\max_{\substack{\{\Pi^{(i)}\}_{i=1}^{d^4+1}\\ \{\sigma^{(i)}\}_{i=1}^{d^4}}}\!\bigg\{\!\tr[\Pi_{AB}^{(1)} \rho_{d,f}] + \frac{1}{d} \sum_{i=1}^{d^4} \tr[\Pi_{AB}^{(1+i)}(\1\otimes\sigma^{(i)})]\bigg\}\\
&=\frac{1}{d^2}\!\!\max_{\{\Pi^{(i)}\}_{i=1}^{d^4+1}}\!\bigg\{\!\tr[\Pi_{AB}^{(1)} \rho_{d,f}] + \frac{1}{d}\sum_{i=1}^{d^4}\Vert\Pi_B^{(1+i)}\Vert_\infty\bigg\}.
\end{align}
where $\Vert X\Vert_\infty =\max_{\ket{\varphi}} \bra{\varphi}X\ket{\varphi} /\ip{\varphi}{\varphi}$ denotes the spectral norm of hermitian operator $X$.
For any unitary $U\in \text{U}(d)$, the $U\otimes \overline{U}$ invariance of $\rho_{d,f}$ and the unitary invariance of the spectral norm imply that
\begin{align}
    \tr[\Pi^{(1)} \rho_{d,f}] + \frac{1}{d} \sum_{i=1}^{d^4} \Vert\Pi_B^{(1+i)}\Vert_\infty
    = \tr[(U\otimes\overline{U})\Pi^{(1)} (U\otimes\overline{U})^\dagger\rho_{d,f}] + \frac{1}{d}\sum_{i=1}^{d^4}\Vert\Pi_B^{(1+i)}\Vert_\infty.
\end{align}
Hence, by averaging $\Pi^{(1)}$ over unitaries drawn from a $2$-design of $\text{U}(d)$ and keeping the rotated copies of the other POVM elements as separate outcomes, there is some finite $L$ so that
\begin{align}
\F_{\pbtq}(\rho_{d,f})=\frac{1}{d^2}\max_{\{\Pi^{(i)}_{AB}\}_{i=1}^{L+1}}\bigg\{\tr[\Pi_{AB}^{(1)} \rho_{d,f}] + 
\frac{1}{d}\sum_{i=1}^{L}\Vert\Pi_B^{(1+i)}\Vert_\infty\bigg\},
\end{align}
where $\Pi^{(1)}$ and $\Pi^{(0)}\coloneqq\1-\Pi^{(1)}$ are both $U\otimes\overline{U}$-invariant operators, thus having the isotropic form of Eq.~\eqref{Eq:isotropic-state} (up to normalization).
By writing $\Pi_{AB}^{(1)}= s\Phi^+_d+ r(\1_{AB}-\Phi^+_d)$ and substituting this into the previous equation, the optimization problem readily reduces to 
\begin{align}
\begin{aligned}
    \F_{\pbtq}(\rho_{d,f})=\frac{1}{d^2} \max\;\; \bigg\{ &sf+r(1-f)+\tfrac{1-s+(1-r)(d^2-1)}{d}\sum_{i=1}^{L}\Vert\Pi_B^{(1+i)}\Vert_\infty :\\
    \text{subject to}\;\;&\sum_{i=1}^{L} \Pi_{AB}^{(1+i)}=\rho_{d,t},\\
    &0\leq r,s\leq 1,\\
    &t=\frac{1-s}{1-s+(1-r)(d^2-1)},\\
    &\Pi_{AB}^{(1+i)}\geq 0.\bigg\}
\end{aligned}
\end{align}

For a given $t$, consider the optimization
\begin{align}
    \max
    &\left\{\sum_{i=1}^{L}\Vert\Pi_B^{(1+i)}\Vert_\infty  \;\bigg|\;\sum_{i=1}^{L} \Pi_{AB}^{(1+i)}=\rho_{d,t},\;\Pi^{(1+i)}_{AB}\geq 0,\;L\in\mbb{N}\right\}.   
\end{align}
First, since the value of $L$ can be taken as large as desired, without loss of generality we can assume that all the $\Pi_{AB}^{(1+i)}$ are rank-one positive operators, $\Pi_{AB}^{(1+i)}\coloneqq\op{\phi_i}{\phi_i}_{AB}$, which need not be normalized.
Then $\Vert \tr_{A}(\op{\phi_i}{\phi_i}_{AB})\Vert_{\infty}$ is the square of the largest Schmidt coefficient of $\ket{\phi_i}_{AB}$.
We can express this as 
\begin{equation}
    \Vert \tr_{A}(\op{\phi_i}{\phi_i}_{AB})\Vert_{\infty}=\ip{\phi_i}{\phi_i}-\mathsf{E}_{d-1}(\phi_i),
\end{equation}
where $\mathsf{E}_{d-1}(\phi_i)=\sum_{k=2}^{d}\lambda_k^\downarrow(\phi_i)$ is the so-called $(d-1)$-Ky-Fan entanglement measure of bipartite vector
\begin{align}
\ket{\phi_i} = \sum_{k=1}^d \sqrt{\lambda_k^\downarrow(\phi_i)} \ket{\alpha_k}\ket{\beta_k}
\end{align}
whose Schmidt coefficients are enumerated in non-increasing order, $\sqrt{\lambda_1^\downarrow(\phi_i)}\geq \sqrt{\lambda_2^\downarrow(\phi_i)} \geq\dots\geq \sqrt{\lambda_d^\downarrow(\phi_i)}$ \cite{Vidal-1999a}.
Then
\begin{align}
    \sum_{i}\Vert \tr_{A}(\op{\phi_i}{\phi_i}_{AB}) \Vert_{\infty}
    &=\sum_{i}\ip{\varphi_i}{\varphi_i}-\sum_{i}\mathsf{E}_{d-1}(\varphi_i)=1-\sum_{i}\mathsf{E}_{d-1}(\varphi_i).
\end{align}
When maximizing both sides over all rank one decompositions of $\rho_{d,t}$, we have
\begin{align}
    \max_{\sum_{i}\op{\varphi_i}{\varphi_i} = \rho_{d,t}}\Vert \tr_{A}(\op{\phi_i}{\phi_i}_{AB})\Vert_{\infty} = 1-\hat{\mathsf{E}}_{d-1}(\rho_{d,t}),
\end{align}
where 
\begin{equation}
    \hat{\mathsf{E}}_{d-1}(\rho_{d,t})=\min_{\sum_{i}\psi_i=\rho_{d,t}}\sum_i\mathsf{E}_{d-1}(\psi_i)
\end{equation}
is the convex-roof extension of $\mathsf{E}_{d-1}$ \cite{Vidal-2000a}.
Here, we again assume that the $\ket{\psi_i}$ are sub-normalized for convenience; since $p_i \mathsf{E}_{d-1}(\psi_i/p_i)=\mathsf{E}_{d-1}(\psi_i)$ for normalization factor $p_i=\ip{\psi_i}{\psi_i}$, we can interpret $\hat{\mathsf{E}}_{d-1}$ as the smallest average value of $\mathsf{E}_{d-1}$ among all ensembles that generate $\rho_{d,t}$.
The fidelity $\F_{\pbtq}(\rho_{d,f})$ can therefore be expressed as
\begin{align}
\begin{aligned}
     \F_{\pbtq}(\rho_{d,f})=\frac{1}{d^2}\max\bigg\{&sf+r(1-f)+\tfrac{1-s+(1-r)(d^2-1)}{d}\big(1-\hat{\mathsf{E}}_{d-1}(\rho_{d,t})\big),\\
    \text{subject to}\;\;&t=\frac{1-s}{1-s+(1-r)(d^2-1)}, \qquad 0\leq r,s\leq 1.\bigg\}
\end{aligned}
    \label{Eq:Fidelity**_reduced}
\end{align}
As computed in \Cref{lem:E-isotropic}, we have
\begin{align}
    1-\hat{\mathsf{E}}_{d-1}(\rho_{d,t})
    &=\begin{cases}
        1 & \text{\upshape for }  t\leq \frac{1}{d}\\
        1-t+\frac{-1+2(t+\sqrt{(d-1)(1-t)t})}{d} &\text{\upshape for }\ t>\frac{1}{d}.
    \end{cases}
\end{align}    
Theorem \ref{thm:iso-fidelity-q} is then readily proven by substituting the expression for $1-\hat{\mathsf{E}}_{d-1}(\rho_{d,t})$ into Eq.~\eqref{Eq:Fidelity**_reduced} and performing an optimization over $r$ and $s$ for the two cases of $t\leq \frac{1}{d}$ and $t>\frac{1}{d}$.
 
 We next carry out this somewhat tedious optimization. Our goal is to compute the following quantity:
\begin{align}
\begin{aligned}
    \F_{\pbtq}(\rho_{d,f})=\frac{1}{d^2}\max\;\; \bigg\{ &sf+r(1-f)+\frac{1-s+(1-r)(d^2-1)}{d}\times
    \begin{cases}
        1\;\;\text{\upshape for}\; t < \frac{1}{d}\\
        1-t+\frac{-1+2(t+\sqrt{(d-1)(1-t)t})}{d}\; \;\text{\upshape for}\; t\geq\frac{1}{d}.
    \end{cases}\\
    \text{subject to}\;\;&t=\frac{1-s}{1-s+(1-r)(d^2-1)}, \qquad 0\leq r,s\leq 1.\bigg\}
    \end{aligned}
\end{align}

\noindent\textbf{Case I:} $t\geq\frac{1}{d}$.
Then for $t=\frac{1-s}{1-s+(1-r)(d^2-1)}$, we have
\begin{align}
    \frac{1-s+(1-r)(d^2-1)}{d}\big(1-\hat{\mathsf{E}}_{d-1}(\rho_{d,t})\big)&=\frac{(1-r)(d^2-1)}{d}-\frac{1-s+(1-r)(d^2-1)}{d^2}\notag\\
    &\qquad +2\left(\frac{1-s+(d-1)\sqrt{(d+1)(1-r)(1-s)}}{d^2}\right)\\
    &\leq \frac{(d+1)(1-s)}{d}-\frac{1-s+(1-r)(d^2-1)}{d^2},
\end{align}
where we have used the inequality $(d-1)(1-s)\geq(1-r)(d^2-1)$.
Then by defining the region $\mc{R}=\{(s,r)\;|\; 0\leq s,r\leq 1, (1-s)\geq(1-r)(d+1)\}$, we obtain the bound
\begin{equation}
    \F_{\pbtq}(\rho_{d,f})\leq \max_{(s,r)\in\mc{R}}\;\; \frac{sf}{d^2}+\frac{r(1-f)}{d^2}+\frac{(d+1)(1-s)}{d^3}-\frac{1-s+(1-r)(d^2-1)}{d^4}.
\end{equation}
Since the RHS is linear in $s$ and $r$, the maximum is attained on the boundary of the region.
If $r=1$, then the RHS becomes
\begin{equation}
    \frac{sf}{d^2}+\frac{1-f}{d^2}+\frac{(d+1)(1-s)}{d^3}-\frac{1-s}{d^4}\leq \frac{2-f}{d^2}+\frac{1}{d^3}-\frac{1}{d^4}.
\end{equation}
On the other hand, if $(1-s)=(1-r)(d+1)$, then the RHS becomes
\begin{equation}
    \frac{sf}{d^2}+\frac{(d+s)(1-f)}{(d+1)d^2}+\frac{1-s}{d^2}\leq \frac{1}{d^2}+\frac{1-f}{d(d+1)}.
\end{equation}

\noindent\textbf{Case II:} $t<\frac{1}{d}$.
In this case, we consider the region $\mc{R}'=\{(s,r)\;|\; 0\leq s,r\leq 1, (1-s)\leq(1-r)(d+1)\}$, and the optimization simplifies to
\begin{align}
    \F_{\pbtq}(\rho_{d,f})&=\max_{(s,r)\in\mc{R}'}\;\; \frac{sf}{d^2}+\frac{r(1-f)}{d^2}+\frac{1-s+(1-r)(d^2-1)}{d^3}.
\end{align}
Again, it suffices to just consider the boundary.
When $s=1$, the RHS becomes
\begin{align}
    \frac{f+r(1-f)}{d^2}+\frac{(1-r)(d^2-1)}{d^3}\leq\frac{1}{d}+\frac{f}{d^2}-\frac{1}{d^3}.
\end{align}
When $s=0$, the RHS is
\begin{equation}
    \frac{r(1-f)}{d^2}+\frac{1+(1-r)(d^2-1)}{d^3}\leq\frac{1}{d}.
\end{equation}
The final condition to consider is when $(1-s)=(1-r)(d+1)$.
Then the RHS becomes
\begin{equation}
    \frac{sf}{d^2} +\frac{(d+s)(1-f)}{(d+1)d^2}+ \frac{1-s}{d^2}\leq \frac{1}{d^2} +\frac{1-f}{d(d+1)}.
\end{equation}
Comparing all the cases, we find that the maximum of $\frac{1}{d}+\frac{f}{d^2}-\frac{1}{d^3}$ is attained by the choice of $s=1$ and $r=0$.
This completes the proof of Theorem \ref{thm:iso-fidelity-q}.

\end{proof}

\section{Computing $\hat{\mathsf{E}}_{d-1}(\rho_{d,t})$ (Proof of Lemma \ref{lem:E-isotropic})}

\label{app:E-isotropic}

\begin{lemma*}[Restatement of Lem.~\ref{lem:E-isotropic}]
The entanglement measure $\hat{\mathsf{E}}_{d-1}$ defined in \eqref{eq:Ky-Fan-measure-mixed} satisfies the following on isotropic states $\rho_{d,f}$:
\begin{align}
    1-\hat{\mathsf{E}}_{d-1}(\rho_{d,f})
    &=\begin{cases}
        1 & \text{\upshape for }  f\leq \frac{1}{d}\\
        1-f+\frac{-1+2(f+\sqrt{(d-1)(1-f)f})}{d} &\text{\upshape for } f>\frac{1}{d}.
    \end{cases}
\end{align}    
\end{lemma*}

\begin{proof}
We follow the setup of Ref.~\cite{Terhal-2000a}.
Suppose that $\{\ket{\psi_i}\}_{i}$ is an optimal ensemble for $\rho_{d,f}$ in the definition of $\hat{\mathsf{E}}_{d-1}(\rho_{d,f})$.
Then 
\begin{align} 
\rho_{d,f}=\mc{T}_{U\otimes\overline{U}}(\rho_{d,f})=\sum_i \mc{T}_{U\otimes\overline{U}}(\psi_i).
\end{align}
Let $\ket{\psi_i}=U_i\otimes V_i\sum_{k}\sqrt{\mu_{k,i}}\ket{k}\otimes \ket{k}$ be a Schmidt decomposition with the $\mu_{k,i}$ labeled in non-increasing order.
Then
\begin{equation}
\label{Eq:lambda1}
    \lambda_i\coloneqq|\ip{\Phi_d^+}{\psi_i}|^2=\frac{1}{d}\left|\sum_{k=1}^d\bra{k}V_i^TU_i\ket{k}\sqrt{\mu_{k,i}}\right|^2\leq \frac{1}{d}\left|\sum_{k=1}^d\sqrt{\mu_{k,i}}\right|^2=:\lambda_i'.
\end{equation}
Thus, $\mc{T}_{U\otimes\overline{U}}(\psi_i)=\lambda_i\Phi^+_d+(\braket{\psi_i}{\psi_i}-\lambda_i)/(d^2-1)(\1-\Phi^+_d)$, and so we have
\begin{align}
    \rho_{d,f}=\left(\sum\nolimits_{i}\lambda_i\right)\Phi^+_d+\frac{1-\left(\sum_{i}\lambda_i\right)}{d^2-1}(\1\otimes\1-\Phi^+_d).
\end{align}
Notice that $f=\sum_{i}\lambda_i$.
If we were instead to use an ensemble $\{\ket{\psi_i'}\}_i$ which is the same as $\{\ket{\psi_i}\}$ except with $U_i=V_i=\1$ for every $i$, then $\sum_{i}\mc{T}_{U\otimes\overline{U}}(\psi_i')=\rho_{d,f'}$ with $f'=\sum_{i}\lambda_i'\geq f$, by Eq.~\eqref{Eq:lambda1}.
But we could always mix $\rho_{d,f'}$ with, say, the state $\mc{T}(\op{12}{12})$ to reduce the singlet fraction from $f'$ to $f$, since $\ip{\Phi^+_d}{12}=0$.
Furthermore, $\mathsf{E}_{d-1}(\op{12}{12})=0$, so this mixing will never increase the average value of $\mathsf{E}_{d-1}$ across the generating ensemble.
We therefore conclude that an optimal ensemble for $\rho_{d,f}$ consists of two subsets: the first has states $\ket{\psi_i}$ with a common Schmidt basis $\ket{\psi_i}=\sum_{k}\sqrt{\mu_{k,i}}\ket{kk}$, the second consists of product states having zero overlap with $\ket{\Phi^+_d}$.
Hence,
\begin{align}
    1-\hat{\mathsf{E}}_{d-1}(\rho_{d,f})=\max \left\{ 1-\sum_{i}\sum_{k=1}^d\mu_{k,i}+\sum_i \mu_{1,i} : \quad f=\frac{1}{d}\sum_{i}(\sum_k\sqrt{\mu_{k,i}})^2, \quad 1\geq\sum_{i}\sum_{k}\mu_{k,i} \right\}.
    \label{Eq:T-cons}
\end{align}
Note that the objective function in this optimization problem is equal to $1-\sum_i\sum_{k=2}^d\mu_{k,i}$.
We can write
\begin{align}
    \sum_{i}\left(\sum_k\sqrt{\mu_{k,i}}\right)^2=p\sum_i \frac{q_i}{p}\left(\sum_k\sqrt{\frac{\mu_{k,i}}{q_i}}\right)^2
\end{align}
where $q_i=\sum_{k=1}^d\mu_{k,i}$ and $p=\sum_iq_i$.
Since the function $f(\mbf{x})=(\sum_{k=1}^d\sqrt{x_k})^{2}$ is concave \cite{Boyd-2004a}, we have 
\begin{align}
\sum_{i} \left(\sum_k\sqrt{\mu_{k,i}}\right)^2 \leq \left( \sum_{k=1}^d \sqrt{\sum_{i}\mu_{k,i}} \right)^2.   
\end{align}
Thus, for any set of vectors $\ket{\psi_i}=\sum_{k}\sqrt{\mu_{k,i}}\ket{kk}$ satisfying the constraint of Eq.~\eqref{Eq:T-cons}, there exists a single vector, $\ket{\psi_0}=\sum_{k=1}^d\sqrt{\mu_{k,0}}\ket{kk}$ with $\mu_{k,0}=\sum_i\mu_{k,i}$, such that
\begin{subequations}
\begin{align}
    &1-\sum_{k=1}^d\mu_{k,0}+\mu_{1,0}=1-\sum_i\sum_{k=1}^d\mu_{k,i}+\sum_i\mu_{1,i}\label{Eq:single-vec1}\\
    &f\leq \frac{1}{d}\left(\sum_{k=1}^d\sqrt{\mu_{k,0}}\right)^2\label{Eq:single-vec2}
\end{align}
\end{subequations}
We can scale down the vector components $\mu_{k,0}$ to make the inequality in Eq.~\eqref{Eq:single-vec2} tight, which only increases the LHS of Eq.~\eqref{Eq:single-vec1}.
Therefore, the maximum in Eq.~\eqref{Eq:T-cons} can be attained by a single vector $\ket{\psi_0}$. Furthermore, by another application of concavity, we have
\begin{align}
    \frac{1}{d}\left(\sum_{k=1}^d\sqrt{\mu_{k,0}}\right)^2\leq \frac{1}{d}\left(\sqrt{\mu_{1,0}}+(d-1)\sqrt{\sum_{k=2}^{d}\mu_{k,0}/(d-1)}\right)^2,
\end{align}
which implies we can restrict attention to vectors of the form $\ket{\psi_0}=\sqrt{\mu_{1}}\ket{11}+\sqrt{\nu}\sum_{k=2}^d\ket{kk}$ with $\mu_1\geq\nu$.
Consequently, we have
\begin{align}
    1-\hat{\mathsf{E}}_{d-1}(\rho_{d,f}) = \max \left\{ 1-(d-1)\nu : f=\frac{1}{d}(\sqrt{\mu_1}+(d-1)\sqrt{\nu})^2, \quad 1\geq\mu_1+(d-1)\nu. \right\}
    \label{Eq:T-cons2}
\end{align}
For the case when $f\leq \frac{1}{d}$, we have $1-\hat{\mathsf{E}}_{d-1}(\rho_{d,f})=1$, which is attained by taking $\nu=0$  and $\mu_1=fd $.
On the other hand, consider when $f>\frac{1}{d}$.
Then, it's clear that the maximum is obtained when the inequality in Eq.~\eqref{Eq:T-cons2} is tight, $1=\mu_1+(d-1)\nu$.
Indeed, if it is a strict inequality, then we can always decrease $\nu$ and increase $\mu_1$ such that the objective function increases while still satisfying the second equality for $f$.
We can therefore solve for $\mu_1$ and the objective function becomes $1-(d-1)\nu=\mu_1$, where
\begin{align}
    \mu_1&=1-f+\frac{-1+2(f+\sqrt{(d-1)(1-f)f})}{d}.
\end{align}
This concludes the proof.
\end{proof}

\section{Background on graph states}
\label{app:graph_states} 

Given a simple graph $\Gamma=(V,E)$ with the node set $V$ and the edge set $E$, the \emph{graph state} associated with $\Gamma$ is defined as
\begin{equation}\label{eq:graph-state}
\ket{\Gamma} = \prod_{uv\in E} CZ_{uv} \ket{+}^{\otimes |V|}
\end{equation}
where $CZ_{uv}$ is the control-$Z$ unitary that acts on the $u,v$-th tensor factors as
\begin{align}
CZ_{uv}\ket{+}_u\ket{+}_v = \frac{1}{\sqrt{2}}(\ket{0}_u\ket{+}_v + \ket{1}_u\ket{-}_v) =: \ket{\Psi^+}_{uv}.
\end{align}
Let $|V| = n$.
Another way to define the graph state $\ket{\Gamma}$ is with stabilizer formalism.
Define $n$-th Pauli group $P_n$ as
\begin{align}
\begin{split}
  &P_n \coloneqq \{ A_1 \otimes \dots \otimes A_n \mid A_i \in P, i \in [n] \},\\
\end{split}
\end{align}
where $P$ is the Pauli group $\{\pm 1, \pm i\} \times \{ X, Y, Z, I \}$.
Given a simple graph $\Gamma(V,E)$, the \emph{(graph) stabilizer group} $G$ associated with $\Gamma$ is defined as
\begin{align}
G \coloneqq \langle g_i : i\in V \rangle,\q
g_i \coloneqq X_i Z_{N(i)},
\end{align}
where $T_i \coloneqq I \otimes \dots\otimes T\otimes\dots \otimes I$ has the Pauli operator $T \in \{X,Y,Z,I\}$ in the $i$-th tensor factor, $Z_{W}\coloneqq\prod_{j\in W} Z_j$ for $W\subseteq V$, and $N(i)$ denotes the set of neighbors of $i$ in $\Gamma$.
A \emph{stabilizer group} is defined as an abelian subgroup of $P_k$ for some $k$ that does not contain $-I$, and therefore $G$ is indeed a stabilizer group, since its generators commute and $-I\notin G$ \cite{Hein_2006}.
A graph state $\ket{\Gamma}$ is a stabilizer state corresponding to the graph stabilizer group $G$ associated with graph $\Gamma$, i.e.
\begin{align}
  \Gamma \equiv \pure{\Gamma} \coloneqq \frac{1}{|G|}\sum_{g\in G} g.
\end{align}

\begin{lemma}{(Reduced graph state)} \cite{Hein_2006} \label{L:reduced-graph}
  For a simple graph $\Gamma = (V,E)$, let $S\subset V$, $T = V\bs S$, and $\Gamma - T$ the subgraph of $\Gamma$ obtained by removing nodes $T$.
      Let $\ket{\Gamma}$ denote the associated graph state and $\Gamma_S\coloneqq\tr_T\op{\Gamma}{\Gamma}$ its reduced state.
  Then,
  \begin{enumerate}[(i)]
    \item $\ds \Gamma_S = \frac{1}{2^{|S|}}\sum_{h\in G_S} h,$
    where $G_S$ is the subgroup of stabilizers for $\ket{\Gamma}$ that act trivially on $V\setminus S$.
    \item We have
    \begin{align}\label{eq:reduced-subset}
    \Gamma_S &= \frac{1}{2^{|T|}} \sum_{\mbf{b}\in\mbb{Z}_2^{|T|}} \prod_{j\in T} Z_{N(j)\cap S}^{b_j} \pure{\Gamma- T} Z_{N(j)\cap S}^{b_j} ,\\
    &=\frac{1}{|\mc{N}(T)|}\sum_{\mbf{b}\in \mc{N}(T)}Z_{f(\mbf{b })}\pure{\Gamma- T}Z_{f(\mbf{b })},
    \end{align}
    where we have defined the natural bijection $f:\mbb{Z}_2^S\leftrightarrow 2^S$ such that node $j\in S$ belongs to $f(\mbf{b})\subset S$ iff $b_j=1$, and $\mc{N}(T)$ is the $\mbb{Z}_2$-linear subspace of $\mbb{Z}_2^{|S|}$ defined by $\mc{N}(T)\coloneqq\Span_{\mbb{Z}_2}\{\mbf{b}\in\mbb{Z}_2^{|S|}\;|\; f(\mbf{b})=N(j)\cap S,\; j\in T\}$.
       \item $\frac{2^{|S|}}{|G_{S}|}\Gamma_{S}$ is a projection with the rank  $\frac{2^{|S|}}{|G_S|}$. \label{it:reduced-3} 
  \end{enumerate}
\end{lemma}
\begin{proof}
    See \cite{Hein_2006} for the proof of (i) and the proof that $\Gamma_S$ is a multiple of a projection.
    Recall that $\{Z_{S}\ket{\Gamma}\mid S\subset V\}$ forms an orthonormal basis for any graph state $\Gamma$.
      From this, it follows that the rank of $\Gamma_S$ is $2^{|S|}/|G_S|$, which proves (iii).
      Finally, (ii) is obtained by a direct computation of the partial trace using the form of $\ket{\Gamma}$ given in Eq.~\eqref{eq:graph-state}.
\end{proof}

\section{The reduced states $\tr_{B_i^c}\op{\Gamma}{\Gamma}$ (Proof of Lemma \ref{L:reduced-pbt-ens})}

\label{app:graph_states_lemma-reduced}

\noindent\textbf{Lemma \ref{L:reduced-pbt-ens}.} (Restated)
\textit{
  Let $\Gamma \in \cG_{r,n}$ be in PBT canonical form with $k$-EPR pairs.
  Then the reduced states $\Gamma^{(i)}_{A^r B}\coloneqq\tr_{B_i^c}\op{\Gamma}{\Gamma}$ are
  \begin{align}
     \Gamma^{(i)}_{A^rB}=\begin{cases}
      \ds \Psi^+_{A_iB} \otimes \wh{\1}_{A^r\setminus A_i} & \text{if }i \leq k,\\
       \\
      \ds   E^{(i)}(\wh{\1}_{A^r} \otimes \pure{+}_{B})E^{(i)} & \text{otherwise.}
    \end{cases}
  \end{align}
  where $\wh{\1}$ denotes the maximally mixed state (i.e. the normalized identity), $\ket{\Psi^+}\coloneqq\frac{1}{\sqrt{2}}(\ket{0+}+\ket{1-})$, and $E^{(i)}$ is the unitary operator
  \begin{equation}\label{eq:edge-op}
    E^{(i)} = \prod_{A_j\in N(B_i)}CZ_{A_jB}.
  \end{equation}
}

\begin{proof}
    Introduce the local dephasing map acting on system $A_i$ as
\begin{align} 
\Delta_{A_i}(\rho)\coloneqq\frac{1}{2}\sum_{x=0}^1Z^x_{A_i}\rho Z^x_{A_i}.
\end{align}
A crucial property of PCF is that each $B_i$ is only connected to the node $A_i$ for $i \leq r$.
Hence, by \Cref{L:reduced-graph}, tracing out any $B_j$ for $j=1,\dots,r$ causes a complete dephasing of system $A_j$ in the graph state $\ket{\Gamma-B_j}$.
On the other hand, for $j>r$, the effect of tracing out $B_j$ is the collective map
\begin{align}\Delta_{N(B_j)}(\rho)\coloneqq\frac{1}{2}\sum_{x=0}^1Z^x_{N(B_j)}\rho Z^x_{N(B_j)}\end{align}
applied to $\ket{\Gamma-B_j}$.

Suppose that $i\leq k$.
This means that $A_i$ and $B_i$ share an EPR state that is disconnected from the rest of the graph.
Hence, tracing out $B_j$ for all $j\not=i$ will cause complete dephasing on systems $A^r\setminus A_i$.
Since $\ket{\Gamma-B_i^c}=\ket{\Psi^+}_{A_iB}\ket{\wh{+}}_{A^r\setminus A_i}$, the reduced state is thus given by $\Psi^+_{A_iB_i}\otimes\wh{\1}_{A^r\setminus A_i}$.

On the other hand, suppose that $i>k$.
If $i>r$, every $B_j$ with $j\leq r$ is traced out, so every system of $A^r$ is completely dephased. Now let $k<i\leq r$.
Then it must be that $A_i$ is connected to both $B_i$ and some $B_l$ with $l\not=i$ and $l>r$.
Tracing out $B_l$ then applies the map $\Delta_{N(B_l)}$ to all systems in $N(B_l)$.
But note that all systems $A_j\in N(B_l)$ other than $A_i$ are already independently dephased due to the tracing out of $B^n\setminus B_i$.
Hence, the effect of $\Delta_{N(B_l)}$ on system $A_i$ is also an independent dephasing map.
In both cases, the end result is a complete dephasing of systems $A^{r}$ in the graph state $\ket{\Gamma-B^n\setminus B_i}=E^{(i)}\ket{\wh{+}}_{A^r}\ket{+}_B$.
This is precisely the state $E^{(i)}(\wh{\1}_{A^r}\otimes\op{+}{+}_{B_i})E^{(i)}$.
\end{proof}

\section{Computing $\F_{\pbt}(\Gamma)$ with no isolated EPR pairs (Proof of \Cref{Prop:PBT-connected})}

\label{Appendix-proof_of_graph_no-epr}

\noindent \textbf{Proposition \ref{Prop:PBT-connected}.} (Restated)
\textit{
    Suppose that $\Gamma \in \cG_{r,n}$ is a PCF bipartite graph with $r=\log_2\rk(\Gamma_{B^n})$ and no independent EPR pairs. Then
  \begin{align}
    &\F_{\pbt}(\Gamma) =\begin{cases} \frac{1}{2} - \frac{1}{2^{r+2}}\quad \text{if $\exists B_j$ w. $|N(B_j)|$ being even},\\ \frac{1}{2} - \frac{1}{2^{r+1}}\quad\text{otherwise},
    \end{cases}\\
    &\F_{\pbtq}(\Gamma)= \frac{1}{2}.
  \end{align}
}

\begin{proof}
The PBT fidelity is given by $\frac{1}{4}\sum_{i=1}^n\tr[\Gamma^{(i)}_{A^rB}\Pi^{(i)}_{A^rB}]$, maximized over all POVMs $\{\Pi^{(i)}_{A^rB}\}_{i=1}^n$.
From Lemma \ref{L:reduced-pbt-ens}, since $\Gamma$ contains no isolated EPR pairs its reduced states $\Gamma^{(i)}$'s are given by
\begin{equation}
    \Gamma^{(i)}_{A^rB}= E^{(i)}(\wh{\1}_{A^r} \otimes \pure{+}_{B})E^{(i)}\qquad\forall i.
\end{equation}
As noted previously, each of the $\Gamma^{(i)}_{A^rB}$ is separable and hence $\Gamma^{(i)} \leq \Tr_{B}\Gamma^{(i)} \otimes \1$ by the reduction criterion.
    Therefore, by Corollary \ref{Cor:reduction}, we have $\F_{\pbtq}(\Gamma)=\frac{1}{2}$.

We next proceed to show that the bound is strictly less than $\frac{1}{2}$ when restricted to standard PBT protocols. 
Observe that (i) $\Gamma^{(i)}_{A^rB}=\frac{1}{2^r}P_i$ with $P_i$ being a projector, and (ii) $\{\Gamma^{(i)}_{A^rB}\}_{i=1}^n$ forms a pairwise commuting set.
Property (i) follows from the fact that $(E^{(i)})^2=\1$, while (ii) follows from the fact that every $\Gamma^{(i)}_{A^rB}$ is diagonal in the basis $\{\ket{s}_{A^r}\ket{\pm}_B\}$, because $E^{(i)}$ maps $\ket{s}\ket{+}$ to $\ket{s}\ket{+}$ or $\ket{s}\ket{-}$ according to the parity of $s$ on $N(B_i)$.
Hence, the problem of maximizing $\sum_{i=1}^n\tr[\Gamma^{(i)}_{A^rB}\Pi^{(i)}_{A^rB}]$ is equivalent to the problem of discriminating classical random variables, for which maximum likelihood estimation is known to be the optimal method \cite[\S 7.3]{Boyd-2004a}.
Explicitly, let $\{\ket{k}\}_{k\in[2^{r+1}]}$ be an orthonormal basis for $A^rB$ that simultaneously diagonalizes all the $\Gamma^{(i)}_{A^rB}$.
For each $i$ we can thus write 
\begin{align}\Gamma^{(i)}_{A^rB}=\frac{1}{2^r}P_i=\frac{1}{2^r}\sum_{k\in[2^{r+1}]}t_{k|i}\op{k}{k}\end{align} with $t_{k|i}=1$ for exactly $2^r$ values of $k$ and $t_{k|i}=0$ otherwise.
The optimal POVM involves projecting into the basis $\{\ket{k}\}_{k\in[2^{r+1}]}$, and then for outcome $k$ guessing any state $\Gamma^{(i)}$ such that
\begin{align}\bra{k}\Gamma^{(i)}_{A^rB}\ket{k}=\max_{j\in[n]}\{\bra{k}\Gamma^{(j)}_{A^rB}\ket{k}\}.\end{align}
Let $\{\Pi_{A^rB}^{(i)}\}_{i=1}^n$ denote the POVM for this guessing strategy.
Since  $\bra{k}\Gamma^{(i)}_{A^rB}\ket{k}=\frac{1}{2^r}$ unless $\ket{k}\not\in\supp(\Gamma^{(i)}_{A^rB})$, it follows that
\begin{align}\sum_{i=1}^n\tr[\Gamma^{(i)}_{A^rB}\Pi^{(i)}_{A^rB}]=\frac{1}{2^r}\rank\bigg(\sum_{i=1}^n \Gamma^{(i)}_{A^rB}\bigg).
\end{align}

Our goal is then to compute  $\rank\bigg(\sum_{i=1}^n \Gamma^{(i)}_{A^rB}\bigg)$.
    Consider first $i\leq r$.
    Since $B_i$ is only connected to $A_i$, we have
\begin{align}
    \Gamma^{(i)}_{A^rB}=\frac{1}{2}(&\op{0+}{0+}+\op{1-}{1-})_{A_iB}\otimes\wh{\1}_{A_i^c}. 
\end{align}
Close inspection shows that
\begin{align}
    \supp \left(\sum_{i=1}^r \Gamma^{(i)}_{A^rB}\right) &= \Span_{\bC} \{ S_1 \setminus S_2 \}\\
\text{where}\;\;    S_1 &= \{\ket{s}_{A^r}\ket{\pm}_B\mid s\in \{0,1\}^r\}\\
    S_2 &= \{\ket{0}^{\otimes r}_{A^r}\ket{-}_B,\;\; \ket{1}^{\otimes r}_{A^r}\ket{+}_B\}.
  \end{align}
On the other hand, the support of $\sum_{i=r+1}^n\Gamma^{(i)}_{A^rB}$ never includes $\ket{0}^{\otimes r}_{A^r}\ket{-}_B$, and it includes $\ket{1}^{\otimes r}_{A^r}\ket{+}_B$ when there exists some $j$ such that $|N(B_j)|$ is even.
    Hence,
  \begin{align}
  \begin{split}
    \rank\bigg(\sum_{i=1}^n \Gamma^{(i)}_{A^rB}\bigg) &= \begin{cases}
      \ds 2^{r+1}-1 & \text{if } |N(B_j)| \text{ is even for some $j$}\\
      \ds 2^{r+1}-2 & \text{otherwise}
    \end{cases}
  \end{split}
  \end{align}
Since $\F_{\pbt}(\Gamma)=\frac{1}{4}\sum_{i=1}^n\tr[\Gamma^{(i)}_{A^rB}\Pi^{(i)}_{A^rB}]$, this completes the proof.

\end{proof}

\section{Isolated EPR pairs essentially determine teleportation power (Proof of \Cref{T:graph-std-bound})}
\label{app:graph-std-bound}

\begin{theorem*}[Restatement of Thm.~\ref{T:graph-std-bound}]
  Suppose $\Gamma \in \cG_{r,n}$ has PCF with $k$ total isolated EPR pairs; i.e. $\Gamma=\Gamma_{\EPR{k}}\otimes\Gamma_0$.
    Then
  \begin{align}
    \Gamma_{\EPR{k}} \preceq_{\Omega} \Gamma \preceq_{\Omega} \Gamma_{\EPR{k}} \otimes \pure{+}_{B_{k+1}} \otimes \pure{-}_{B_{k+2}}
  \end{align}
  for $\Omega\in\{\pbt,\pbtcl\}$. Furthermore, $\Gamma\sim_{\pbtq}\Gamma_{\EPR{k}}$.
\end{theorem*}
\begin{proof}
The inequality $\F_{\Omega}(\Gamma_{\EPR{k}}) \leq \F_{\Omega}(\Gamma)$ for $\Omega\in\{\pbt,\pbtcl,\pbtq\}$ is trivial since Alice and Bob can always just restrict to their sub-graph $\Gamma_{\EPR{k}}$ when using $\Gamma$.

We next establish that
\begin{align}
    \F_{\Omega}(\Gamma)\leq \F_{\Omega}(\Gamma_{\EPR{k}}\otimes \pure{+}_{B_{k+1}} \otimes \pure{-}_{B_{k+2}})\label{Eq:graph-fidelity-upb}
\end{align}
for $\Omega\in\{\pbt,\pbtcl,\pbtq\}$.
    By Eq.~\eqref{Eq:PBT**=PBT}, this immediately implies that $\F_{\pbtq}(\Gamma)\leq \F_{\pbtq}(\Gamma_{\EPR{k}})$, and so $\F_{\pbtq}(\Gamma)=\F_{\pbtq}(\Gamma_{\EPR{k}})$.

To prove Eq.~\eqref{Eq:graph-fidelity-upb}, suppose that $\Gamma_0$ is an $(m_0,n_0)$ bipartite graph. 
Since $\Gamma_0$ contains no isolated EPR pairs, we have by Lemma \ref{L:reduced-pbt-ens} that the reduced states take the form 
\begin{align}
 \Gamma^{(j)}_{A^{m_0}B}= E^{(j)}(\wh{\1}_{A^{m_0}} \otimes \pure{+}_{B})E^{(j)}\quad\forall j\in[n_0]. 
\end{align}
We see that for any $i\in[m_0]$, system $A_i$ is completely dephased in each of the reduced states $\Gamma^{(j)}_{A^{m_0}B}$ for $j\in[n_0]$.
    From this it follows that $\pure{\Gamma_{\EPR{k}}}\otimes\pure{\Gamma_0}$ and $\pure{\Gamma_{\EPR{k}}}\otimes\Delta_{A^{m_0}}(\pure{\Gamma_0})$ will have the same PBT fidelities since they share the same reduced density matrices, where $\Delta_{A^{m_0}}(\pure{\Gamma_0})$ denotes the state obtained by completely dephasing Alice's systems in state $\Gamma_0$.
Because $\ket{\Gamma_0}$ is a bipartite graph, completely dephasing it on Alice's side will transform it into a convex combination of product states, 
\begin{align} 
\rho\coloneqq\Delta_{A^{m_0}}(\pure{\Gamma_0})=\sum_{\lambda}p(\lambda)\op{\alpha_\lambda}{\alpha_\lambda}\otimes\op{\beta_\lambda}{\beta_\lambda},
\end{align} 
where $\ket{\alpha_\lambda}\in\{\ket{0},\ket{1}\}^{\otimes m_0}$ and $\ket{\beta_\lambda}\in\{\ket{+},\ket{-}\}^{\otimes n_0}$.
    Therefore,
\begin{align}
    \F_{\Omega}(\Gamma)&=\F_\Omega(\Gamma_{\EPR{k}}\otimes \rho)\\
    &\leq\sum_\lambda p(\lambda)\F_{\Omega}(\Gamma_{\EPR{k}}\otimes\op{\alpha_\lambda,\beta_\lambda}{\alpha_\lambda,\beta_\lambda})\\
    &\leq \F_{\Omega}(\Gamma_{\EPR{k}}\otimes \pure{+}_{B_{k+1}} \otimes \pure{-}_{B_{k+2}}).
\end{align}
The first inequality follows from convexity in the fidelity measure $\F_{\Omega}$.
    The second inequality follows in two steps.
    Since $\F_{\Omega}$ remains invariant under local isometries by Alice, she can remove her state $\ket{\alpha_\lambda}$ so that
\begin{align}
    \F_{\Omega}(\Gamma_{\EPR{k}}\otimes\op{\alpha_\lambda,\beta_\lambda}{\alpha_\lambda,\beta_\lambda})=\F_{\Omega}(\Gamma_{\EPR{k}}\otimes\op{\beta_\lambda}{\beta_\lambda}). 
\end{align}
Then since the single-qubit marginals of $\ket{\beta_\lambda}$ are either $\ket{+}$ or $\ket{-}$, we have
\begin{align} 
\F_{\Omega}(\Gamma_{\EPR{k}}\otimes\op{\beta_\lambda}{\beta_\lambda})\leq \F_{\Omega}(\Gamma_{\EPR{k}}\otimes \pure{+} \otimes \pure{-}),
\end{align}
which finishes the proof.
\end{proof}

\section{Preliminaries on representation theoretic methods for calculating $\F_{\pbt}(\Phi^{+\otimes n})$}
\label{Appendix:schur-weyl}
Standard PBT problems with EPR pairs exhibit group symmetries and allow for the use of a variant of ``Schur--Weyl duality'' when calculating its optimal fidelity.
The following facts on representation theory and Schur--Weyl duality can be found in standard texts, such as \cite{Fulton2004} or \cite{Etingof2011}.

Let $V = \bC^d$, and $\{\ket{i}\}_{i=1}^d$ be the canonical basis of $V$.
$S_n$ denotes symmetric group of $n$ elements and $\cU(d)$ denotes the group of $d \times d$ unitary matrices.
Irreducible representations (irreps) of $S_n$ and $\cU(d)$, denoted $V_{\lambda}$ and $W_{d,\lambda}$ respectively, are labeled by a Young diagram $\lambda$ with $n$ boxes.
$\lambda$ represents the highest weight vector associated to the irrep $W_{d,\lambda}$, and $W_{d,\lambda}$ is zero if and only if $\lambda$ has more than $d$ rows.

Consider the following representations $\varphi:S_n\to \End(V^{\otimes n})$ and $\psi:U(d) \to \End(V^{\otimes n})$ given by their actions on the basis elements in $V^{\otimes n}$ for $\pi \in S_n$ and $U \in \cU(d)$:
\begin{align}
  &\varphi(\pi)(\ket{{i_1}}\otimes \dots \otimes \ket{{i_n}}) = \ket{i_{\pi^{-1}(1)}} \otimes \dots \otimes \ket{i_{\pi^{-1}(n)}}\\
  &\psi(U)(\ket{{i_1}} \otimes \dots \otimes \ket{{i_n}}) = U\ket{{i_1}} \otimes \dots \otimes U\ket{{i_n}}.
\end{align}
That is, $\varphi(\pi)$ moves the vector in the $k$-th tensor factor to the $\pi(k)$-th factor.
The \emph{Schur--Weyl duality} says that the algebras generated by $\varphi(S_n)$ and $\psi(\cU(d))$ are commutants of each other, and hence admits a decomposition
\begin{align}
  (\bC^{d})^{\otimes n} \iso \bigoplus_{\lambda \vdash_d n} V_{\lambda} \otimes W_{d,\lambda},
\end{align}
where $\lambda = (\lambda_1,\dots,\lambda_k) \vdash_d n$ denotes the Young diagram $\lambda$ corresponding to a partition of $n$ with at most $d$ rows.
We write $d_\lambda$ as the dimension of $V_\lambda$ and $m_{d,\lambda}$ as the dimension of $W_{d,\lambda}$.
The dimension $d_{\lambda}$ corresponds to the number of standard Young tableaux associated to the Young frame $\lambda$ and given by the \emph{Hook length formula}
\begin{align}
  d_{\lambda} = \frac{n!}{\prod_{i,j} h(i,j)},
\end{align}
where $h(i,j)$ is the hook length at cell $(i,j)$ in $\lambda$.
For $\lambda \vdash_d n$, the dimension $m_{d,\lambda}$ is the number of semi-standard Young tableaux associated to $\lambda$ given by the \emph{Weyl character formula}
\begin{align}
  m_{d,\lambda} = \prod_{1\leq i < j \leq d} \frac{\lambda_i - \lambda_j + j - i}{j-i}.
\end{align}
In particular, when $d = 2$ with $\lambda = (\lambda_1, \lambda_2)$, we have
\begin{equation}
    m_{2,\lambda} = \lambda_1 - \lambda_2 + 1, \q\q\q
    d_\lambda = \frac{n!(\lambda_1 - \lambda_2 + 1)}{\lambda_2! (\lambda_1+1)!} = \frac{\lambda_1-\lambda_2+1}{n+1}\binom{n+1}{\lambda_2}.
\label{eq:dim-irreps}
\end{equation}

Consider a marginal state $\wt{\rho}_{B^n}$ that commutes with $U^{\otimes n}$ and $\pi$ for any $U\in \cU(2)$ and any $\pi \in S_n$.
By Schur--Weyl duality and Schur's lemma, $\wt{\rho}_{B^n}$ can be expressed as
\begin{equation}\label{eq:sym-rho}
  \wt{\rho}
  = \sum_{\lambda \vdash_2 n} c_\lambda \1_{V_\lambda} \otimes \1_{W_{2,\lambda}}, \q
  c_\lambda = \frac{\Tr[\wt{\rho}P_\lambda]}{m_\lambda d_\lambda},
\end{equation}
where $P_\lambda$ is the isotypic projection onto the $S_n$ irrep $V_{\lambda}$, and $m_\lambda\equiv m_{2,\lambda}$.
It is given in terms of the character $\chi_\lambda$ of $V_\lambda$ by
\begin{equation}\label{eq:young-symmetrizer}
  P_\lambda = \frac{d_\lambda}{n!}\sum_{\sigma\in S_n}\chi_\lambda(\sigma)\,\varphi(\sigma),
\end{equation}
whereas the normalized sum $\frac{d_\lambda}{n!}\sum_{T}\varphi(e_T)$ of Young symmetrizers over all standard Young tableaux of shape $\lambda$ is in general not $P_\lambda$ (the first failures occur for $n=5$ and $\lambda=(3,2)$). Here $e_T$ is given by the formula
\begin{align}
  e_T \coloneqq r_T\cdot c_T, \q r_T = \sum_{\pi \in \cR_T} \pi, \q c_T = \sum_{\sigma\in \cC_T} \text{sgn}(\sigma) \sigma,
\end{align}
where $\cR_T, \cC_T \subset S_n$ are subgroups such that elements of $\cR_T$ permute within each row of $T$ and those of $\cC_T$ permute within each column of $T$. 
The canonical purification of $\wt{\rho}_{B^n}$ is
\begin{equation}
  \ket{\wt{\rho}}_{A^nB^n}
  = \sqrt{2^n}(\1_{A^n}\otimes \wt{\rho}^{\frac{1}{2}}_{B^n})\ket{\Phi^+}_{A^nB^n}
  = \sum_{\lambda\vdash_2 n} \sqrt{c_\lambda}\sqrt{2^n}(\1_{A^n} \otimes  P_\lambda)\ket{\Phi^+}_{A^nB^n}\label{eq:purification},
\end{equation}
and the optimal fidelity of $\F_{\pbt}(\pure{\wt{\rho}})$ is achieved by the \emph{pretty good measurement} (PGM) $\{M^{(i)}\}_{i=1}^n$ associated with the ensemble $\{\rho^{(i)} \coloneqq \Phi^+_{A_iB} \otimes \wh{\1}_{A_i^c}\}_i$ defined as
\begin{align}
  M^{(i)} \coloneqq \bar{\rho}^{-\frac{1}{2}} \rho^{(i)} \bar{\rho}^{-\frac{1}{2}}, \q
  \bar{\rho} = \sum_{i=1}^n \rho^{(i)}.
\end{align}

The dual Pieri formula gives the following decomposition of $S_n \times \cU(d)$-representations  as
\begin{align}
  (\bC^d)^{\otimes n} \otimes (\bC^d)^* \iso \bigoplus_{\lambda \vdash_d n} \bigoplus_{\alpha=\lambda-e_i} V_\lambda \otimes W_{d,\alpha},
\end{align}
where $U\in \cU(d)$ acts by $U^{\otimes n} \otimes \bar{U}$, and $\pi \in S_n$ acts by $\pi \otimes \1_{V^*}$.
When $\alpha\vdash_d n-1$ is a genuine partition, grouping the terms by $\alpha$ gives the subrepresentation $\bigoplus_{\alpha\vdash_d n-1}\bigoplus_{\mu=\alpha+\square}V_\mu\otimes W_{d,\alpha}$.
The remaining summands, those with $\alpha_d=-1$, are orthogonal to every vector that pairs $(\bC^d)^*$ maximally entangled with one of the $n$ factors, and hence play no role here.
We write $\mu = \alpha+ \square$ to mean the direct sum is over all Young diagram that is obtained by adding a box to the Young diagram $\alpha$.
The sum $\sum_i M^{(i)}\rho^{(i)}$ commutes with $U^{\otimes n} \otimes \bar{U}$ and $\pi \otimes \1$ for all $U \in \cU(d)$ and $\pi\in S_n$, and the entanglement fidelity with the PGM $M$ yields
\begin{equation} 
  \F_{\pbt}(\pure{\wt{\rho}}) = \F_{\pbt}(\pure{\wt{\rho}},M) = \frac{1}{4} \sum_{\alpha \vdash_2 n-1} \left(\sum_{\mu = \alpha+\square} \sqrt{c_\mu m_\mu d_\mu}\right)^2.\label{eq:sym-optimal}
\end{equation}
The optimality of \eqref{eq:sym-optimal} is proved with the dual feasibility of the associated SDP in \cite{Leditzky_2020}.

\section{An upper bound on $\F_{\pbtcl}(\Gamma_{\EPR{k}})$} 
\label{app:Fpbt-upper-bound}
In this section, we prove the following upper bound:
\begin{align}
    \F_{\pbtcl}(\Gamma_{\EPR{k}}) &\leq \F_{\pbt}(\Phi^{+\otimes k}) + \frac{1}{2^{k+2}}\Tr[\1 - \wt{M}]
    = \F_{\pbt}(\Phi^{+ \otimes k}) + \frac{k+2}{2^{k+2}},
\end{align}
where $\wt{M} = \sum_{i=1}^k M^{(i)}$ is the sum of the PGM $M = \{M^{(i)}\}_{i=1}^k$ associated with the PBT states of $k$-EPR pairs $\{\sigma^{(i)} = \Phi^+_{A_iC} \otimes \wh{\1}_{A_i^c}\}_{i=1}^k$.
Since $\F_{\pbtcl}$ is invariant under unitaries on Alice's nodes, we consider $\ket{\Phi^+}^{\otimes k}$ instead of $\ket{\Gamma_{\EPR{k}}} = H^{\otimes k}_{A^k} \ket{\Phi^+}^{\otimes k}$.
Let $\{\rho^{(i)}\}_{i=1}^{k+2}$ be the PBT states for $\rho_{A^kB^{k+2}} = \Phi^{+\otimes k} _{A^kB^k} \otimes \pure{0}_{B_{k+1}} \otimes \pure{1}_{B_{k+2}}$.
Because $\rho^{(i)}$ commutes with $X^{\otimes k+1}$ for $i \leq k$ and
\begin{align}
X^{\otimes k+1}(\wh{\1}_{A^k} \otimes \pure{0}_C) X^{\otimes k+1} = \wh{\1}_{A^k} \otimes \pure{1}_C,
\end{align}
we can assume $\Pi^{(k+2)} = X^{\otimes k+1}\Pi^{(k+1)} X^{\otimes k+1}$.

We then compute:
\begin{align}
\F_{\pbtcl}(\Phi^{+\otimes k}) &= \F_{\pbt}(\rho)\\
&= \frac{1}{4} \hspace{-1.5em} \max_{\substack{\{\Pi^{(i)}\}_{i=1}^{k+2} \; \POVM: \\ \Pi^{(k+2)} = X^{\otimes k+1}\Pi^{(k+1)} X^{\otimes k+1}}} \hspace{-1.0em} \bigg \{ \sum_{i=1}^{k} \Tr[\sigma^{(i)} \Pi^{(i)}]
 + \Tr\big[\wh{\1}_{A^k} \otimes \pure{0}_C \Pi^{(k+1)} + \wh{\1}_{A^k} \otimes \pure{1}_C \Pi^{(k+2)} \big] \bigg \}  \\
&\leq \frac{1}{4} \max_{\substack{\{\Pi^{(i)}\}_{i=1}^{k+1} \; \POVM}} \bigg\{ \sum_{i=1}^{k} \Tr[\sigma^{(i)} \Pi^{(i)}] + \Tr[\wh{\1}_{A^k} \otimes \1_C \Pi^{(k+1)}] \bigg\} \label{eq:upper} \\
&= \frac{1}{4} \max_{\sum_{i=1}^k \Pi^{(i)} \leq \1} \bigg\{\sum_{i=1}^{k} \Tr[\sigma^{(i)} \Pi^{(i)}] + \frac{1}{2^k}\Tr[\1 - \sum_i^k \Pi^{(i)}] \bigg\}\\
&= \frac{1}{2} + \frac{1}{4} \max_{\sum_{i=1}^k \Pi^{(i)} \leq \1} \bigg\{\sum_{i=1}^{k} \Tr[(\sigma^{(i)} - \frac{\1}{2^k}) \Pi^{(i)}] \bigg\}. \label{eq:upper2}
\end{align}
Note that $\F_{\pbtq}(\Phi_k^+)$ is also bounded from above by \Cref{eq:upper}.
We claim that the maximum in \Cref{eq:upper2} is achieved with the PGM $M$.
Writing
\begin{align}
    Y = \sum_{i=1}^k \left(\sigma^{(i)}-\frac{\1}{2^k} \right)\Pi^{(i)},
\end{align}
the maximum in \Cref{eq:upper2} can be expressed in terms of the dual SDP problem with a slack variable as follows \cite[Sec~1.2.3]{Watrous2018a}:
\begin{align}
\begin{aligned}
    \text{minimize} & \q \Tr \;Z\\
    \text{subject to} & \q Z \geq 0 \\
    & \q Z \geq \sigma^{(i)}-\frac{\1}{2^k} \q i = 1,2,\dots, k.
\end{aligned}
\end{align}
We show that $Y' = \sum_{i=1}^k(\sigma^{(i)} - \frac{\1}{2^k})M^{(i)}$ is dual feasible.
It is easy to check that $Y' \geq \sigma^{(i)} - \frac{\1}{2^k}$, and to show $Y' \geq 0$, observe that
\begin{align}
    Y' = \sum_{i=1}^k \left(\sigma^{(i)} - \frac{\1}{2^k}\right)\bar{\sigma}^{-1/2}\sigma^{(i)} \bar{\sigma}^{-1/2}
    = \Big(\sum_{i=1}^k \sigma^{(i)} \bar{\sigma}^{-1/2}\sigma^{(i)} \bar{\sigma}^{-1/2}\Big) - \frac{1}{2^k}\bar{\sigma}^{-1/2}\bar{\sigma} \bar{\sigma}^{-1/2}
\end{align}
The term $\sum_i \sigma^{(i)} M^{(i)}$ can be expressed as \cite{Leditzky_2020}
\begin{align}
   \sum_{i=1}^k \sigma^{(i)} \bar{\sigma}^{-1/2}\sigma^{(i)} \bar{\sigma}^{-1/2} = \sum_{\alpha \vdash_2 k-1}\sum_{\mu \in \alpha + \square} c_{\mu,\alpha} \1_{V_{\mu}} \otimes \1_{W_{2,\alpha}}, \q\q
    c_{\mu,\alpha} \coloneqq  \frac{\sqrt{m_\mu}}{2^k m_\alpha \sqrt{d_\mu}} \sum_{\mu'\in \alpha+\square} \sqrt{d_{\mu'}m_{\mu'}},
\end{align}
and we have
\begin{align}
   Y' = \sum_{\alpha \vdash_2 k-1}\sum_{\mu \in \alpha + \square} \bigg(c_{\mu,\alpha} - \frac{1}{2^k}\bigg) \1_{V_{\mu}} \otimes \1_{W_{2,\alpha}}.
\end{align}
It remains to show that $c_{\mu,\alpha} - \frac{1}{2^k} \geq 0$ for all feasible $\mu$ and $\alpha$.
Indeed, if $\mu$ is obtained by adding a box to the first row of $\alpha$, then
\begin{equation}
    c_{\mu,\alpha} \geq \frac{m_{\mu}}{2^k m_{\alpha}} > \frac{1}{2^k}.
\end{equation}
If instead $\alpha=(p,q)$ and $\mu=(p,q+1)$, let $\mu'=(p+1,q)$.
Then $m_\mu=m_\alpha-1$ and $m_{\mu'}=m_\alpha+1$, and
\begin{equation}
    2^k c_{\mu,\alpha}=\frac{1}{m_\alpha}\Big(m_\mu+\sqrt{m_\mu m_{\mu'}d_{\mu'}/d_\mu}\Big),
\end{equation}
so $2^kc_{\mu,\alpha}\geq 1$ is equivalent to $(m_\alpha^2-1)\,d_{\mu'}/d_\mu\geq 1$.
The hook length formula gives
\begin{equation}
    (m_\alpha^2-1)\frac{d_{\mu'}}{d_\mu}=\frac{(q+1)(p-q+2)^2}{p+2}\geq\frac{(q+1)(p-q+2)}{p+2}\geq 1,
\end{equation}
where both steps use $p\geq q$.
This concludes that $Y'$ is dual-feasible.
By weak duality, the maximum in \eqref{eq:upper2} is therefore at most $\Tr Y'=\sum_i\Tr[\sigma^{(i)}M^{(i)}]-\frac{1}{2^k}\Tr\wt{M}$, where $\sum_i\Tr[\sigma^{(i)}M^{(i)}]=4\F_{\pbt}(\Phi^{+\otimes k})$ by the optimality of the PGM.

The term $\Tr\;\wt{M} = 2^{k+1} - k - 2$ is computed in \Cref{L:pgm-sum} proven below, giving the desired upper bound:
\begin{align}
  \F_{\pbtcl}(\Gamma_{\EPR{k}}) \leq \F_{\pbt}(\Phi^+_k)+ \frac{1}{4}\bigg(2 - \frac{1}{2^k}\Tr[\wt{M}]\bigg)
  = \F_{\pbt}(\Gamma_{\EPR{k}}) + \frac{k+2}{2^{k+2}}.
\end{align}

\begin{lemma} \label{L:pgm-sum}
  Let $\{M^{(i)} \}_{i=1}^k$ be the PGM of the ensemble $\{\Phi^+_{A_iC} \otimes \wh{\1}_{A_i^c} \}_{i=1}^k$ with the uniform prior distribution. Then
  \begin{align}
    \sum_{i=1}^k\Tr M^{(i)} = 2^{k+1} - k - 2.
  \end{align}
\end{lemma}
\begin{proof}
  As shown in \cite{Christandl_2018,Leditzky_2020}, Schur--Weyl decomposition gives the following expression for the sum over all PGM elements:
  \begin{align}
    \sum_{i=1}^k M^{(i)} = \bigoplus_{\alpha \vdash_2 k-1} \bigoplus_{\mu = \alpha + \square} \1_{V_\mu} \otimes \1_{W_{2,\alpha}},
  \end{align}
  and taking the trace gives
  \begin{align}
    x\coloneqq \sum_{i=1}^k \Tr M^{(i)} = \sum_{\alpha \vdash_2 k-1} \sum_{\mu = \alpha + \square} d_\mu \cdot m_{\alpha}.
  \end{align}
  In the following we parametrize Young diagrams $\alpha\vdash_2 k-1$ as $\alpha = (k-j-1,j)$, where $j=0,\dots, p\coloneqq\floor{\frac{k-1}{2}}$ indicates the length of the second row.
  In this notation,
  \begin{align}
      x &= \sum_{j=0}^p m_{(k-j-1,j)}(d_{(k-j,j)} + d_{(k-j-1,j+1)}) \label{eq:x-with-k}\\
      &= \sum_{j=0}^p \frac{k-2j}{k+1}\left[(k-2j+1)\binom{k+1}{j} + (k-2j-1)\binom{k+1}{j+1} \right],
  \end{align}
  where we used \eqref{eq:dim-irreps} in the second equality.
  Using the identities $\binom{m}{\ell}=\binom{m}{m-\ell}$ and $\ell\binom{m}{\ell} = m\binom{m-1}{\ell-1}$, this expression can be further simplified to
  \begin{align}
      x = \sum_{j=0}^p (k-2j)\left( \binom{k}{j+1} - \binom{k}{j-1}\right),\label{eq:x-simplified}
  \end{align}
  where we use the convention $\binom{k}{-1}=0$.

  To evaluate \eqref{eq:x-simplified} we use the following partial sum identities, which can be proved by noting that the sequence $\binom{k}{0},\dots,\binom{k}{k}$ is symmetric around $k/2$ and using $\sum_{j=0}^k \binom{k}{j} = 2^k$ and $\sum_{j=0}^k j \binom{k}{j} = k2^{k-1}$:
  \begin{subequations}
  \begin{enumerate}
      \item If $k$ is odd, then
      \begin{align}
          \sum_{j=0}^{(k-1)/2} \binom{k}{j} &= 2^{k-1}\\
          \sum_{j=0}^{(k-1)/2} j \binom{k}{j} &= k\left( 2^{k-2} - \frac{1}{2}\binom{k-1}{(k-1)/2}\right).
      \end{align}
      \item If $k$ is even, then
      \begin{align}
          \sum_{j=0}^{(k-2)/2} \binom{k}{j} &= \frac{1}{2}\left( 2^k - \binom{k}{k/2}\right)\\
          \sum_{j=0}^{(k-2)/2} j \binom{k}{j} &= k\left( 2^{k-2} - \binom{k-1}{(k-2)/2}\right).
      \end{align}
  \end{enumerate}
  \label{eq:partial-sum-identities}
  \end{subequations}
  
  First, let $k$ be odd so that $k-1$ is even and $p=\frac{k-1}{2}$ as defined in \eqref{eq:x-with-k}.
  We use \eqref{eq:partial-sum-identities} to compute the four sums appearing in the expression \eqref{eq:x-simplified} for $x$ (using some simple change of variables steps):
  \begin{align}
      \sum_{j=0}^p \binom{k}{j+1} &= 2^{k-1} - 1 + \binom{k}{p+1}\\
      \sum_{j=0}^p \binom{k}{j-1} &= 2^{k-1} - \binom{k}{p}\\
      \sum_{j=0}^p j\binom{k}{j+1} &= k 2^{k-2} - 2^{k-1} + 1 -\frac{k}{2}\binom{k-1}{p} + \frac{k-1}{2} \binom{k}{p+1}\\
      \sum_{j=0}^p j \binom{k}{j-1} &= k 2^{k-2} + 2^{k-1} - \frac{k}{2} \binom{k-1}{p} - \frac{k+1}{2} \binom{k}{p}.
  \end{align}
  Substituting these in \eqref{eq:x-simplified} and canceling most terms, we arrive at
  \begin{align}
      x = 2^{k+1} - 2 - k + \binom{k}{p+1} - \binom{k}{p}.
  \end{align}
  Noting that $\binom{k}{p+1} = \binom{k}{(k+1)/2} = \binom{k}{(k-1)/2} = \binom{k}{p}$ then concludes the proof of the lemma for odd $k$. 
  The case of even $k$ is handled analogously and yields the same result, proving the lemma.
\end{proof}

\section{Symmetries in $\pbtcl(\Phi^{+\otimes n})$ and SDP construction}
\label{Appendix:pbtcl-symmetries}

Let $\rho \coloneqq \Phi^{+\otimes n}_{A^nB^n} \otimes \pure{0}_{B_{n+1}} \otimes \pure{1}_{B_{n+2}}$.
The PBT states $\rho^{(i)} \coloneqq \Tr_{B_i^c} \rho$ satisfy the following symmetric properties:
  \begin{enumerate}[(i)]
    \item $(\sigma_{A^n}\otimes \1_C) \rho^{(i)} (\sigma^{-1}_{A^n} \otimes \1_C) = \rho^{\sigma(i)}$ for any $\sigma \in S_n$ and $i \leq n$; \label{it:s1}
    \item $[\pi_{A_i^c} \otimes \1_{A_iC}, \rho^{(i)}] = 0$ for any $\pi \in S_{n-1}$ and $i \leq n$; \label{it:s2}
    \item $[\sigma_{A^n} \otimes \1_C, \rho^{(j)}] = 0$ for any $\sigma \in S_n$ and $j \in \{n+1, n+2\}$; \label{it:s3}
    \item $[U^{\otimes n}_{A^n} \otimes \overline{U}_C,\rho^{(i)}]=0$ for any $U \in \cU(2)$ and $i \leq n$; \label{it:s4}
    \item $[(U_\theta^{\otimes n})_{A^n} \otimes (\overline{U}_{\theta})_{C}, \rho^{(j)}] = 0$ for any $U_\theta \coloneqq \mathrm{diag}(1,e^{i\theta})$ and $j \in \{n+1, n+2\}$; \label{it:s5}
    \item $X^{\otimes n+1} \rho^{(n+1)} X^{\otimes n+1} = \rho^{(n+2)}$ for the Pauli operator $X$. \label{it:s6}
  \end{enumerate}
By the permutation symmetries \Cref{it:s1} and \Cref{it:s2},
  \begin{align}
    \sum_{i=1}^{n+2} \Tr[\rho^{(i)} \Pi^{(i)}]
    &= \frac{1}{n!} \sum_{\sigma \in S_n} \sum_{i=1}^{n+2} \Tr[\sigma \rho^{(i)} \sigma^{-1} \sigma \Pi^{(i)} \sigma^{-1}]\\
    &= \frac{1}{n!} \sum_{\sigma \in S_n} (\sum_{i=1}^{n} \Tr[\rho^{\sigma(i)} \sigma \Pi^{(i)} \sigma^{-1}] + \sum_{j=n+1}^{n+2}\Tr[\rho^{(j)}\sigma \Pi^{(j)} \sigma^{-1}])\\
    &= \frac{1}{n!}\sum_{\sigma\in S_n} (\sum_{i=1}^n \Tr[\rho^{(i)} \sigma \Pi^{(\sigma^{-1}(i))} \sigma^{-1}] + \sum_{j=n+1}^{n+2}\Tr[\rho^{(j)}\sigma \Pi^{(j)} \sigma^{-1}]) \\
    &= \sum_{i=1}^{n+2} \Tr[\rho^{(i)} \wt{\Pi}^{(i)}],
  \end{align}
  where $\wt{\Pi}^{(i)}$'s are defined as follows.
  \begin{align}
    \wt{\Pi}^{(i)} \coloneqq \begin{cases}
        \ds \frac{1}{n!} \sum_{\sigma\in S_n} \sigma \Pi^{(\sigma^{-1}(i))} \sigma^{-1} & (i \leq n)\\
        \ds \frac{1}{n!}\sum_{\sigma\in S_n} \sigma \Pi^{(i)} \sigma^{-1} & (i > n)
    \end{cases}
  \end{align}
  It is easy to check that for any $i$ and $\pi \in S_n$ we have $\pi\wt{\Pi}^{(i)}\pi^{-1} = \wt{\Pi}^{(i)}$ and $\wt{\Pi}^{(i)} \geq 0$.
  It follows that $\wt{\Pi}^{(i)}$'s form a valid POVM since
  \begin{align}
    \sum_{i=1}^{n+2} \wt{\Pi}^{(i)}
    &= \frac{1}{n!} \sum_{\sigma} \left[\sigma\left(\sum_{i=1}^n \Pi^{(\sigma^{-1}(i))}\right) \sigma^{-1} + \sigma\left(\Pi^{(n+1)} + \Pi^{(n+2)}\right)\sigma^{-1}\right]\\
    &= \frac{1}{n!} \sum_{\sigma} \left[\sigma\left(\sum_{i=1}^n \Pi^{(i)}\right) \sigma^{-1} + \sigma\left(\Pi^{(n+1)} + \Pi^{(n+2)}\right)\sigma^{-1}\right]\\
    &= \frac{1}{n!} \sum_{\sigma} \sigma\left(\sum_{i=1}^{n+2} \Pi^{(i)}\right) \sigma^{-1} = \1.
  \end{align}
  This shows that we can assume the optimal POVM to have permutation symmetry, and similarly, symmetries on $\rho^{(i)}$'s induce the following symmetries on the POVM elements $\Pi^{(i)}$.
  \begin{enumerate}[(i)]
    \item $\sigma_{A^n} \Pi^{(i)} \sigma^{-1}_{A^n} = \Pi^{\sigma(i)}$ for any $\sigma \in S_n$ and $i \leq n$; \label{it:sym1}
    \item $[\sigma_{A^n} \otimes \1_C, \Pi^{(j)}] = 0$ for any $\sigma \in S_n$ and $j > n$;\label{it:sym2}
    \item $[\pi_{A_i^c} \otimes \1_{A_iC}, \Pi^{(i)}] = 0$ for any $\pi \in S_{n-1}$ and $i \leq n$; \label{it:sym3}
    \item $[U_\theta^{\otimes n} \otimes \overline{U}_{\theta}, \Pi^{(i)}] = 0$ for any $U_\theta \coloneqq \mathrm{diag}(1,e^{i\theta})$, $i \in [n+2]$; \label{it:sym4}
    \item $[X^{\otimes n+1}, \Pi^{(i)}] = 0$ for any $i \leq n$; \label{it:sym5}
    \item $X^{\otimes n+1} \Pi^{(n+1)} X^{\otimes n+1} = \Pi^{(n+2)}$. \label{it:sym6}
  \end{enumerate}
Note that $U_\theta^{\otimes n} \otimes U_{-\theta}$ can be expressed as
\begin{equation}
  U_\theta^{\otimes n} \otimes U_{-\theta} = \sum_{k=0}^n \Big(\sum_{s \in S_{k,n}} e^{ik\theta}\pure{s}_{A^n}\otimes \pure{0}_C + e^{i(k-1)\theta}\pure{s}_{A^n} \otimes \pure{1}_C\Big),
\end{equation}
where $S_{k,n}$ is a subset of $\{0,1\}^{n}$ defined as
\begin{align}
  S_{k,n} \coloneqq \begin{cases}
   \{ s\in \{0,1\}^n \mid |s| = k\} & \text{for $k \in \{0,1,\dots,n\}$}\\
   \varnothing & \text{otherwise}
  \end{cases}
\end{align}
with $|s|$ denoting the Hamming weight of the string $s$.
Define $D_{k,n} \coloneqq \sum_{s \in S_{k,n}} \pure{s}$ if $S_{k,n}$ is not empty and $D_{k,n} \coloneqq 0$ if $S_{k,n}$ is empty.
Then the symmetry condition \Cref{it:sym4} reduces to
\begin{gather}\label{eq:phase-sym-2}
  \Pi^{(i)} = \sum_{k,\ell = -1}^n e^{i(k-\ell)\theta} \Big(D_{k,n} \otimes \pure{0} + D_{k+1,n} \otimes \pure{1} \Big)
   \,\Pi^{(i)}\,\Big(D_{\ell,n} \otimes \pure{0} + D_{\ell+1,n} \otimes \pure{1}\Big)
\end{gather}
Since $\{D_{k,n}\}_{k=1}^n$'s are orthogonal, i.e. $D_{k,n} D_{\ell,n} = 0$ for $k\neq \ell$, in order for \Cref{eq:phase-sym-2} to hold for any $\theta$, the term with $e^{i(k-\ell)\theta}$ should be zero for $k\neq \ell$.
This means that for $s, t \in \{0,1\}^{n}$,
\begin{enumerate}[(1)]
  \item $\bra{s0}\Pi^{(i)}\ket{t0} = 0$\q if\q $|s| \neq |t|$; \label{it:1}
  \item $\bra{s1}\Pi^{(i)}\ket{t0} = 0$\q if\q $|s| \neq |t| + 1$; \label{it:2}
  \item $\bra{s0}\Pi^{(i)}\ket{t1} = 0$\q if\q $|s| + 1 \neq |t|$; \label{it:3}
  \item $\bra{s1}\Pi^{(i)}\ket{t1} = 0$\q if\q $|s| \neq |t|$. \label{it:4}
\end{enumerate}
Let $\overline{x} \in \{0,1\}^k$ denote $x \oplus 1^k$, e.g. $\overline{0100} = 1011$.
Using other symmetry conditions, we impose additional constraints on $\Pi^{(n)}$ and $\Pi^{(n+1)}$ that for any $i,j,x,y \in \{0,1\}$, $p,p',q,q' \in \{0,1\}^{n-1}$ and $s,s',t,t' \in \{0,1\}^n$,
\begin{enumerate}[(1)]
  \setcounter{enumi}{4}
  \item $\bra{sx} \Pi^{(n)} \ket{ty} = \bra{\overline{sx}}\Pi^{(n)} \ket{\overline{ty}}$; \label{it:5}
  \item (withdrawn); \label{it:6}
  \item $\bra{i}_{A_n}\bra{px}_{A_n^c C} \Pi^{(n)} \ket{j}_{A_n}\ket{qy}_{A_n^c C} = \bra{i}_{A_n}\bra{p'x}_{A_n^c C} \Pi^{(n)} \ket{j}_{A_n}\ket{q'y}_{A_n^c C}$ if $p'=\pi(p)$ and $q'=\pi(q)$ for some $\pi\in S_{n-1}$; \label{it:7}
  \item $\bra{sx}\Pi^{(n+1)}\ket{ty} = \bra{s'x}\Pi^{(n+1)}\ket{t'y}$ if $s'=\pi(s)$ and $t'=\pi(t)$ for some $\pi\in S_n$. \label{it:8}
\end{enumerate}

\noindent The positivity conditions for $\Pi^{(i)}$'s follow from the positivity conditions for $\Pi^{(n)}$ and $\Pi^{(n+1)}$.
The POVM condition adds the constraints
\begin{enumerate}[(1)]
  \setcounter{enumi}{8}
  \item $\Pi^{(n)} \geq 0$, $\Pi^{(n+1)} \geq 0$ \label{it:9}
  \item $\ds \left(\frac{n}{|S_n|}\sum_{\pi\in S_n}\bra{\pi(s)x} \Pi^{(n)} \ket{\pi(s)x} \right) + \bra{sx}\Pi^{(n+1)}\ket{sx} + \bra{\overline{sx}}\Pi^{(n+1)}\ket{\overline{sx}} = 1$ \label{it:10}
  \item $\ds \left(\frac{n}{|S_n|}\sum_{\pi \in S_n}\bra{\pi(s)x} \Pi^{(n)} \ket{\pi(t)y} \right) + \bra{sx}\Pi^{(n+1)}\ket{ty} + \bra{\overline{sx}}\Pi^{(n+1)}\ket{\overline{ty}} = 0$ if $s \neq t$ or $x \neq y$. \label{it:11}
\end{enumerate}

The fidelity to maximize is
\begin{align}
  \frac{1}{4}\sum_{i=1}^{n+2} \Tr[\Pi^{(i)} \Tr_{B_i^c}(\rho)]
  &= \frac{1}{4}(n\cdot \Tr[\Pi^{(n)} \Phi^+_{A_nC} \otimes \wh{\1}_{A_n^c}] + 2\cdot \Tr[\Pi^{(n+1)} \wh{\1}_{A^n} \otimes \pure{0}_C])\\
  &= \frac{1}{4\cdot 2^{n}}\Big(n\cdot \Tr[(\bra{00}+\bra{11})_{A_nC}\Pi^{(n)} (\ket{00}+\ket{11})_{A_nC}] + 2\cdot \Tr[\bra{0}_C \Pi^{(n+1)} \ket{0}_C]\Big)\\
  &= \frac{1}{2^{n+2}}\sum_{k=0}^n 2n\cdot \Tr[\Pi^{(n)} D_{k,n-1}^{A_n^c} \otimes \Phi^+_{A_nC} ]
  + 2 \cdot \Tr[\Pi^{(n+1)} D_{k,n}^{A^n}\otimes \pure{0}_C]
\end{align}
which is a linear function of the variables in $\Pi^{(n)}$ and $\Pi^{(n+1)}$.

Consider the case for $n = 2$.
The symmetry conditions combined with Hermiticity result in $\Pi^{(2)}$ taking the form
\begin{align}
  \Pi^{(2)} = \begin{pmatrix}
    x_1 & 0 & 0 & 0                 & 0 & \bar{x}_5 & \bar{x}_6 & 0 \\
    0 & x_2 & x_3 & 0               & 0 & 0 & 0 & x_6 \\
    0 & \bar{x}_3 & x_4 & 0         & 0 & 0 & 0 & x_5 \\
    0 & 0 & 0 & x_7                 & 0 & 0 & 0 & 0 \\
    0 & 0 & 0 & 0                   & x_7 & 0 & 0 & 0 \\
    x_5 & 0 & 0 & 0                 & 0 & x_4 & x_3 & 0 \\
    x_6 & 0 & 0 & 0                 & 0 & \bar{x}_3 & x_2 & 0 \\
    0 & \bar{x}_6 & \bar{x}_5 & 0   & 0 & 0 & 0 & x_1
  \end{pmatrix},
\end{align}
where $x_1, x_2, x_4, x_7 \in \bR$ and $x_3,x_5,x_6 \in \bC$.
Here the basis is ordered as $\ket{C\,A_2\,A_1}$.
All data of this SDP are real, so we may replace all POVM elements simultaneously by their real parts.
This keeps the constraint $\sum_i\Pi^{(i)}=\1$ and the symmetries above, so we may assume that all $x_i$, and all $a_i$ below, are real.
Likewise, by applying the condition \Cref{it:8} alongside Hermiticity to $\Pi^{(3)}$, we obtain 
\begin{align}
  \Pi^{(3)} = \begin{pmatrix}
    a_1 & 0 & 0 & 0                 & 0 & a_4 & a_4 & 0 \\
    0 & a_2 & a_3 & 0               & 0 & 0 & 0 & a_5 \\
    0 & \bar{a}_3 & a_2 & 0         & 0 & 0 & 0 & a_5 \\
    0 & 0 & 0 & a_9                 & 0 & 0 & 0 & 0 \\
    0 & 0 & 0 & 0                   & a_{10} & 0 & 0 & 0 \\
    \bar{a}_4 & 0 & 0 & 0           & 0 & a_6 & a_7 & 0 \\
    \bar{a}_4 & 0 & 0 & 0           & 0 & \bar{a}_7 & a_6 & 0 \\
    0 & \bar{a}_5 & \bar{a}_5 & 0   & 0 & 0 & 0 & a_8
  \end{pmatrix}
\end{align}
where $a_1, a_2, a_6, a_8, a_9, a_{10} \in \bR$ and $a_4,a_5 \in \bC$.
The POVM condition $\sum_i \Pi^{(i)} = \1$ corresponds to
\begin{gather}
  2x_1 + a_1 + a_8 = x_2 + x_4 + a_2 + a_6 = 2x_7 + a_9 + a_{10} = 1 \\
  2x_3 + a_3 + \bar{a}_7 = x_5 + x_6 + a_4 + \bar{a}_5 = 0
\end{gather}
It is clear that $a_{10} = x_7 = 0$ and $a_9 = 1$.
Consequently, the objective function to maximize can be expressed as
\begin{equation}\label{eq:obj}
  \frac{1}{4}\sum_{i=1}^{4}\Tr[\Tr_{B_i^c}(\rho) \Pi^{(i)}] = \frac{1}{8}(2x_1 + 2x_2 + 4x_6 + a_1 + 2a_2 + 1).
\end{equation}

By the replacement above, $\Pi^{(3)}$ is real symmetric, so all $a_i$ are real.
It matters that $\Pi^{(2)}$ and $\Pi^{(3)}$ are replaced by their real parts at the same time, because the constraint $2x_3+a_3+\bar{a}_7=0$ couples them.

Finally, we obtain a simplified SDP:
\begin{align}
  \text{maximize} \; &\frac{1}{8}(2x_1 + 2x_2 + 4x_6 + a_1 + 2a_2 + 1)\\
  \text{subject to}
  \; &\begin{pmatrix}
    x_1 & x_5 & x_6 \\ x_5 & x_4 & x_3 \\ x_6 & x_3 & x_2
  \end{pmatrix},
  \begin{pmatrix}
    a_1 & a_4 & a_4 \\ a_4 & a_6 & a_7 \\ a_4 & a_7 & a_6
  \end{pmatrix},
  \begin{pmatrix}
    a_2 & a_3 & a_5 \\ a_3 & a_2 & a_5 \\ a_5 & a_5 & a_8
  \end{pmatrix} \geq 0\\
  &2x_1 + a_1 + a_8 = x_2 + x_4 + a_2 + a_6 = 1 \\
  &2x_3 + a_3 + a_7 = x_5 + x_6 + a_4 + a_5 = 0.
\end{align}
Solving the aforementioned SDP yields the optimal value $\approx 0.6848$. We confirmed this value numerically, both for the simplified SDP above and for the full SDP over all POVMs on $A_1A_2C$ without any symmetry reduction.

\section{Computing $\F_{\pbt}(\wt{\Sigma}^{k,n})$ (Proof of \Cref{P:k-epr-sym})}

\label{sec:A-lemmas}

\begin{proposition*}[Restatement of Prop.~\ref{P:k-epr-sym}]
Let $k, n$ be positive integers such that $2k-1\leq n$. Then
\begin{align}\label{eq:sym-fidelity}
\begin{aligned}
  \F_{\pbt}(\wt{\Sigma}^{k,n})
  &= (2^{k+2}(n-k+1))^{-1}\Bigg[(n-2k+1) \delta_{n-1 \geq 2k}\\
  &\qquad + \sum_{i=1}^{k} \delta_{n+1\geq 2i} {\left(\sqrt{ (n-2i+3) \binom{k}{i-1}} + \sqrt{(n-2i+1) \cdot \binom{k}{i}} \right)}^2 \Bigg]
\end{aligned}
\end{align}
where $\delta$ is the indicator function.
Moreover, for a fixed $k$, we have
\begin{equation}\label{eq:sym-fixed-k-app}
  \lim_{n\to \infty} \F_{\pbt}(\wt{\Sigma}^{k,n}) = 2^{-k-2}\Bigg(1 + \sum_{i=1}^{k} \Bigg(\sqrt{\binom{k}{i-1}} + \sqrt{\binom{k}{i}} \Bigg)^2\Bigg).
\end{equation}
\end{proposition*}

\begin{proof}

Given a partition $\lambda\vdash_2 n$, define $x_{\lambda} \coloneqq \Tr[\wt{\Sigma}^{k,n} P_\lambda]$ if $\lambda$ is a valid Young diagram and $x_{\lambda} = 0$ otherwise, where $P_\lambda$ is the isotypic projection onto the $S_n$-irrep $V_\lambda$.
Using the partial trace formula for the Young projector in \eqref{eq:sym-partial-tr} iteratively yields the following result:
\begin{equation}
  x_{\lambda}
  = \frac{1}{2^k} \bra{+}^{\otimes n-k} \Tr_{B^k}P_\lambda \ket{+}^{\otimes n-k}
  = \frac{1}{2^k} \sum_{\beta \in \lambda - k\cdot \square} \frac{m_\lambda}{m_{\beta}} \bra{+}^{\otimes n-k} P_{\beta} \ket{+}^{\otimes n-k} \label{eq:x-lambda}.
\end{equation}
The summation in \eqref{eq:x-lambda} is taken over all Young diagrams $\beta$ that can be obtained by iteratively removing $k$ boxes, while ensuring that each removal maintains a valid Young diagram.
For any $\beta \neq (n-k,0)$, the inner product appearing in the sum is 0, and for $\beta = (n-k,0)$
\begin{align}
x_\lambda
  = \frac{m_\lambda \cdot t_{(n-k,0) \to \lambda}}{2^k m_{(n-k,0)}}
  = \frac{m_\lambda \cdot t_{(n-k,0) \to \lambda}}{2^k \cdot (n-k+1)}
\end{align}
where $t_{\mu\to \lambda}$ denotes the number of paths from $\mu$ to $\lambda$ in the Bratteli diagram, whose vertices are Young diagrams $\nu \vdash_2 m$ for $m = 1, 2, \dots, n$, and $\alpha \to \beta$ is an edge if $\beta \in \alpha + \square$.
For $i \leq k$ the number $t_{(n-k)\to (n-i,i)}$ is given by the recursive expression
\begin{align}
t_{(n-k)\to (n-i,i)} = t_{(n-k) \to (n-i-1,i)} + t_{(n-k) \to (n-i,i-1)},
\end{align}
and solving this leads to $t_{(n-k)\to (n-i,i)} = \binom{k}{i}$, and $t_{(n-k,0)\to (n-i,i)} = 0$ if $i > k$.
The first count is valid because no path leaves the Young diagrams, as the second row never exceeds $i$, the first row is at least $n-k$, and $i\leq n-k$ for every $\lambda=(n-i,i)$ that occurs when $2k-1\leq n$.
Alternatively, we can think of $t_{(n-k) \to (n-i,i)}$ as choosing $i$ number of steps out of $k$ steps to add a box to the second row in the Young diagram.

\begin{align}
\F_{\pbt}&(\wt{\Sigma}^{k,n})
    = \frac{1}{4} \sum_{i=0}^{k} \delta_{n-i-1\geq i}\cdot \Bigg(\sum_{\lambda \in (n-1-i,i) + \square} \sqrt{x_\lambda}\Bigg)^2\\
    &= \frac{1}{4} \Bigg( \delta_{n-k-1\geq k} \cdot x_{(n-k,k)} + \sum_{i=0}^{k-1} \delta_{n-i-1\geq i} \cdot \Bigg(\sqrt{x_{(n-i,i)}} + \sqrt{x_{(n-i-1,i+1)}}\Bigg)^2\Bigg)\\
  &= \frac{1}{2^{k+2}(n-k+1)} \Bigg[ \delta_{n-1 \geq 2k} \cdot \,m_{(n-k,k)} t_{(n-k)\to (n-k,k)} +  \sum_{i=0}^{k-1} \delta_{n-1\geq 2i} \Bigg(\sqrt{ m_{(n-i,i)} \cdot t_{(n-k)\to (n-i,i)}}  \\
  & \qquad {} + \sqrt{m_{(n-i-1,i+1)} \cdot t_{(n-k) \to (n-i-1,i+1)}}\Bigg)^2   \Bigg]  \\
  &= \frac{1}{2^{k+2}(n-k+1)} \Bigg(\delta_{n-1 \geq 2k} \cdot m_{(n-k,k)} + \sum_{i=0}^{k-1} \delta_{n-1\geq 2i}\Bigg(\sqrt{ m_{(n-i,i)} \cdot \binom{k}{i}} + \sqrt{m_{(n-i-1,i+1)} \cdot \binom{k}{i+1}} \Bigg)^2\Bigg).
\end{align}
Recall that $m_\lambda$ is the number of semi-standard tableaux of shape $\lambda = (\lambda_1,\lambda_2)$ with each entry being $1$ or $2$.
If $T$ is such semi-standard tableau, first $\lambda_2$ columns must have all 1's in the first row and all 2's in the second row.
This means
\begin{align}
  m_{(\lambda_1,\lambda_2)} = m_{(\lambda_1-\lambda_2)} = \lambda_1-\lambda_2 + 1,
\end{align}
and plugging these values yields \eqref{eq:sym-fidelity-app} and \eqref{eq:sym-fixed-k-app}.
\end{proof}

\section{Analytical results for $\F_{\pbtcl}(\Phi^+_{AB})$} \label{sec:A:1-epr-Fpbt}
The PBT ensemble for $\pbtcl$ with resource state $\Phi^+_{AB}$ is composed of the following states:
\begin{align}
  \rho^{(1)} = \Phi^+_{AB} \q\q
  \rho^{(2)} = \wh{\1}_A \otimes \pure{0}_B \q\q
  \rho^{(3)} = \wh{\1}_A \otimes \pure{1}_B.
\end{align}
After relabeling the basis elements, they have the form
\begin{align}
  \begin{aligned}
\rho^{(1)} &= \begin{pmatrix} P & 0 \\ 0 & 0 \end{pmatrix}  &
  \rho^{(2)} &= \begin{pmatrix} Q & 0 \\ 0 & R \end{pmatrix}  &
  \rho^{(3)} &= \begin{pmatrix} R & 0 \\ 0 & Q \end{pmatrix},
  \end{aligned}
 \label{eq:block}
\end{align}
where
\begin{equation}
  P = \begin{pmatrix}\frac{1}{2} & \frac{1}{2} \\[.5em] \frac{1}{2} & \frac{1}{2} \end{pmatrix} \q
  Q = \begin{pmatrix} \frac{1}{2} & 0 \\[.5em] 0 & 0 \end{pmatrix} \q
  R = \begin{pmatrix} 0 & 0 \\[.5em] 0 & \frac{1}{2} \end{pmatrix}.
\end{equation}
Our objective is to find a POVM $\Pi = \{\Pi^{(1)}, \Pi^{(2)}, \Pi^{(3)}\}$ that maximizes $\sum_i \Tr[\rho^{(i)} \Pi^{(i)}]$.
We first observe that it is sufficient to consider the POVMs whose elements are of block diagonal matrix form in the basis used in \eqref{eq:block}.
If each $\Pi^{(i)}$ is of the form
\begin{align}
  \Pi^{(i)} = \begin{pmatrix}
    A_i & C_i \\ C_i^\dagger & B_i
  \end{pmatrix}
\end{align}
the objective function yields
\begin{align}
  \sum_i \Tr[\rho^{(i)} \Pi^{(i)}] = \Tr[PA_1 + QA_2 + RB_2 + RA_3 + QB_3],
\end{align}
which has no dependence on $C_i$'s.
If $\Pi$ achieves the maximum where $C_i \neq 0$ for some $i$, then the operators $\wt{\Pi}^{(i)}$ given by
\begin{align}
  \wt{\Pi}^{(i)} = \begin{pmatrix}
    A_i & 0 \\ 0 & B_i
  \end{pmatrix}
\end{align}
forms a POVM since $A_i, B_i \geq 0$ for all $i$, and $\{\wt{\Pi}^{(i)}\}_i$ also achieves the maximum.
With a similar argument, we can further assume that the maximizing $\Pi^{(i)}$'s are of the form
\begin{align}
  \Pi^{(1)} = \begin{pmatrix} N_1 & 0 \\ 0 & 0 \end{pmatrix} \q
  \Pi^{(2)} = \begin{pmatrix} N_2 & 0 \\ 0 & \pure{1} \end{pmatrix} \q
  \Pi^{(3)} = \begin{pmatrix} N_3 & 0 \\ 0 & \pure{0} \end{pmatrix}.
\end{align}
Then the fidelity of $\pbtcl$ for $\Phi^+$ reduces to
\begin{align}
  \F_{\pbtcl}(\Phi^+) = \frac{1}{4} \left(\max_{N\; \POVM} \Tr [PN_1 + QN_2 + RN_3] + 1\right),
\end{align}
where the maximization is over all POVMs $\{N_i\}_{i=1}^3$ on $\bC^2$.
Let
\begin{align}
  N_i = \begin{pmatrix} a_i & b_i \\ \bar{b}_i & c_i \end{pmatrix}
\end{align}
By symmetry, we can assume $a_1 = c_1$, $a_2 = c_3$ and $a_3 = c_2$.
Then
\begin{align}
  \Tr[PN_1 + QN_2 + RN_3] = a_1 + \text{Re}(b_1) + a_2
\end{align}
with positivity constraints
\begin{align}
  a_1,a_2,c_2 \geq 0 \q\q a_1^2 - b_1 \bar{b}_1 \geq 0 \q\q a_2c_2 - b_2\bar{b}_2 \geq 0
\end{align}
and the constraints for $\sum_i N_i = \1$ are given by
\begin{align}
  a_1 + a_2 + c_2 = 1 \q\q b_1 + b_2 + b_3 = 0.
\end{align}
Solving the system of equations yields the optimal POVM
\begin{align}
  N_1 = \begin{pmatrix}
    \frac{4}{9} & \frac{4}{9}\\
    \frac{4}{9} & \frac{4}{9}
  \end{pmatrix}\q
  N_2 = \begin{pmatrix}
    \frac{4}{9} & -\frac{2}{9}\\
    -\frac{2}{9} & \frac{1}{9}
  \end{pmatrix}\q
  N_3 = \begin{pmatrix}
    \frac{1}{9} & -\frac{2}{9}\\
    -\frac{2}{9} & \frac{4}{9}
  \end{pmatrix}
\end{align}
and we have $\F_{\pbtcl}(\Phi^+) = \frac{7}{12} \approx 0.5833$. This POVM is optimal, because $Y=\frac{1}{3}\left(\begin{smallmatrix}2&1\\1&2\end{smallmatrix}\right)$ satisfies $Y\geq P$, $Y\geq Q$ and $Y\geq R$ with $\tr Y=\frac{4}{3}$, so weak duality bounds the maximum by $\frac{4}{3}$.

\end{document}